\documentclass[journal]{IEEEtran}

\usepackage{cite,url, amsmath, amssymb, array, balance}
\usepackage{algorithm, algorithmic, calc, color, multicol, multirow}
\usepackage{graphics, epsfig, psfrag, subfigure, graphicx, epstopdf}
\usepackage{amsmath, amsfonts, stmaryrd}
\usepackage{pstricks, pst-node, pst-tree}
\usepackage{verbatim}

\newtheorem{theorem}{Theorem}[section]
\newtheorem{lemma}[theorem]{Lemma}

\newtheorem{corollary}[theorem]{Corollary}
\newtheorem{definition}{Definition}
\newtheorem{example}{Example}

\newcommand{\ie}{{\em i.e., }}
\newcommand{\eg}{{\em e.g., }}

\begin{document}
\IEEEoverridecommandlockouts

\title{A Network Coding Approach to Loss Tomography}

\author{Pegah~Sattari,~\IEEEmembership{Student Member,~IEEE,}
        Athina~Markopoulou,~\IEEEmembership{Senior~Member,~IEEE,}
        Christina~Fragouli,~\IEEEmembership{Member,~IEEE,}
        and~Minas~Gjoka
\thanks{Pegah Sattari, Athina Markopoulou, and Mina Gjoka are with the Department
of Electrical Engineering and Computer Science, University of California, Irvine, CA,
92697 USA (e-mail: psattari, athina, mgjoka@uci.edu).}
\thanks{Christina Fragouli is with the School of Computer and Communication Sciences, EPFL, Lausanne, Switzerland (e-mail: christina.fragouli@epfl.ch).}
\thanks{This work has been supported by the following awards: NSF CAREER 0747110, AFOSR  MURI FA9550-09-0643, AFOSR FA9550-10-1-030; FNS award no PP00P2-128639 and ERC-2009-StG-240317.}
}

\maketitle

\begin{abstract}
Network tomography aims at inferring internal network characteristics based on measurements at the edge of the network. In loss tomography, in particular, the characteristic of interest is the loss rate of individual links and multicast and/or unicast end-to-end probes are typically used. Independently, recent advances in network coding have shown that there are advantages from allowing intermediate nodes to process and combine, in addition to just forward, packets. In this paper, we study the problem of loss tomography in networks with network coding capabilities. We design a framework for estimating link loss rates, which leverages network coding capabilities, and we show that it improves several aspects of tomography including the identifiability of links, the trade-off between estimation accuracy and bandwidth efficiency, and the complexity of probe path selection. We discuss the cases of inferring link loss rates  in a tree topology and in a general topology. In the latter case, the benefits of our approach are even more pronounced compared to standard techniques, but we also face novel challenges, such as dealing with cycles and multiple paths between sources and receivers. Overall, this work makes the connection between active network tomography and network coding.
\end{abstract}

\begin{IEEEkeywords}
Network Coding, Network Tomography, Link Loss Inference.
\end{IEEEkeywords}

\section{Introduction}

\IEEEPARstart{D}{istributed} Internet applications often need to know information about the
characteristics of the network. For example, an overlay or peer-to-peer network may want to detect and recover  from failures or degraded performance of the underlying Internet infrastructure. A company with several geographically distributed campuses  may want to know the behavior
of one or several Internet service providers (ISPs) connecting the campuses, in order
 to optimize traffic engineering decisions and achieve the best end-to-end performance.
 To achieve this high-level goal, it is necessary for the nodes participating in the application or overlay to monitor Internet paths, assess and predict their behavior, and
eventually make efficient use of them by taking appropriate  control and traffic engineering decisions
both at the network and at the application layers. Therefore, accurate monitoring at minimum overhead and complexity is of crucial importance in order to provide the input needed to take such informed decisions. However, there is currently no incentive for ISPs to provide detailed information about their internal operation and performance or to collaborate with other ISPs for this purpose. As a result, distributed applications usually rely on their own end-to-end measurements between nodes they have control over, in order to infer performance characteristics of the network.

Over the past decade,  a significant research effort has been devoted to a class of monitoring
problems that aim at inferring internal network characteristics using measurements at the edge \cite{tomography-survey}. This class of problems is commonly referred to as {\em tomography} due to its analogy to medical tomography. In this work, we are particularly interested in loss tomography, {\em i.e.,} inferring the loss probabilities (or loss rates) of individual links using active end-to-end measurements \cite{minc,general,nowak,tomo-unicast-1,tomo-unicast-2}. The topology is assumed known and sequences of probes are sent and collected between a set of sources and a set of receivers at the network edge. Link-level parameters, in this case loss rates
of links, are then inferred by the observations at the receivers. The bandwidth efficiency of these methods can be measured by the number of probes needed to estimate the loss rates of interest within a desired accuracy. Despite its significance and the research effort invested, loss tomography remains a hard problem for a number of reasons, including complexity (of optimal probe routing and of estimation), bandwidth overhead, and identifiability (the fundamental fact that tomography is an inverse problem and we cannot directly observe the parameters of interest). Moreover, there are some practical limitations such as the lack of cooperation of ISPs, the need for synchronization of sources in some schemes, etc.

Recently, a new paradigm to routing information has emerged with the advent of network coding \cite{nc1,nc2,netcodingwebpage}. The main idea in network coding is that, if we allow intermediate nodes to not only forward but also combine packets, we can obtain significant benefits in terms of throughput, delay and robustness of distributed algorithms. Our work is based on the observation that, in networks equipped with network coding capabilities, we can leverage these capabilities to significantly improve several aspects of loss tomography. For example, with network coding, we can combine probes from different paths into one, thus reducing the bandwidth needed to cover a general graph and also increasing the information per packet. Furthermore, the problem of optimal probe routing, which is known to be NP-hard, can be solved with linear complexity when network coding is used.

This paper proposes a framework for loss tomography (including mechanisms for probe routing, probe and code design, estimation, and identifiability guarantees) in networks that already have network coding capabilities. Such capabilities do not exist yet on the Internet today, but are available in wireless mesh networks, peer-to-peer and overlay networks and we expect them to appear in more environments as network coding becomes more widely adopted. We show that, in those settings, our network coding-based approach improves the following aspects of the loss tomography problem:  how many links of the network we can infer (identifiability); the tradeoff between how well we can infer link loss rates (estimation accuracy) and how many probes we need in order to do so (bandwidth efficiency); how to select sources and receivers and how to route probes between them (optimal probe routing). Overall, this is a novel application of network coding techniques to a practical networking problem, and it opens a promising research direction.

The structure of the paper is as follows. Section \ref{sec:related} discusses related work. Section \ref{sec:statement} states the problem and summarizes the challenges and main results.
Section \ref{sec:singlelink} presents a motivating example and provides the conditions of identifiability. Sections \ref{sec:trees} and \ref{sec:general} present in detail the framework and mechanisms in the cases of trees and general topologies, respectively. Section \ref{sec:conclusion} concludes the paper.

\section{\label{sec:related}Related Work}

{\bf Network Tomography.} The term network tomography typically refers to a family of problems that aim at inferring internal network characteristics from measurements at the edge of the network. Internal characteristics of interest may include link-level parameters (such as loss and delay metrics) or the network topology. Another type of tomography problem aims at inferring path-level traffic intensity (\eg traffic matrices) from link-level measurements \cite{vardi}. Our paper focuses on inferring the loss rates of internal links using active end-to-end measurements and assuming that the topology is known. Therefore, it is related to the literature on loss tomography, part of which is discussed below.

Caceres {\em et al.} considered a single multicast tree with a known topology and inferred the link loss rates from the receivers' observations \cite{minc}. In particular, they developed a low-complexity algorithm to compute the maximum likelihood estimator (MLE), by taking into account the dependencies introduced by the tree hierarchy to factorize the  likelihood function and eventually compute the MLE in a recursive way. Throughout this paper, we refer to the MLE for a multicast tree, developed in \cite{minc}, as MINC, and we build on it. Bu {\em et al.} used multiple multicast trees to cover a general topology and proposed an EM algorithm for link loss rate estimation \cite{general}. Follow-up approaches have been developed for unicast probes \cite{tomo-unicast-1, tomo-unicast-2}, joint inference of topology and link loss rates \cite{nowak}, and adaptive tomography and delay inference \cite{tomo-delay}. The above list of references is not comprehensive. Good surveys of network tomography can be found in \cite{tomography-survey, gmichael-survey}.

{\bf Active vs. Passive Tomography.}
Tomography can be based either on active (generating probe traffic) or on passive (monitoring traffic flows and sampling existing traffic) measurements.  Passive approaches have been most commonly used for estimating path-level information, in particular, origin-destination traffic matrices, from data collected at various nodes of the network \cite{vardi}. This approach and problem statement are well-suited for the needs of a network provider. For the problem of inferring link loss rates, active probes are typically used, and information about individual packets received or lost  is analyzed at the edge of the network.  This approach is better suited for end users that do not have access to the network. However, there are also papers that study link loss inference by using existing traffic flows to sample the state of the network \cite{passive1, passive2}. Once measurements have been collected following either of the two methods, statistical inference techniques are applied to determine network characteristics that are not directly observed.

The passive approach has the advantage that it does not impose additional burden on the network and that it measures the actual loss experienced by real traffic. However, it must also ensure that the characteristics of the traffic ({\em e.g.,} TCP) do not bias the sample. In the active approach, one has more control over designing the probes, which can thus be optimized for efficient estimation. The downside is that we inject measurement traffic that may increase the load of the network, may be treated differently than regular traffic, or may even be dropped {\em e.g.,} due to security concerns.

{\bf Network Coding and Inference.}
An extensive body of work on network coding \cite{netcodingwebpage, monograph} has emerged after the seminal work of Ahlswede {\em et al.} \cite{nc1} and Li {\em et al.} \cite{nc2}. The main idea in network coding is that, if we allow intermediate nodes to not only forward but also combine packets, we can realize significant benefits in terms of throughput, delay, and robustness of distributed algorithms. Within this large body of work, closer to ours are a few papers
that leverage the headers of network coded packets for passive inference of properties of a network. In \cite{ho-monitoring}, Ho {\em et al.} showed how information contained in network codes can be used for passive inference of possible locations of link failures or losses.
In \cite{jaggi}, Sharma {\em et al.} considered random intra-session network coding and showed that nodes can passively infer their upstream network topology, based on the headers of the received coded packets they observe (which play essentially the role of probes). The main idea is that the transfer matrix ({\em i.e.,} the linear transform from the sender to the receiver) is distinct for different networks, with high probability. All possible transfer matrices are enumerated, and matched to the observed input/output, and a large finite field is used to ensure that all topologies remain distinguishable. An extended version of this work to erroneous networks is provided by Yao {\em et al.} in \cite{jaggi-arxiv}, where different (ergodic or adversarial) failures lead to different transfer functions. The approach in \cite{jaggi, jaggi-arxiv} has the advantage of keeping the measurement bandwidth low (not higher than the transmission of coefficients, which is anyway required for data transfer with network coding) and the disadvantage of high complexity. In \cite{mahti}, Jafarisiavoshani {\em et al.} considered peer-to-peer systems and used subspace nesting structures to passively identify local bottlenecks. Similar to these papers, we leverage network coding operations for inference; in contrast to these papers, which use the headers of network-coded packets for passive inference of topology, we use the contents of active probes for inference of link loss rates.

{\bf Our Work.} We make the connection between active network tomography and network coding capabilities. In \cite{allerton05}, we introduced the basic idea of leveraging network coding capabilities to improve network monitoring. In \cite{ita07}, we studied link loss estimation in tree topologies. In \cite{globecom08}, we extended the approach to general graphs. In \cite{netcod2011}, we built on MINC \cite{minc}, and we provided the MLEs of the loss rates for all links simultaneously, in multiple-source tree topologies with multicast and network coding; similarly to MINC, we presented an efficient algorithm for computing the MLEs, we proved the correctness, and we analyzed the rate of convergence. This paper combines ideas from these preliminary conference papers into a common framework, and extends them by a more in-depth analysis of identifiability, routing, estimation and code design.

Our approach is active in that probes are sent/received from/to the edge of the network and observations at the receivers are used for statistical inference. Intermediate nodes forward packets using unicast, multicast and simple coding operations. However, the operations at the intermediate nodes need to be set-up once, fixed for all experiments, and be known for inference. Therefore, our approach requires more support from the network than traditional tomography, for the benefit of more accurate/efficient estimation. Our methods may also be applicable to passive tomography, where instead of sending specialized probes, one can view the coding coefficients on a network coded packet as the ``probe'', thus overloading them with both communication and tomographic goals, as it is the case in \cite{jaggi, jaggi-arxiv}. In this paper, we focus exclusively on the tomographic goals by taking an active approach, {\em i.e.,} sending, collecting, and analyzing specialized probes for tomography.

\section{Problem Statement}\label{sec:statement}

\subsection{Model and Definitions} \label{sec:model-definitions}

\subsubsection{Network and Monitoring Scheme}
We consider a network represented as a graph $G=(V,E)$, where  $V$ is the set of nodes and $E$ is
the set of edges corresponding to {\em logical links}\footnote{A logical link results
from combining several consecutive physical links into a single link. This results in a graph $G$ where every intermediate vertex has degree at least three, and in-degree and out-degree at least one. This is a standard assumption in the tomography literature, which is imposed for identifiability purposes, as discussed after Definition \ref{def_ident}.}. We use the notation $e=AB$ for the link $e$ connecting vertex $A$ to vertex $B$. We assume that $G$ has no self-loops and that there is a loss rate associated with every edge in $G$.\footnote{In general, the loss rates in the two directions of an edge can be different, as it is the case on the Internet due to different congestion levels.} The topology $G=(V,E)$ is assumed to be known.

We assume that packet loss on a link $e\in E$ is i.i.d Bernoulli with probability $0\leq \overline \alpha_e < 1$, where $\overline \alpha_e=1-\alpha_e$, and $\alpha_e$ is the success probability of link $e$. Losses are assumed to be independent across links. Let $\alpha=(\alpha_e)_{e \in E}$ be the vector of the link success probabilities\footnote{Note that the notation $\alpha$ refers to the vector of all success probabilities, and $\alpha_e$ refers to the success prob. of an individual edge $e$.}. In loss tomography, we are interested in estimating all or a subset of the parameters in $\alpha$. We use additional notation for the case of tree topologies, as we explain in Section~\ref{sec-model}.

A set $S$ of $|S|=M$ source nodes in the periphery of the network can inject probe packets, while a set $R$ of $|R|=N$ receivers can collect such packets. Several problem variations in the choice of sources and receivers are possible, and we will discuss the following  in this paper:  (i) the set of sources and the set of receivers are given and fixed; (ii) a set of nodes that can act as either sources or receivers  is given (and we can select among them); (iii) we are allowed to select any node to act as a source or a receiver. We assume that intermediate nodes are equipped with unicast, multicast and network coding capabilities. Probe packets are routed and coded inside the network following specific paths and according to specified coding operations. We assume that the packets incur zero transmission, propagation and processing delay as they travel through the network. The routes selected and the operations the intermediate nodes perform are part of the design of the tomography scheme: they are chosen once at set-up time and are kept the same throughout all experiments; all operations of intermediate nodes are known during estimation. For the theoretical results of this paper, we focus on {\em synchronized acyclic networks with zero delay}\footnote{Note that the link delays will only affect where the probe packets would meet in the network; they will not affect our general model.}; for cyclic networks, we convert them to acyclic networks by a proper choice of routing and sources/receivers.

In general, a probe packet is a vector of $M$ symbols, with each symbol being in a finite field $F_q$. This includes as special cases: scalar network coding (for $M=1$), operations over binary vectors (for $q=2$), and more generally, vector network coding (for $M>1$)\footnote{What is important is that a probe can take one of the $q^M$ possible values. We note, however, that there is an equivalence between operations with elements in a finite field and operations with vectors of appropriate length. {\em E.g.,} in \cite{vnc}, the multicast scenario was considered, and scalar network coding over a finite field of size $2^M$ was used equivalently to vector network coding over the space of binary vectors of length $M$. Thinking in terms of one of the aforementioned special cases is appropriate in special topologies, as we will see, {\em e.g.,} in tree and reverse tree topologies, where scalars and binary vectors are used, respectively.}.
In one {\em experiment}, we send probes from all sources and we collect probes at the receivers: each source $S_i\in S$ injects one probe packet $x_i$ in the network, and each receiver $R_j\in R$ receives one probe $X_j$. The observations at all receivers $R$ is a vector $X_{(R)}=(X_1, X_2, ... X_N)$ in the space $\Omega \subseteq (F_{q^M})^N$.
For a given set of link success probabilities $\alpha=({\alpha_e})_{e\in E}$, the probability distribution of all observations $X_{(R)}$ will be denoted by $P_{\alpha}$. The probability mass function for a single observation $x\in \Omega$ is $p(x;\alpha)=P_{\alpha}(X_{(R)}=x)$.

To estimate the success rates of links, we perform a sequence of $n$ independent experiments. Let $n(x)$ denote the number of probes for which the observation $x\in \Omega$ is obtained, where $\sum_{x\in \Omega} n(x)=n$. The probability of $n$ independent observations $x^1, \cdots , x^n$ (each $x^t=(x^t_k)_{k\in R}$) is:
\begin{equation}
\label{eq-probability}
p(x^1, \cdots , x^n;\alpha)=\prod_{t=1}^{n} p(x^t;\alpha)=\prod_{x\in \Omega} p(x;\alpha)^{n(x)}
\end{equation}
It is convenient to work with the log-likelihood function, which calculates the logarithm of this probability:
\begin{equation}
\label{eq-likelihoodSum1}
\mathcal{L}(\alpha) = \log p(x^1, \cdots , x^n;\alpha) = \sum_{x\in \Omega} n(x) \log p(x;\alpha)
\end{equation}
We make two assumptions, which are both realistic in practice and standard in the tomography literature:
\begin{itemize}
\item We perform sufficient measurements so that each observation $x\in \Omega$ at the receivers occurs at least once, {\em i.e.,} $n(x)> 0$. This ensures that no term in the likelihood function becomes a constant (due to a zero exponent). Note that the final equality in Eq.(\ref{eq-probability}) and Eq.(\ref{eq-likelihoodSum1}) is valid due to this assumption.
\item The probability of loss $\overline \alpha_i$ on a link $i$ is not 1, {\em i.e.,} $\overline \alpha_i\in [0,1)$. This ensures that the log-likelihood function is well-defined and differentiable.
\end{itemize}

The goal is to use the observations at the receivers, the knowledge of the network topology, and the knowledge of the routing/coding scheme to estimate the success rates of internal links of interest. We may be interested in estimating the success rate on a subset of links, or on all the links.

\begin{definition}\label{def_scheme}
A {\em monitoring scheme} for a given graph $G$ refers to a set of $M$ source nodes, a set of $N$ receivers, a set of paths that connect the sources to the receivers,
the probe packets that sources send, and the operations that intermediate nodes perform on these packets.
\end{definition}

We use the notion of link identifiability as it was defined in \cite{minc} (Theorem 3, Condition (i)):
\begin{definition}\label{def_ident}
A link $e$ is called {\em identifiable} under a given monitoring scheme iff: $\alpha, \alpha' \in (0,1]^{|E|}$ and $P_\alpha=P_{\alpha'}$ implies $\alpha_e=\alpha'_e$.
\end{definition}

To illustrate the concept, consider two consecutive links $e_1=AB$ and $e_2=BC$ in a row, where node $B$ has degree 2, and is neither a source nor a receiver. These links are not identifiable, as maximizing the log-likelihood function would only allow us to identify the value of the product $\alpha_{e_1}\alpha_{e_2}$, and thus, would lead to an infinite number of solutions. This is because, it is not possible to distinguish whether a packet gets dropped on link $e_1$ or $e_2$. Note, however, that the case of having two links in a row is ruled out by our assumption of working on a graph with logical links (all vertices in the graph have degree three or greater). Another case that $e_1, e_2$ are not identifiable, which is possible even on a graph with logical links, is when both links belong to every path used from any source to any receiver.

Identifiability is not only a property of the network topology, but also depends on the
monitoring scheme. One of the main goals of the monitoring scheme design is to maximize the number of identifiable links. However, our definition of identifiability does not depend on the estimator employed. Essentially, identifiability depends on the probability distribution $P_\alpha$ and on whether this uniquely determines $\alpha$.

\subsubsection{Estimation}

The maximum likelihood estimator (MLE) $\breve{\alpha}$ identifies the parameters $(\alpha_e)_{e\in E}$ that maximize the probability of the observations $\mathcal{L}(\alpha)$:
\begin{equation}\label{eq_ML}
\breve{\alpha}= \mbox{argmax}_{\alpha \in (0,1]^{|E|}} \mathcal{L}(\alpha)
\end{equation}
Candidates for the MLE are the solutions $\hat{\alpha}$ of the {\em likelihood equation}:
\begin{equation}
\frac{\partial \mathcal{L}}{\partial \alpha_e}(\alpha)=0, \quad e\in E
\end{equation}
We can compute the MLE for tree networks as we see in Section~\ref{sec-estimation1}. However, it becomes computationally hard for large networks; this creates the need for faster algorithms that provide good approximate performance in practice.

To measure the per link estimation accuracy, we use the
mean-squared error (MSE): $\mbox{MSE}=E(|\alpha_e-\hat{\alpha}_e|^2)$.
In order to measure the estimation performance on all links $e \in E$, we need a metric that summarizes all links. We use an entropy measure $ENT$ that captures the residual
uncertainty. Since we expect the scaled estimation errors to be
asymptotically Gaussian (similar to the case in \cite{minc}), we
define the quality of the estimation across all links as
\begin{equation}
\label{eq:Ent} 
\mathrm{ENT} = \sum_{e\in E}\log\left (
E[\hat{\alpha}_e - \alpha_e]^2\right ),
\end{equation}
which is a shifted version of the entropy of independent Gaussian
random variables with the given variances \cite{CoverThomas91}. If
the entire error covariance matrix $\mathcal{R}$ is available, then we can
compute the metric as $ENT = \log \mathrm{det}\mathcal{R}$, which captures
also the correlations among the errors on different links. The
metric $ENT$ defined above captures only the diagonal elements
of $\mathcal{R}$, {\em i.e.}, the $MSE$ for each link independently of the others.

In some cases, we approximate the error covariance matrix $\mathcal{R}$ using the Fisher information matrix $\mathcal{I}$. Under mild regularity conditions (see for example Chapter 7 in \cite{Lehmann99}), the scaled asymptotic covariance matrix of the optimal estimator is lower-bounded by the Cramer-Rao bound $\mathcal{I}^{-1}$. The Fisher information matrix $\mathcal{I}$ is a square matrix with element $\mathcal{I}_{p,q}$ defined as
\begin{equation}
\label{eq:FIM}
\begin{aligned}
&\mathcal{I}_{p,q}(\alpha) =
-E\left [
\frac{\partial}{\partial \alpha_p} \log p(X_{(R)};\alpha)
\frac{\partial}{\partial \alpha_q} \log p(X_{(R)};\alpha) \right ]
\end{aligned}
\end{equation}
where $\alpha_p,\alpha_q$ are the success probabilities of two links. In particular, under the regularity conditions, the MLE is asymptotically efficient; {\em i.e.,} it asymptotically, in sample size achieves this lower bound.

\subsection{\label{sec:decomposition}Subproblems}
Given a certain network topology, a monitoring scheme for loss tomography can be designed by solving  the following subproblems.

1) {\bf Identifiability:} For each link $e \in E$, derive conditions that the scheme should satisfy so that the edge is identifiable. Whether the goal is to maximize the number of identifiable edges, or to measure the link success rate on a particular set of edges, the identifiability conditions will guide the routing and  code design choices.

2) {\bf Routing:} Select the sources and receivers of probe packets, the paths through which probes are routed, and the nodes where they will be linearly combined.\footnote{Depending on the practical constraints, such flexibility may or may not be available. If one cannot choose the source/receiver nodes and/or routing, as it is the case in most of the tomography literature, then this step can be skipped. If one can choose some of these parameters, then this can lead to further optimization of identifiability and estimation accuracy.}
The design goals include minimizing the utilized bandwidth, and improving the estimation accuracy, while respecting the required identifiability conditions.

3) {\bf Probe and Code Design:}
Select the contents of the probes sent by the sources and the operations performed at intermediate nodes. The goal is to use the simplest  operations and the smallest finite field, while ensuring that the identifiability conditions are met.

4) {\bf Estimation Algorithm}: This is the algorithm that processes the collected probes at the receivers
and estimates the link loss rates. The objective is low complexity with good estimation performance. There is clearly a tradeoff between the estimation error and the measurement bandwidth.

We note that these steps are {\em not} independent from each other. In fact, the design of routing, probe and code design needs to be done with identifiability and estimation in mind.

\subsection{Main Results}
\label{sec-mainResults}

In this paper, we propose a monitoring scheme for loss tomography in networks that have multicast and network coding capabilities. In Sections \ref{sec:trees} and \ref{sec:general}, we present our design for the cases of  trees and  general topologies, respectively. We evaluate all our schemes through extensive simulation results. Below we preview the main results, in each subproblem.

1) {\bf Identifiability:}
(1) We provide simple necessary and sufficient conditions for {\em identifying} the loss rate of a single link. In (logical) tree topologies, all links are  identifiable, using a very simple monitoring scheme\footnote{This scheme is described in Section~\ref{sec-model}: it selects some leaf nodes as sources, and the remaining leaf nodes as receivers; the sources send simple binary vectors, and the intermediate nodes do simple \texttt{XOR} operations or multicast.}. In general topologies, where identifiability depends on the routing and code design as well, these conditions still apply. (2) We also prove a structural property, which we call {\em reversibility}: if a link is identifiable under a given monitoring scheme, it remains identifiable if we reverse the directionality of all paths and exchange the role of sources and receivers (which we call the {\em dual configuration}). 

2) {\bf Routing:} (1) For a given set of sources and receivers over an arbitrary topology,
the problem of selecting a routing that meets the identifiability conditions while minimizing
the employed bandwidth is NP-hard. We prove that, when network coding is used, this problem can be solved in polynomial time. (2) Moreover, we demonstrate, via simulation, that the choice of sources and receivers affects the estimation accuracy. (3) Finally, we present heuristic orientation algorithms for general  graphs, designed to  achieve identifiability, small number of receivers, and high estimation accuracy.

3) {\bf Probe and Code Design:} (1) In trees, we show that binary vectors sent by the sources and deterministic code design with \texttt{XOR} operations at the intermediate nodes are sufficient. (2) In general graphs, we need to use operations over higher finite fields. We provide bounds on the required alphabet size, and we propose and evaluate deterministic code design.  

4) {\bf Loss Estimation:} 
(1) In a  tree topology (under mild conditions on the selection of sources and receivers),  we develop a low-complexity method for computing the MLE of the loss rates {\em for all links simultaneously}. Our algorithm builds on and extends MINC (the well-known ML estimator \cite{minc} for a multicast tree) to multiple-source multiple-destination tree topologies (with multicast at branching points and network coding at joining points). We describe the algorithm, prove its correctness, and analyze its rate of convergence. (2) A key property that we formulate, prove, and extensively use in this work, is {\em reversibility}, \ie the fact that the MLE's for a configuration and its dual (defined as the same topology, but with the role of sources and receivers reversed) have the same functional form. For example, the MLE for a {\em reverse multicast tree} (with several sources and one receiver) has the same functional form as MINC for a multicast tree (with the role of the source and the receivers reversed); we refer to the MLE for the reverse multicast tree as RMINC. (3) For topologies other than trees, no efficient MLE algorithm is known for  estimating the loss rates of all links simultaneously. Therefore, we propose a number of heuristic algorithms, including belief propagation and subtree decomposition algorithms, and we evaluate their performance through simulation. 
(4) We provide a simple algorithm for computing the MLE of {\em a single link} at a time in {\em any} topology. This is particularly useful in practice because: (i) a few bottleneck links are typically congested, thus of interest; and (ii) the method is applicable to {\em any} topology, even if it is not of the type (1) above. 

The use of network coding at intermediate nodes, in addition to unicast and multicast, offers several benefits for loss tomography: it increases the number of identifiable links; it improves the  tradeoff between number of probes and estimation accuracy; and it reduces  the complexity of selecting probe paths for minimum cost monitoring of a general graph  from NP-hard to linear. The approach gracefully generalizes from trees to general topologies ({\em e.g.,} having the same identifiability conditions, using the same estimation algorithm, and avoiding the use of overlapping trees or paths), where its advantages are amplified.

\section{\label{sec:singlelink}Motivating Example}

In this section, we present a motivating example to demonstrate the benefits of network coding in identifying the link loss rates; we derive the conditions of identifiability for a single link; and we discuss the identifiability of all links in the network.

\begin{figure*}[t!]\centering
\begin{center}
\psset{unit=0.045in}
\begin{pspicture}(5,40)(105,100)
\psset{linewidth=0.5mm}
\begin{small}

\rput(15,91){ \large{\em Tree 1}}
\rput(5,85){\circlenode{A}{A}}
\rput(25,85){\circlenode{B}{B}}
\rput(15,77){\circlenode{C}{C}}
\rput(15,63){\circlenode{D}{D}}
\rput(5,55){\circlenode{E}{E}} 
\rput(25,55){\circlenode{F}{F}}
\ncline[linewidth=0.5mm,linecolor=blue]{->}{A}{C}\Bput{$x_1$}
\ncline[]{->}{B}{C}
\ncline[linewidth=0.5mm,linecolor=blue]{->}{C}{D}\Aput{$x_1$}
\ncline[linewidth=0.5mm,linecolor=blue]{->}{D}{E}\Bput{$x_1$}
\ncline[linewidth=0.5mm,linecolor=blue]{->}{D}{F}\Aput{$x_1$}

\rput(55,91){\large{\em Tree 2}} 
\rput(45,85){\circlenode{A2}{A}} 
\rput(65,85){\circlenode{B2}{B}} \rput(55,77){\circlenode{C2}{C}}
\rput(55,63){\circlenode{D2}{D}} \rput(45,55){\circlenode{E2}{E}}
\rput(65,55){\circlenode{F2}{F}}
\ncline[]{->}{A2}{C2}
\ncline[linewidth=0.5mm,linecolor=green]{->}{B2}{C2}\Aput{$x_2$}
\ncline[linewidth=0.5mm,linecolor=green]{->}{C2}{D2}\Aput{$x_2$}
\ncline[linewidth=0.5mm,linecolor=green]{->}{D2}{E2}\Aput{$x_2$}
\ncline[linewidth=0.5mm,linecolor=green]{->}{D2}{F2}\Bput{$x_2$}

\rput(95,91){\large{\em Network Coding}}
\rput(85,85){\circlenode{A3}{A}} 
\rput(105,85){\circlenode{B3}{B}}
\rput(95,77){\circlenode{C3}{C}}
\rput(95,63){\circlenode{D3}{D}}
\rput(85,55){\circlenode{E3}{E}}
\rput(105,55){\circlenode{F3}{F}}
\ncline[]{->}{A3}{C3}\Bput{$x_1$}
\ncline[]{->}{B3}{C3}\Aput{$x_2$}
\ncline[linewidth=0.5mm,linecolor=red]{->}{C3}{D3}\Aput{$x_1+x_2$}
\ncline[]{->}{D3}{E3}\Bput{$x_1+x_2$}
\ncline[]{->}{D3}{F3}\Aput{$x_1+x_2$}

\end{small}
\end{pspicture}
\end{center}
\vspace{-5em}
\caption{\label{figure_1} Link loss monitoring for the basic 5-link topology. Nodes $A$ and $B$ are sources, and nodes $E$ and $F$ are receivers. Using multicast-based tomography, the topology can be covered using two multicast trees 1 and 2. Alternatively, the topology can be covered using coded packets, if node  $C$ can add (\texttt{XOR}) incoming packets.}
\end{figure*}
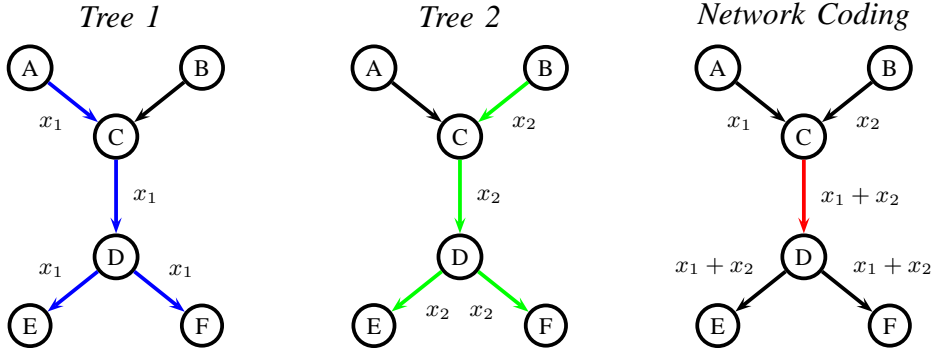

\begin{example} \label{ex:ex1}

Consider the 5-link topology depicted in Fig.~\ref{figure_1}. Nodes $A$ and $B$ send probes and nodes $E$ and $F$ receive them. Every link can drop a packet according to an i.i.d. Bernoulli distribution, with probability $\overline \alpha_e$, independently of other links. We are interested in estimating the success probabilities of all links, namely $\alpha_{AC}$, $\alpha_{BC}$, $\alpha_{CD}$, $\alpha_{DE}$, and $\alpha_{DF}$.

\begin{table*}[t!]
\scriptsize
\centering
\begin{tabular}{|c||c|c|c|c|c||c|c|c|c||c|c|c||c|c|}
\hline
  \# &  \multicolumn{5}{|c||}{Is link working (1) or not (0)?} & \multicolumn{2}{c|}{{\bf Original (5-link) Tree}} & {\footnotesize Prob.} & {\footnotesize \#times} & \multicolumn{3}{c||}{{\bf Reduced Multicast Tree}}  & \multicolumn{2}{c|}{{\bf Reduced Reverse}} \\
 ~ &  \multicolumn {5}{|c||}{~} & \multicolumn{2}{c|}{actual probes received at:} & ~ & ~ & \multicolumn{3}{c||}{observations} & \multicolumn{2}{c|}{{\bf Multicast Tree}} \\
~ & AC & BC & CD & DE & DF  & E & F &  $P_{\alpha}$ & ~ & E & F & $P_{\alpha}^{m}$ & EF & $P_{\alpha}^{r}$\\
\hline \hline
1 & \multicolumn{5}{c||}{Multiple possible events} & - & -     & $p_{0}$ & $n_0$  & 0 & 0 & $p_0$ &$[0,0]$ & $p_0$ \\
\hline
2 &  1 & 0 & 1 & 1 & 0  & $x_1$  & -               & $p_{1}$ & $n_1$ & \multirow{3}{*}{1} & \multirow{3}{*}{0}& \multirow{3}{*}{$p_1+ p_2+ p_3$} & $[1,0]$ & $p_1+p_4+p_7$ \\
3 & 0 & 1 & 1 & 1 & 0  & $x_2$ & -              & $p_{2}$ & $n_2$ & & & & $[0,1]$ & $p_2+p_5+p_8$\\
4 &  1 & 1 & 1 & 1 & 0  & $x_1 \oplus x_2$ & -            & $p_{3}$ & $n_3$ & & & & $[1,1]$ & $p_3+p_6+p_9$\\
\hline
5 &  1 & 0 & 1 & 0 & 1 & - & $x_1$ & $p_{4}$ & $n_4$ & \multirow{3}{*}{0} & \multirow{3}{*}{1} & \multirow{3}{*}{$p_4+ p_5+ p_6$} & $[1,0]$ & $p_1+p_4+p_7$ \\
6 &  0 & 1 & 1 & 0 & 1 &  -    & $x_2$  &$p_{5}$ & $n_5$ & & & & $[0,1]$  & $p_2+p_5+p_8$  \\
7 & 1 & 1 & 1 & 0 & 1  & - & $x_1 \oplus x_2$     & $p_{6}$ & $n_6$ & & & & $[1,1]$  & $p_3+p_6+p_9$\\
\hline
8 & 1 & 0 & 1 & 1 & 1  & $x_1$ & $x_1$            & $p_{7}$ & $n_7$ & \multirow{3}{*}{1} & \multirow{3}{*}{1} & \multirow{3}{*}{$p_7+p_8+ p_9$} &$[1,0]$ & $p_1+p_4+p_7$ \\
9 & 0 & 1 & 1 & 1 & 1  & $x_2$ & $x_2$            & $p_{8}$ & $n_8$ & & & & $[0,1]$ & $p_2+p_5+p_8$ \\
10 & 1 & 1 & 1 & 1 & 1  & $x_1 \oplus x_2$ & $x_1 \oplus x_2$ & $p_{9}$ & $n_{9}$ & & & & $[1,1]$ & $p_3+p_6+p_9$\\
\hline
\end{tabular}
\vspace{0.5em}
\caption{The 10 leftmost columns of this table refer to the 5-link topology in Fig.~\ref{figure_1}(c). They show the possible pairs of probes collected ({\em i.e.,} the observations $x \in \Omega$) at the receivers $E$, $F$, their probabilities $P_\alpha$, and the number of times $n_i$ each observation occurred. These observations depend on the combination of loss (0) and success (1) on the five links, which happen w.p. $\alpha$. The remaining rightmost columns show how the same probes can be interpreted as observations at the receiver(s) of the reduced topologies, namely the multicast and the reverse multicast trees (as we describe in Section~\ref{sec-reductions}), and their corresponding probabilities. \label{tabl_1} }
\end{table*}

The traditional multicast-based tomography approach would use two multicast trees rooted at nodes $A$ and $B$ and ending at $E$ and $F$. This approach is depicted in Fig. \ref{figure_1}-(a) and (b).
At each experiment, source $A$ sends packet $x_1$ and source  $B$ sends packet $x_2$.
The receivers $E$ and $F$ infer the link loss rates by keeping track of how many times they
receive packets $x_1$ and $x_2$. Note that, due to the overlap of the two trees,
for each experiment, links $CD$, ${DE}$, and ${DF}$ are used twice, leading to inefficient bandwidth usage. Moreover, from this set of experiments, we cannot calculate $\alpha_{CD}$,
and thus edge  ${CD}$ is not identifiable. Indeed, by observing the outcomes of experiments on each multicast tree, we cannot distinguish whether packet  $x_1$ is dropped on edge ${AC}$ or  ${CD}$; similarly, we cannot  distinguish whether packet  $x_2$ is dropped on edge ${BC}$ or  ${CD}$.
(Note that if we restricted ourselves to unicast only, four unicast probes from $A,B$ to $E,F$ would be needed to cover all five links. Not only would the problems of identifiability and overlap of probe paths still be present, but they would be further amplified.)

If network coding capabilities are available, they can help alleviate these problems. Assume that the intermediate node $C$ can combine incoming packets before forwarding them to outgoing links. Node $A$ sends to $C$ a probe packet with payload that contains the binary string $x_1=[1\;0]$. Similarly, node $B$ sends probe packet $x_2=[0\;1]$ to node $C$. If node $C$ receives only $x_1$ or only $x_2$, then it just forwards the received packet to node $D$; if $C$ receives both
packets $x_1$ and $x_2$, then it creates a new packet, with payload their linear combination $x_3=[1\;1]$, and forwards it to node $D$; more generally, $x_3=x_1 \oplus x_2$, where $\oplus$ is the bit-wise \texttt{XOR} operation. Node $D$ multicasts the incoming packet $x_3$ to both outgoing links $DE$ and $DF$. The flow of packets in this experiment is shown in Fig. \ref{figure_1}(c).
In every experiment, probe packets $(x_1,x_2)$ are sent from $A$,
$B$, and may or may not reach $E$, $F$, depending on the state of the links.
Observe that with the network coding approach, link $CD$ becomes identifiable. Moreover,
we have avoided the overlap of probes on link CD during each experiment.

Table \ref{tabl_1} lists the 10 possible observed outcomes, the state of the links that leads to a particular outcome, the probability $p_{i}$, $i=0,...,9$ of observing this outcome, and the number of times $n_i$, $i=0,...,9$ we observe this outcome in a sequence of $n$ independent experiments. The probability of observing an outcome $p_i$ can be computed from the success probabilities $\alpha=(\alpha_{AC}, \alpha_{BC}, \alpha_{CD}, \alpha_{DE}, \alpha_{DF})$ of the five links. {\em E.g.,} for outcomes 1-4:
\begin{equation}
\begin{split}
p_{0}& =1-p_{1} \cdots -p_{9}= 1-(1-\overline \alpha_{AC} \overline \alpha_{BC})\alpha_{CD}(1-\overline \alpha_{DE} \overline \alpha_{DF}) \\
p_{1}& = \alpha_{AC} \overline \alpha_{BC} \alpha_{CD} \alpha_{DE} \overline \alpha_{DF} \\
p_{2}& = \overline \alpha_{AC} \alpha_{BC} \alpha_{CD} \alpha_{DE} \overline \alpha_{DF} \\
p_{3}& =  \alpha_{AC} \alpha_{BC} \alpha_{CD} \alpha_{DE} \overline \alpha_{DF} \\
\cdots
\end{split}
\end{equation}
and we can write similar expressions for the probabilities of the remaining observations. Thus, we can explicitly write down the probability distribution of the observations $P_{\alpha}$.

In a sequence of $n=\sum_{i=0}^{i=9}{n_i}$ independent experiments, the frequency of each event $i$ is $\hat{p}_i = \frac{n_i}{n}$. After sending $n$ independent probes, the log-likelihood function of the observations given the set of parameters $(\alpha_e)$ is:
$\mathcal{L}(\alpha_{AC},\alpha_{BC},\alpha_{CD}, \alpha_{DE}, \alpha_{DF})=\sum_{i=0}^{i=9}{n_i \log p_{i}(\alpha)}$. The MLE would compute the $\alpha$'s that maximize $\mathcal{L}(\alpha)$.
\hfill{$\square$}
\end{example}

In general, we may be interested in estimating one of the $\alpha$ variables, some of them, or all five of them. In the next Section, we discuss a single link, namely link $CD$. Note that the remaining four links can depict the equivalent paths connecting $CD$ to the sources and receivers. In Section~\ref{sec:trees-identify}, we discuss the identifiability of all links.

\subsection{Identifiability of One Link} \label{sec-identify-one-link}

Let us focus on a single link $CD$ with success probability $\alpha_{CD}$. Consider Fig.~\ref{fig_basic} , which generalizes the motivating example of the previous Section. Note that links other than $CD$ can be viewed as summarizing paths: {\em e.g.,} AC could correspond to a path from A to C, possibly consisting of the concatenation of several links.

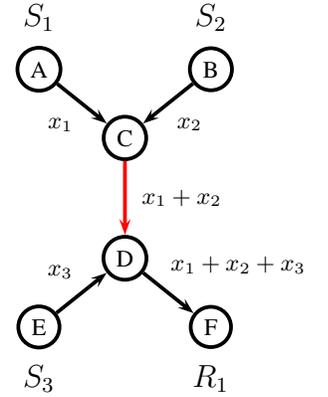
\begin{figure}[t!]
\centering
\vspace{1.5em}
\begin{center}
\psset{unit=0.045in}
\begin{pspicture}(-2,-8)(82,145)
\psset{linewidth=0.5mm}
\begin{small}
\rput(15,147){ \normalsize{\em 3-links, Multicast Tree}}
\rput(15,142){\large{$S_1$}}
\rput(15,137){\circlenode{C12}{C}}
\rput(15,124){\circlenode{D12}{D}}
\rput(5,115){\circlenode{E12}{E}}
\rput(5,110){\large{$R_1$}}
\rput(25,115){\circlenode{F12}{F}}
\rput(25,110){\large{$R_2$}}
\ncline[]{->}{C12}{D12}\Bput{$x_1$}
\ncline[]{->}{D12}{E12}\Bput{$x_1$}
\ncline[]{->}{D12}{F12}\Aput{$x_1$}

\rput(55,147){\normalsize{\em 3-links, Reverse Multicast Tree}}
\rput(45,142){\large{$S_1$}}
\rput(45,137){\circlenode{A11}{A}}
\rput(65,142){\large{$S_2$}}
\rput(65,137){\circlenode{B11}{B}}
\rput(55,128){\circlenode{C11}{C}}
\rput(55,115){\circlenode{D11}{D}}
\rput(55,110){\large{$R_1$}}
\ncline[]{->}{A11}{C11}\Bput{$x_1$}
\ncline[]{->}{B11}{C11}\Aput{$x_2$}
\ncline[linewidth=0.5mm,linecolor=red]{->}{C11}{D11}\Aput{$x_1+x_2$}

\rput(15,102){ \normalsize{\em 5-links, Case 1}}
\rput(5,98){\large{$S_1$}}
\rput(5,92){\circlenode{A}{A}}
\rput(25,98){\large{$S_2$}}
\rput(25,92){\circlenode{B}{B}}
\rput(15,84){\circlenode{C}{C}}
\rput(15,70){\circlenode{D}{D}}
\rput(5,62){\circlenode{E}{E}} 
\rput(5,56){\large{$R_1$}}
\rput(25,62){\circlenode{F}{F}}
\rput(25,56){\large{$R_2$}}
\ncline[]{->}{A}{C}\Bput{$x_1$}
\ncline[]{->}{B}{C}\Aput{$x_2$}
\ncline[linewidth=0.5mm,linecolor=red]{->}{C}{D}\Aput{$x_1+x_2$}
\ncline[]{->}{D}{E}\Bput{$x_1+x_2$}
\ncline[]{->}{D}{F}\Aput{$x_1+x_2$}
\rput(55,102){\normalsize{\em 5-links, Case 2}} 
\rput(45,98){\large{$S_1$}}
\rput(45,92){\circlenode{A2}{A}} 
\rput(65,98){\large{$R_3$}}
\rput(65,92){\circlenode{B2}{B}} 
\rput(55,84){\circlenode{C2}{C}}
\rput(55,70){\circlenode{D2}{D}} 
\rput(45,62){\circlenode{E2}{E}}
\rput(45,56){\large{$R_1$}} 
\rput(65,62){\circlenode{F2}{F}}
\rput(65,56){\large{$R_2$}}
\ncline[]{->}{A2}{C2}\Bput{$x_1$}
\ncline[]{<-}{B2}{C2}\Aput{$x_1$}
\ncline[linewidth=0.5mm,linecolor=red]{->}{C2}{D2}\Aput{$x_1$}
\ncline[]{->}{D2}{E2}\Aput{$x_1$}
\ncline[]{->}{D2}{F2}\Aput{$x_1$}

\rput(15,47){\normalsize{\em 5-links, Case 3}} 
\rput(5,43){\large{$S_1$}}
\rput(5,37){\circlenode{A3}{A}} 
\rput(25,43){\large{$R_1$}}
\rput(25,37){\circlenode{B3}{B}} 
\rput(15,29){\circlenode{C3}{C}}
\rput(15,15){\circlenode{D3}{D}} 
\rput(5,7){\circlenode{E3}{E}}
\rput(5,1){\large{$S_2$}} 
\rput(25,7){\circlenode{F3}{F}}
\rput(25,1){\large{$R_2$}}
\ncline[]{->}{A3}{C3}\Bput{$x_1$}
\ncline[]{<-}{B3}{C3}\Aput{$x_1$}
\ncline[linewidth=0.5mm,linecolor=red]{->}{C3}{D3}\Aput{$x_1$}
\ncline[]{<-}{D3}{E3}\Bput{$x_2$}
\ncline[]{->}{D3}{F3}\Aput{$x_1+x_2$}
\rput(55,47){\normalsize{\em 5-links, Case 4}} 
\rput(45,43){\large{$S_1$}}
\rput(45,37){\circlenode{A4}{A}} 
\rput(65,43){\large{$S_2$}}
\rput(65,37){\circlenode{B4}{B}} 
\rput(55,29){\circlenode{C4}{C}}
\rput(55,15){\circlenode{D4}{D}} 
\rput(45,7){\circlenode{E4}{E}}
\rput(45,1){\large{$S_3$}}
\rput(65,7){\circlenode{F4}{F}}
\rput(65,1){\large{$R_1$}}
\ncline[]{->}{A4}{C4}\Bput{$x_1$}
\ncline[]{->}{B4}{C4}\Aput{$x_2$}
\ncline[linewidth=0.5mm,linecolor=red]{->}{C4}{D4} \Aput{$x_1+x_2$}
\ncline[]{<-}{D4}{E4}\Bput{$x_3$}
\ncline[]{->}{D4}{F4}\Aput{{$x_1+x_2+x_3$}}
\end{small}
\end{pspicture}
\end{center}
\vspace{-3.0em}
\caption{{\bf Configurations ({\em i.e.,} combinations of Conditions 1 and 2) that allow us to identify the success rate of a single link (CD).} Recall that links, other than CD, can correspond to paths with the same loss probability. The top of the figure shows a 3-link topology where C is a source (of a multicast tree) or D is a receiver (of a reverse multicast tree). The trivial case that C is a source and D is a receiver corresponds to a single-link topology and is omitted here. The bottom of the figure shows a 5-link topology and four configurations (choices of sources and receivers), where neither $C$ nor $D$ are edge nodes and packets are sent and received at the edge nodes $A$, $B$, $E$ and $F$. Case $1$ is our familiar motivating example; Case $2$ is similar to a single multicast tree rooted at $A$; Case $3$ uses sources $A$ and $E$ and linear combinations whenever the two flows meet; Case $4$ does the same thing for sources $A$, $B$ and $E$, and is equivalent to an inverse multicast tree (with sink at $F$).}
\vspace{-1.0em}
\label{fig_basic}
\end{figure}

For a given choice of sources and receivers and a coding scheme described in Section~\ref{sec-model} (which is extremely simple: just pick any leaf or leaves as sources and the remaining leaves as receivers; sources send binary vectors; intermediate nodes simply code using bit-wise \texttt{XOR} or multicast), we want to translate the conditions for identifiability of link $CD$ in Definition~\ref{def_ident} to graph properties of the network. Our intuition is that a link $CD$ is identifiable if $C$ is a source, a coding point or a branching point, and $D$ is a receiver, a coding point or a branching point. These are the structures depicted in Fig.~\ref{fig_basic}, where we want to identify the link success rate associated with edge $CD$, and interpret the remaining edges as corresponding to paths. The top two cases of Fig.~\ref{fig_basic} depict the simple cases where node $C$ is a source, or node $D$ is a receiver; the four bottom cases depict the cases where $C$ and $D$ are coding or branching points.

To formalize this intuition, consider the following two conditions:
\begin{itemize}
\item {\bf Condition 1:} At least one of the following holds:\\
{\bf (a)} $C \in S$.\\
{\bf (b)} There exist two edge-disjoint paths $(X_1, C)$ and $(X_2,C)$ that do not employ edge $CD$, with distinct $X_1,X_2 \in S$.\\
{\bf (c)} There exist two paths $(X_1, C)$ and $(C,X_2)$ that do not employ edge $CD$, with $X_1 \in S$, $X_2 \in R$.
\item {\bf Condition 2:} At least one of the following holds:\\
{\bf (a)} $D \in R$. \\
{\bf (b)} There exist two edge-disjoint paths $(D,X_1)$ and $(D,X_2)$ that do not employ edge $CD$, with distinct $X_1,X_2 \in R$.\\
{\bf (c)} There exist two paths $(X_1, D)$ and $(D,X_2)$ that do not employ edge $CD$, with $X_1 \in S$, $X_2 \in R$.
\end{itemize}

\begin{theorem} \label{theorem_1}
For a given choice of sources and receivers and for the simple coding scheme described above, link $CD$ is identifiable if and only if both Conditions 1 and 2 hold.
\end{theorem}
The proof is provided in Appendix A.1.

\subsection{\label{sec:trees-identify}Identifiability of All Links}

In fact, we can identify {\em all} links at the same time. It is sufficient to ensure that each link is identifiable, according to the conditions of Theorem \ref{theorem_1}. This is true in all directed trees, where each leaf node is either a source or a receiver, and each intermediate node satisfies the following mild conditions: (i) it has degree at least three (which is true in all logical topologies); (ii) it has in-degree at least one (otherwise, the node should be a source); and (iii) it has out-degree at least one (otherwise, the node should be a receiver).

\begin{example}
Table~\ref{tabl_2} lists which links are identifiable in the four bottom cases of Fig. \ref{fig_basic},
if we use our approach vs. if we use multicast tomography. All four configurations depict the same basic 5-link topology, but they differ in the choice of sources and receivers. Our approach is able to identify all links for any sets of sources and receivers. This is not always the case for the multicast tomography. \hfill{$\square$}
\end{example}
\begin{table}[h!]
\scriptsize
\centering
\begin{tabular}{|c|c|c|}\hline
Case & Network Coding  & Multicast Probes  \\\hline 1 &    all links
& $DE$, $DF$    \\\hline 2 &    all links  &     all links  \\\hline
3 & all links  &  $AC$, $CB$     \\\hline 4 &   all links & no links
\\\hline
\end{tabular}
\vspace{-0.5em}
\caption{\label{tabl_2} Identifiable links in the four cases (different choices of sources and receivers, for the same 5-link topology) depicted at the bottom of Fig. \ref{fig_basic}.}
\vspace{-1.5em}
\end{table}

\section{\label{sec:trees}Tree topologies}

In this Section, we consider tree topologies, and we describe our design choices in the four subproblems: we have already discussed identifiability in the previous Section. Next, we describe routing in Section~\ref{sec:trees-pov}, probe and code design in Section~\ref{sec-model} (operation of sources and intermediate nodes), and estimation algorithms in Sections~\ref{sec-estimation1},~\ref{sec-estimation-onelink}, and \ref{sec-estimation2}.

\subsection{\label{sec:trees-pov}Routing, Selection of Sources and Receivers}

Routing in trees is well defined: there exists a single path that connects a source to a receiver, through which probes flow.  For a tree with $L$ leaf nodes, some leaves act as sources $S$ and the remaining leaves act as receivers $R=L\setminus S$. Intermediate nodes simply combine (\texttt{XOR}) the probes coming on all incoming links and forward (multicast) to all their outgoing links. This Section looks at situations where we may have some freedom in the choice of the nodes  that act as sources and receivers. If such flexibility is not available (as it is assumed in most tomography work), this step can be skipped. We study the effect of the selection of sources and receivers on estimation accuracy and we come up with empirical guidelines for source selection, obtained through a number of examples and simulation scenarios.

In Example 2, we saw that, with network coding, all links are identifiable, while if we use two multicast trees, they are not. In Appendix B.2, we revisit the basic 5-link topology of Fig. \ref{fig_basic} and we show that, even though with network coding links are identifiable for all four cases, the estimation accuracy differs depending on the number of sources and their relative positions in the tree. This idea also applies to larger topologies. For example, in \cite{technicalReport}, we consider a 9-link tree and we run simulations for different number and location of sources and we summarize the intuition obtained.

Link loss tomography is essentially a parameter estimation problem, and different choices of sources and receivers lead to different estimators. That is, for a fixed number of probes, each topology leads to a different estimation accuracy; put differently, to achieve the same  mean square error ($MSE$), we may need a different number of probes for each topology. In general, the optimal selection of the number and location of sources depends on the network topology, the values of link loss rates, and possibly the number of employed probes. This is currently an open problem.

\subsection{Maximum Likelihood Estimation of All Link Loss Rates} \label{sec-estimation1}

In this Section, we focus on tree topologies 
 and we develop an efficient maximum likelihood estimator to estimate all link loss rates from the observations at the receivers. In the special case where the topology is a {\em multicast tree}, \ie probes are sent between one source and several receivers,  an efficient ML estimator (MINC) has been designed in the pioneering paper \cite{minc}. We build on MINC, and we extend it to multiple-source multiple-receiver trees, where multicast is used at all branching points and network coding is used at all joining points \cite{netcod2011}.  We propose Alg. \ref{alg-MLE} in Section \ref{sec-MLE}, which provides an efficient way to compute the MLE of {\em all links at the same time}. 

A key property that we formulate, prove, and extensively use in this Section, is {\em reversibility}, as discussed in Section~\ref{sec-mainResults}, and as we describe in detail in Section~\ref{sec-MINCandRMINC}.
 In Section~\ref{sec-estimation-onelink}, we also describe how to efficiently compute the  MLE for {\em a single link} at a time (in both trees and general topologies).
 In Section~\ref{sec-estimation2}, we describe heuristic estimation algorithms, some of which apply to general topologies as well.

\subsubsection{Model and framework}
\label{sec-model}
We first describe the model of tree networks for which we derive the MLE.

\begin{figure}[t!]
\centering
\includegraphics[scale=0.38]{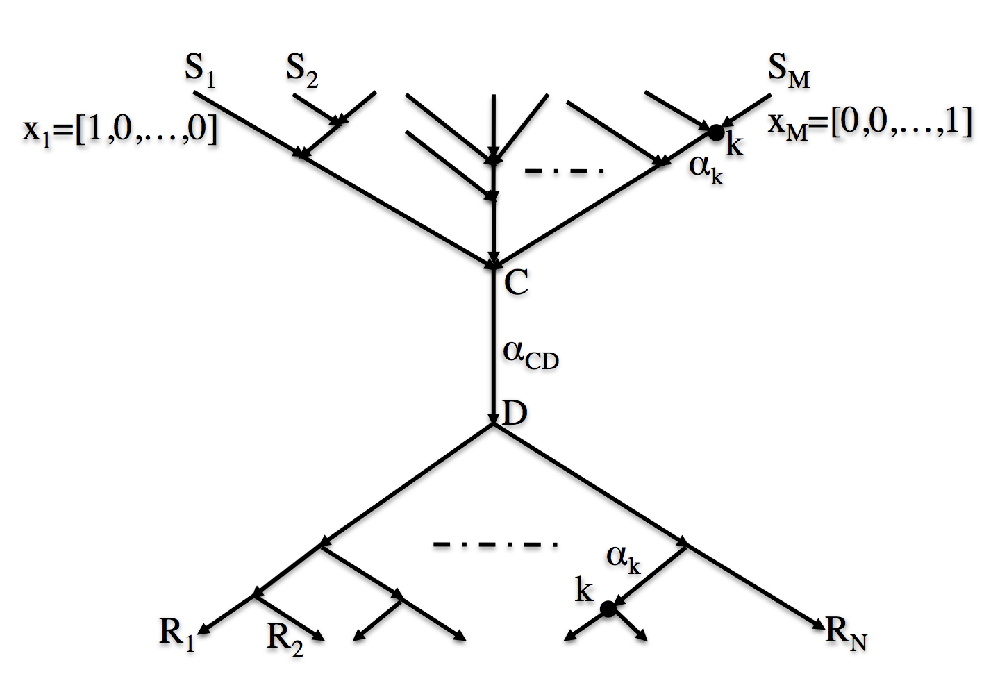}
\vspace{-0.5em}
\caption{A tree topology with multiple sources and multiple receivers. All sources are located at the top $M$ leaves, and all receivers are located at the bottom $N$ leaves. Multicast is used in all branching points and network coding is used in all joining/coding points. All coding points are located above all branching points. (This is a mild assumption that can be enforced if we are allowed to appropriately pick the sources and receivers.) For this tree topology, we have designed an algorithm that efficiently computes the MLE for {\em all links simultaneously.}}
\label{fig-generalTreeModel}
\vspace{-0.5em}
\end{figure}

{\bf Logical Tree.} We consider a tree topology, like the one depicted in Fig.~\ref{fig-generalTreeModel}, $G = (V,E)$ consisting of the set $V$ of nodes and the set $E$ of directed links. $M$ leaf nodes, shown on top of the tree, act as sources of probe packets. The remaining $N$ leaves, shown at the bottom of the tree, act as receivers. As typically assumed in tomography problems (as described in Section~\ref{sec:statement}), this is a ``logical'' tree topology, \ie every intermediate node has degree at least three.  An intermediate node is either a coding point (with multiple incoming links and one outgoing link) or a branching point (with one incoming link and multiple outgoing links).  For each node $j$, we denote the set of its parents (nodes with a link outgoing to $j$) by $f(j)$ and the set of its children (nodes with a link coming from $j$) by $d(j)$. The source nodes $S=\{S_1,...,S_M\}$ have no parent and the receiver nodes $R=\{R_1,...,R_N\}$ have no children. $G, S, R$ are considered known and fixed throughout the experiments. 

In this Section, we focus on the tree topology shown in Fig. \ref{fig-generalTreeModel}, which has the property that all coding points are located above all branching points. This is actually a mild assumption: starting from an undirected tree, if one is allowed to choose the sources among the leaf nodes, then one can always ensure this property.\footnote{Once the sources are properly chosen, the rest of the leaves are receivers; the direction of the links is uniquely defined along the paths from the sources to the receivers; and intermediate nodes perform either coding or multicast, as uniquely dictated by the direction of their incoming and outgoing edges.}  Note that this tree model includes all cases in  Fig.~\ref{fig_basic} (except for Case 3 in the 5-link topology, which is treated separately in Section~\ref{sec-estimation-onelink}).

{\bf Operation of Sources.} Each source $S_i$ sends a probe packet $x_i$, which is a vector of length $M$ in the form of: 
\begin{displaymath}
x_i=[\overbrace{\underbrace{0, \cdots, 0, 1}_i, 0, \cdots, 0}^M], \quad i=1,2,\cdots,M
\end{displaymath}
{\bf Operation of Intermediate Nodes.} Each coding point (bit-wise) \texttt{XOR}s all packets it receives from its parents, and forwards the result to its child\footnote{We assume that the network is delay-free and all packet arrivals at a coding point are synchronized. Link delays only affect where the probe packets would meet.}. This very simple design effectively keeps the presence of each source orthogonal from every other source. This ensures versatility, in the sense that no matter which probe packets get \texttt{XOR}-ed, they will not cancel each other out. For most practical purposes, this simple probe design is sufficient: a single IP packet can be up to 1500B (including the headers) and thus, can accommodate roughly 12,000 probe sources (bits). In large networks, one can also spatially reuse probe packets by  allocating the same probe packet to all sources whose packets do not meet. Finally, each branching point multicasts the packet it receives from its parent to all its children.

One can see that there will be a node after which $x_1+x_2+\cdots+x_M$ flows thought the network. We denote this node by $C$. Node $C$ is the last coding point in the tree. Node $C$ has $P$ parents $f(C)_1,\cdots,f(C)_P$, and only one child, which we denote by node $D$. Node $D$ multicasts the packet it receivers from node $C$ to all its $Q$ children $d(D)_1, \cdots,d(D)_Q$.

We use the notation that $k< k'$, $k,k'\in V$ when $k$ is a descendant of $k'$, and that $k>k'$ when $k$ is an ancestor of $k'$. Every node $k>C$ has multiple parents and only one child, while every node $k<D$ has one parent and multiple children. We are going to treat these two sets of nodes differently in the rest of Section~\ref{sec-estimation1}. We name any link of the tree that is above node $C$ by its starting point, and we name any link that is below node $D$ by its end point. In other words, link $k$ denotes a link between nodes $(k,j)$ if $k>C$ and $j>C$, while link $k$ denotes a link between nodes $(j,k)$ if $j<D$ and $k<D$.

{\bf Loss Model.} As described in Section~\ref{sec:statement}, we model the loss rate of individual links by an i.i.d. Bernoulli process, independent across links. In particular, we use the following notation:
\begin{itemize} 
\item A packet that traverses a link $k$ above node $C$ is lost with probability $\overline \alpha_k=1-\alpha_k$ and arrives at node $j$ with probability $\alpha_k$.
\item A packet that traverses a link $k$ below node $D$ is lost with probability $\overline \alpha_k=1-\alpha_k$ and arrives at node $k$ with probability $\alpha_k$.
\item Finally, we denote the loss rate of link $CD$ by $\overline \alpha_{CD}$.
\end{itemize}

In general, we use the notation $\overline \alpha=1-\alpha$ for any quantity $0<\alpha<1$.

Let $X_k$ denote the packet observed at node $k$, and let $X=(X_k)$, $k\in V$ denote the set of all $X_k$'s. $X_k$ is a binary vector of length $M$. Its $i^{th}$ element, $(X_k)_i$, represents the probe packet of source $i$: $(X_k)_i=1$ indicates that the probe packet of source $i$ reaches node $k$, and 0 that it does not. For the sources, $X_{S_i}=x_i$, thus $(X_{S_i})_i=1$ and $(X_{S_i})_{i'}=0$, $\forall i'\neq i$. For any node $k\geq C$, if $(X_j)_i=1$ for $j$ a parent of $k$, $(X_k)_i=1$ with probability $\alpha_j$, and $(X_k)_i=0$ with probability $\overline \alpha_j$, independently for all the parents of $k$. For any node $k\leq D$, if $X_k=[0,0,\cdots,0]$ (the all-zero vector), then $X_j=[0,0,\cdots,0]$, for the children $j$ of $k$ (and hence for all descendants of $k$). If $X_k\neq [0,0,\cdots,0]$, then for $j$ a child of $k$, $X_j=X_k$ with probability $\alpha_j$, and $X_j=[0,0,\cdots,0]$ with probability $\overline \alpha_j$, independently for all the children of $k$.

{\bf Data, Likelihood, and Inference.} As described in Section~\ref{sec:model-definitions}, in each experiment, one probe is dispatched from each source. The outcome of a single experiment is a record of whether or not each source probe was received at each receiver, which is the set of vectors $X_k$ observed at receiver $k\in R$. It is denoted by $X_{(R)}=(X_k)_{k\in R}$ and is an element of the space $\Omega \subseteq \{[\cdots,0,1,\cdots]\}^N$ of all such outcomes. For a given set of link probabilities $\alpha = (\alpha_k)_{k\in V\backslash\{C,D\}} \cup \alpha_{CD}$, the distribution of the outcomes $X_{(R)}$ on $\Omega$ will be denoted by $P_{\alpha}$. The probability mass function for a single outcome $x\in \Omega$ is $p(x;\alpha)=P_{\alpha}(X_{(R)}=x)$.

We perform $n$ experiments. The probability of $n$ independent observations $x^1,\cdots,x^n$ (each $x^t=(x^t_k)_{k\in R}$) is given by Eq.(\ref{eq-probability}). Our task is to estimate $\alpha$ using maximum likelihood, from the data $(n(x))_{x\in \Omega}$. We work with the log-likelihood function $\mathcal{L}(\alpha)$ given in Eq.(\ref{eq-likelihoodSum1}). The MLE of the loss rates $\breve{\alpha}$ is the $\alpha$ that maximizes $\mathcal{L}(\alpha)$, as given by Eq.(\ref{eq_ML}).

\subsubsection{The Likelihood Equation and its Solution}
\label{sec-MINCandRMINC}

Candidates for the MLE are solutions $\hat{\alpha}$ of the {\em likelihood equation:}
\begin{equation}
 \frac{\partial \mathcal{L}}{\partial \alpha_k}(\alpha)=0, \quad k\in V
\end{equation}
We need to define some additional variables to compute the MLEs. For each node $k\geq D$, let $\Omega^r(k)$ be the set of outcomes $x\in \Omega$ such that $(x_a)_j\neq 0$ for at least one source $j\in S$ that is an ancestor of $k$ and for any arbitrary set of receivers $\{a\}\subset R$. Let $\gamma^r_k=\Gamma^r_k(\alpha) =P_{\alpha}[\Omega^r(k)]$; an estimate of $\gamma^r_k$ can be computed from:
\begin{equation}
\label{eq-gamma-r}
\hat{\gamma}^r_k = \sum_{x\in \Omega^r(k)} \hat{p}(x),  \quad \text{where} \quad  \hat{p}(x) = \frac{n(x)}{n}
\end{equation}
is the observed proportion of experiments with outcome $x$. $\gamma^r_k$ shows the probability of the set of outcomes $\Omega^r(k)$ in which link $k$ has definitely worked. Note that link $k$ may have worked for some other outcomes as well, but they are not included in $\Omega^r(k)$. Also note that $\gamma^r_k$ can be directly estimated from the observations at the receivers.

For each node $k\leq C$, we define $\Omega^m(k)$ to be the set of outcomes $x\in \Omega$ such that $x_j\neq [0,0,\cdots,0]$ for at least one receiver $j\in R$ which is a descendant of $k$. Let $\gamma^m_k=\Gamma^m_k(\alpha) =P_{\alpha}[\Omega^m(k)]$; an estimate of $\gamma^m_k$ is:
\begin{equation}
\label{eq-gamma-m}
\hat{\gamma}^m_k = \sum_{x\in \Omega^m(k)} \hat{p}(x)
\end{equation}
$\gamma^m_k$ is the probability of the outcomes $\Omega^m(k)$ in which link $k$ has definitely worked; and it can be directly estimated from the observations at the receivers. Our goal is to compute $\hat{\alpha}$ from $\hat{\gamma}={(\hat{\gamma}^r_k \cup \hat{\gamma}^m_k)}_{k\in V}$.

{\bf Special Case (i): Multicast Tree (MINC).} If $M=1$, the general model turns into a multicast tree with a single source, which is the case considered in \cite{minc}. We represent the source node by $0\in V$. Each node $j$ other than the source node, has one parent $f(j)$, and a set $d(j)$ of children. We denote the link loss rates by $\overline \alpha_k$, where $k$ is the end point. We simply assume that $\alpha_0=1$.

The outcome of each experiment is $X_{(R)}=(X_k)_{k\in R}$, where each $X_k$ is a single binary value (instead of a binary vector of length $M$ in the general case), corresponding to whether the source probe is observed at each receiver $k\in R$ or not. The state space of the observations $X_{(R)}$ is $\Omega=\{0,1\}^N$. We say that a link $k$ is at level $l^m(k)$ if there is a chain of $l^m(k)$ ancestors $k<f(k)<f^2(k)\cdots<f^{l^m(k)}(k)=0$ leading back to the source. 

Only $\Omega^m(k)$ is used for each node $k$ in the multicast tree; it is the set of outcomes $x\in \Omega$ where $x_j=1$ for at least one receiver $j\in R$ that is a descendant of $k$. The definition of $\gamma^m_k$ is like before.

The MLE for the multicast tree has been computed in \cite{minc}: Let $A^m_k = \prod_{i=0}^{l^m(k)} \alpha_{f^i(k)}$ show the probability that the path from the source to node $k$ works, which we denote by $P(Y_{0\rightarrow k}=1)$. Its estimate $\hat{A}^m_k$ can be computed as follows. For the source node, $\hat{A}^m_0=1$, for the leaf nodes $k\in R$, $\hat{A}^m_k=\hat{\gamma}^m_k$, and for all other nodes $k\in V\backslash \{0,R\}$,  $\hat{A}^m_k$ is the unique solution in $(0,1]$ of:
\begin{equation}
\label{eq-MC1}
1-\frac{\hat{\gamma}^m_k}{\hat{A}^m_k} = \prod_{j\in d(k)} (1-\frac{\hat{\gamma}^m_j}{\hat{A}^m_k})
\end{equation}
$\hat{\alpha}_k$ can then be computed from $\hat{\gamma}^m_k$, \ie $\hat{\alpha}={\Gamma^{m}} ^{-1}(\hat{\gamma^m})$, as follows:
\begin{equation}
\label{eq-MC2}
\hat{\alpha}_k = \frac{\hat{A}^m_k}{\hat{A}^m_{f(k)}}, \quad k\in V\backslash \{0\} \quad (\hat{\alpha}_0=1)
\end{equation}
We refer to Eq.(\ref{eq-MC2}) as MINC in the rest of the paper.

{\em Note.} Eq.(\ref{eq-MC1}) is obtained from the following relations, after some computations in \cite{minc}, which we repeat here for completeness. Let $\beta^m_k=P[\Omega^m(k)|X_{f(k)}=1]$ denote the conditional probability of $\Omega^m(k)$ given that $f(k)$ has observed something. Failure can be due to either $\overline \alpha_k$ (failure of link $k$), or all paths towards the destinations failing. Therefore, the $\beta^m_k$ obey the following recursion:
\begin{equation}
\label{eq-MCbeta1}
\overline \beta^m_k = \overline \alpha_k + \alpha_k \prod_{j\in d(k)} \overline \beta^m_j, \quad k\in V\backslash R
\end{equation}
\begin{equation}
\label{eq-MCbeta2}
\beta^m_k = \alpha_k, \quad k\in R
\end{equation}
Eq.(\ref{eq-MC1}) then follows from the following relation between $\alpha$ and $\gamma^m$:
\begin{equation}
\label{eq-MCmainRelation}
\gamma^m_k = \beta^m_k \prod_{i=1}^{l^m(k)} \alpha_{f^i(k)}
\end{equation}

{\bf Special Case (ii): Reverse Multicast Tree (RMINC).} If $N=1$, the general model turns into a reverse multicast tree with a single receiver, which we denote by $0\in V$. Each node $j$ other than $0$ has one child $d(j)$, and a set $f(j)$ of parents. We denote link loss rates by $\overline \alpha_k$, where $k$ is the starting point. We assume that $\alpha_0=1$.

The outcome of each experiment, $X_R$, is a binary vector of length $M$. Each of its elements, $(X_R)_i$, represents whether the probe packet of source $i$ is observed at the receiver or not. The state space of the observations $X_{R}$ is $\Omega=\{0,1\}^M$. We say that a link $k$ is at level $l^r(k)$ if there is a chain of $l^r(k)$ descendants $k>d(k)>d^2(k)\cdots>d^{l^r(k)}(k)=0$ leading down to the receiver. 

Only $\Omega^r(k)$ is used for each node $k$ in the reverse multicast tree; it is the set of outcomes $x\in \Omega$ where $x_j=1$ for at least one source $j\in S$ that is an ancestor of $k$. The definition of $\gamma^r_k$ is like before.

The MLE for the reverse multicast tree is similar to the multicast tree. Let $A^r_k =  \prod_{i=0}^{l^r(k)} \alpha_{d^i(k)}$ show the probability that the path from node $k$ to the receiver node works, which we denote by $P(Y_{k\rightarrow 0}=1)$. Its estimate $\hat{A}^r_k$ can be computed as follows. For the receiver node, $\hat{A}^r_0=1$, for the source nodes $k\in S$, $\hat{A}^r_k=\hat{\gamma}^r_k$, and for all other nodes $k\in V\backslash \{S,0\}$, $\hat{A}^r_k$ is the unique solution in $(0,1]$ of:
\begin{equation}
\label{eq-RMC1}
1-\frac{\hat{\gamma}^r_k}{\hat{A}^r_k} = \prod_{j\in f(k)} (1-\frac{\hat{\gamma}^r_j}{\hat{A}^r_k})
\end{equation}
We can then compute $\hat{\alpha}_k$ from $\hat{\gamma}^r_k$, \ie $\hat{\alpha}={\Gamma^r}^{-1}(\hat{\gamma^r})$, as follows:
\begin{equation}
\label{eq-RMC2}
\hat{\alpha}_k = \frac{\hat{A}^r_k}{\hat{A}^r_{d(k)}}, \quad k\in V\backslash \{0\} \quad (\hat{\alpha}_0=1)
\end{equation}
We refer to Eq.(\ref{eq-RMC2}) as RMINC in the rest of the paper.

{\em Note.} Eq.(\ref{eq-RMC1}) results from the following relations. Let $\beta^r_k=P[\Omega^r(k)|Y_{d(k)\rightarrow 0}=1]$ denote the conditional probability of $\Omega^r(k)$ given that the path from $d(k)$ to the receiver works. We have that:
\begin{equation}
\label{eq-RMCbeta1}
\overline \beta^r_k = \overline \alpha_k + \alpha_k \prod_{j\in f(k)} \overline \beta^r_j, \quad k\in V\backslash S
\end{equation}
\begin{equation}
\label{eq-RMCbeta2}
\beta^r_k = \alpha_k, \quad k\in S
\end{equation}
\begin{equation}
\label{eq-RMCmainRelation}
\gamma^r_k = \beta^r_k \prod_{i=1}^{l^r(k)} \alpha_{d^i(k)}
\end{equation}

{\bf Comparison of MINC and RMINC.} The reader will notice that the MLE for the multicast tree and the reverse multicast tree have the same functional form. This is a special case of the more general ``reversibility'' property, first observed in \cite{globecom08}. Indeed, there is a 1-1 correspondence between the observable outcomes in the two cases; furthermore, the corresponding outcomes have the same probability, as a function of $\alpha_k$'s, thus leading to the same MLE. In the following, we describe the reversibility property in more detail.

{\bf Reversibility -- A Structural Property.} 
Consider a tree topology $G=(V,E)$ with $L$ leaf nodes, some of which act as sources $S$ and the remaining ones, $R=L\setminus S$, act as receivers of probes. Routing from $S$ to $R$ is given ({\em e.g.,} determined in the routing subproblem) and defines a direction on every link $e\in E$, along which probes flow. 
\begin{definition}
We call the triplet ($G, S,R)$ a {\em configuration}. 
\end{definition}

We define as dual the configuration that results from reversing the orientation of all links in the network, and from having the sources $S$ become receivers, while the receivers $R$ act as sources. More formally:
\begin{definition}
Consider the original configuration $(G,S,R)$. Consider the graph $G^d=(V,E^d)$ that has the same nodes but reversed edges, {\em i.e.,} $e=(i,j)\in E$ iff $e^d=(j,i)\in E^d$, and success rate $\alpha_e^d=\alpha_e$, associated with every edge $e^d\in E^d$. Select sources $S^d=R$ and receivers $R^d=S$. We call the $(G^d, S^d, R^d)$ the {\em dual configuration} of $(G, S,R)$.
\end{definition}

For example, a multicast tree is the dual configuration of a reverse multicast tree (Cases $2$ and $4$ in Fig.~\ref{fig_basic}). In Appendix B, we show that the dual configurations of Fig.~\ref{fig_fischer_5}(a) and Fig.~\ref{fig_fischer_5}(b) result in the same mean square error bound. In fact, a closer look reveals that not only the values but also the functional forms of these two ML estimators coincide.
The following theorem generalizes this notion to general trees.

\begin{theorem}
\label{th_reverse} Consider a configuration $(G,S,R)$ with observations at the receivers $\Omega$, and probability distribution $P_\alpha=\{p(x;\alpha), x\in \Omega\}$. Consider its dual configuration $(G^d,R,S)$, with observations $\Omega^d$ and probability  distribution $P_{\alpha}^d$. Then, there is a bijection between outcomes and their probabilities in the original $(x\in \Omega, p(x;\alpha))$ and in the dual configuration $(x^d\in \Omega^d, p(x^d;\alpha))$.
\end{theorem}
\begin{proof}
Let $G=(V,E)$ be the original tree graph, and $G^d$ its dual. In every experiment, there exist $2^{|E|}$ possible error events, depending on  which subset of the links fail. Observing the outcomes at the receivers corresponds to observing unions of events,
that occur with the corresponding probability ({\em e.g.,} as in the example of Table \ref{tabl_1}).
We show that for each observable outcome, which occurs with probability $p$ in $G$, there exists exactly one observable outcome that occurs with the same probability in  $G^d$
and vice-versa. This establishes a bijection.

With every edge $e$ of $G$, we can associate a set of sources
$S(e)\subset V$ that flow through this edge, and a set of receivers
$R(e)\subset V$ that observe the flow through $e$. Our main
observation is that the pair  $\{S(e),\;R(e) \}$ uniquely identifies
$e$, {\em i.e.,} no other edge has the same pair. In the dual configuration
$G^d$, edge $e$ is uniquely identified by the pair $\{R(e),\;S(e)
\}$. If in $G$, edge $e$ fails while all other edges do not, the
receivers $R(e)$ will not receive the contribution in the probe
packets of the sources  $S(e)$. If in $G^d$, edge $e$ fails while all
other edges do not, the receivers $S(e)$ will not receive the
contribution in the probe packets of the sources $R(e)$. Thus, there
is a one-to-one mapping between these events. Using this
equivalence, an observable outcome consisting of a union of events
can be mapped to an observable outcome in the reverse tree.
\end{proof}

\begin{corollary}
The maximum likelihood estimators for a configuration and its dual have the same
functional form.
\end{corollary}
\begin{proof}
The bijection established above implies that a configuration and its dual have  the same set of observable outcomes, with the same probabilities. Therefore, they have the same likelihood function and thus, the same maximum likelihood estimator.
\end{proof}
We note that this corollary establishes reversibility only for the
maximum likelihood estimation. The performance of suboptimal
algorithms may differ when applied to a configuration and its dual.

{\bf A note on directional networks.} It is also important to note that
the notion of dual configurations does {\em not} assume that the loss rates in both directions of a link are the same. Reversibility means that the two ML estimators for a configuration and its dual are described by the same function. However, the loss parameters we try to estimate (using the same estimator function) in the two directions may have different values.

\subsubsection{Maximum Likelihood Estimation of Loss Rates}
\label{sec-reductions}

We now present how to ``reduce'' the original tree to a multicast and to a reverse multicast tree, and how to estimate $\alpha_{CD}$. These intermediate results are then used in the MLE algorithm in Section~\ref{sec-MLE}.

{\bf Reduction to a Multicast Tree (m).} If we take the upper part of the original tree in Fig.~\ref{fig-generalTreeModel} and consider it as an aggregate link, we obtain the reduced multicast tree in Fig.~\ref{fig-generalMCreduction}. The aggregate link $agg^m$ summarizes the operation of all links above node $C$ and link $CD$. Node  $D$ receives a packet if at least one path from the sources to node $C$ works and link $CD$ works. In other words, the success probability of the aggregate link, $\alpha_{agg}^m$, depends on the paths from the sources to node $C$, and also link $CD$.

\begin{figure}[t!]
\subfigure[Reduced multicast tree.]{\includegraphics[width=4.3cm, height=3.4cm]{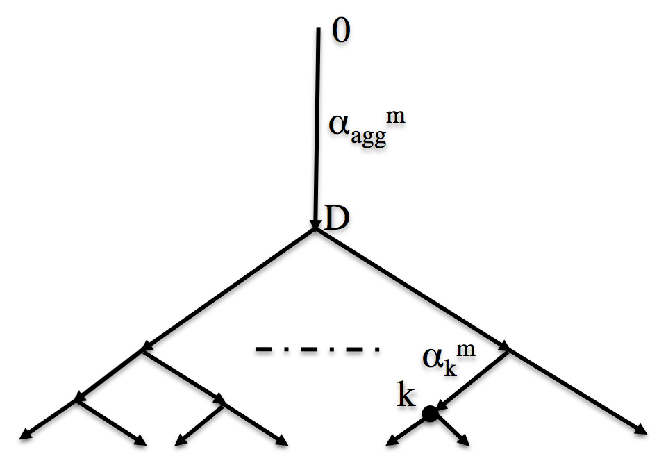}\label{fig-generalMCreduction}}
\hspace{0cm}
\centering \subfigure[Reduced reverse multicast tree.]{\includegraphics[width=4.3cm, height=3.4cm]{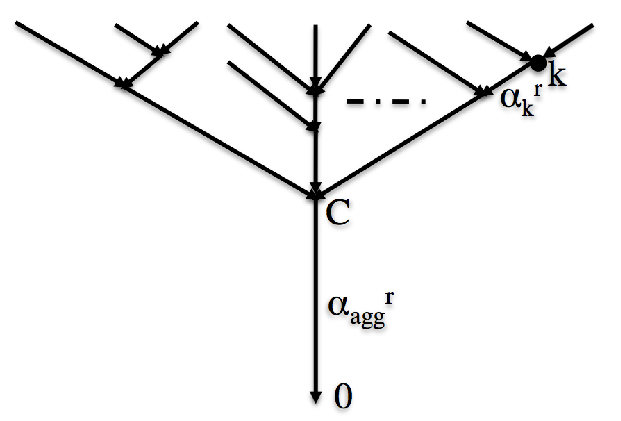}\label{fig-generalRMCreduction}}
\caption{\label{fig-generalTreeReductions} Reducing the tree topology in Fig.~\ref{fig-generalTreeModel} to a multicast tree and to a reverse multicast tree.}
\end{figure}

More formally, we map the outcomes $x\in \Omega$ of the original tree to the outcomes $x^m$ of the multicast tree, as follows. Each $x$ is a set of $N$ binary vectors, each of length $M$, while each $x^m$ is a single binary vector of length $N$. Any outcome $x^m$ is obtained by taking a set of outcomes $\{x\}$, in all of which the same receivers have observed all-zero vectors\footnote{Note that if a receiver does not receive any packet, then this is treated as an all-zero vector.} and the same receivers have observed non-zero vectors, and by replacing each non-zero vector (that may contain any of the source probes $x_1,...,x_M$) by value 1, and each all-zero vector by value 0. {\em I.e.:}
\begin{equation}
\label{eq-nRelation}
\begin{split}
& \sum_{x_{R_t}\neq [0,0,\cdots,0], x_{R_{t'}}=[0,0,\cdots,0]} n(x) = n^m(x^m), \\
& \quad x^m_{R_t}=1, x^m_{R_{t'}}=0, \; t,t'\in \{1,\cdots N\}, t\neq t' \\
\end{split}
\end{equation}
If the original tree has link success rates $\alpha$ and an associated probability distribution of outcomes $P_{\alpha}$, then the multicast tree is defined with parameters $\alpha^{m}$ and associated probability distribution $P_{\alpha}^{m}$, such that:
\begin{equation}
\label{eq:mt}
\alpha_{k}^m=\alpha_{k},  \, k<D, \quad \alpha_{agg}^m=\alpha_{CD} (1-\prod_{i=1}^{P} \overline \beta^r_{f(C)_i})
\end{equation}
$P_{\alpha}^m$ can be directly calculated from $P_{\alpha}$, since each event in $P_{\alpha}^m$ is the union of a disjoint subset of events in $P_{\alpha}$ and has probability equal to the sum of probabilities of those events in $P_{\alpha}$ (such as the 5-link example in Table~\ref{tabl_1}). 

{\bf Reduction to a Reverse Multicast Tree (r).} Similarly, if we consider the lower part of the original tree in Fig.~\ref{fig-generalTreeModel} as an aggregate link, we obtain the reduced reverse multicast tree in Fig. \ref{fig-generalRMCreduction}, with parameters $\alpha^{r}$ and associated probability distribution $P_\alpha^{r}$, such that:
\begin{equation}
\label{eq:rmt}
\alpha_{k}^{r}=\alpha_{k}, \, k>C, \quad \alpha_{agg}^r=\alpha_{CD} (1-\prod_{j=1}^{Q} \overline \beta^m_{d(D)_j})
\end{equation}

{\bf The Relation Between the Two Reduced Trees.}

\begin{lemma}
\label{thr-gammasEqual}
We have that: $\hat{\gamma}^r_C = \hat{\gamma}^m_D = 1-\hat{p}([0,0,\cdots,0])$.
\end{lemma}

The proof directly results from the definitions of $\gamma^m_D$ in the reduced multicast tree and $\gamma^r_C$ in the reduced reverse multicast tree.

{\bf Estimating $\mathbf{\alpha_{CD}}$.} The MLE of $\alpha_{CD}$ can be obtained from:
\begin{equation}
\label{eq-alphaCDestimate}
{\hat{\alpha}}_{CD} = \frac{\hat{A}^r_C\cdot \hat{A}^m_D}{\hat{\gamma}^r_C} = \frac{\hat{A}^r_C\cdot \hat{A}^m_D}{\hat{\gamma}^m_D}
\end{equation}
The proof can be found in Appendix A.2.

\begin{algorithm}[t!]
{\footnotesize
\caption{\label{alg-MLE} Computing the MLE of all link loss rates in the original tree topology of Fig. \ref{fig-generalTreeModel}.}
\begin{algorithmic}[1]
\FORALL{links $k$, where $k<D$}
\STATE Reduce the original tree to a multicast tree. Use MINC \cite{minc} (Eq.(\ref{eq-MC2})) to compute the MLEs ${\hat{\alpha}}_{k}^m$ and ${\hat{\alpha}}_{agg}^m$.
\STATE Let ${\hat{\alpha}}_{k}={\hat{\alpha}}_{k}^m$.
\ENDFOR
\FORALL{links $k$, where $k>C$}
\STATE Reduce the original tree to a reverse multicast tree. Use RMINC (Eq.(\ref{eq-RMC2})) to compute the MLEs ${\hat{\alpha}}_{k}^r$ and ${\hat{\alpha}}_{agg}^r$.
\STATE Let ${\hat{\alpha}}_{k}={\hat{\alpha}}_{k}^r$.
\ENDFOR
\STATE Use Eq.(\ref{eq-alphaCDestimate}) to compute the MLE ${\hat{\alpha}}_{CD}$.
\end{algorithmic}
}
\end{algorithm}

\subsubsection{The Analysis of the MLE}
\label{sec-MLE}

In this Section, we propose the MLE algorithm, we discuss its complexity, and we illustrate our results through the example tree topology in Fig.~\ref{figure_1}(c).

{\bf MLE Algorithm.} Algorithm~\ref{alg-MLE} computes the MLE of all link loss rates in the tree topology of Fig.~\ref{fig-generalTreeModel}; it proceeds in the following steps: (i) it computes $\hat{\alpha}_k$ for any link $k$ below node $D$ from the reduced multicast tree using Eq.(\ref{eq-MC2}); (ii) it computes $\hat{\alpha}_k$ for any link $k$ above node $C$ from the reduced reverse multicast tree using Eq.(\ref{eq-RMC2}); and (iii) it computes $\hat{\alpha}_{CD}$  from Eq.(\ref{eq-alphaCDestimate}).
These are indeed the MLEs of the link loss rates, $\hat{\alpha}$, for the tree of Fig.~\ref{fig-generalTreeModel}.

\begin{figure*}[t!]
{\scriptsize
\vspace*{4pt} \hrulefill
\begin{displaymath}
\mathcal{I}^{-1}(\alpha)=\left( \begin{array}{ccccc} \frac{\alpha_A\overline \alpha_A}{\alpha_B\alpha_{CD}(\alpha_E+\alpha_F-\alpha_E\alpha_F)} & \frac{\overline \alpha_A\overline \alpha_B}{\alpha_{CD}(\alpha_E+\alpha_F-\alpha_E\alpha_F)} &  \frac{-\overline \alpha_A\overline \alpha_B}{\alpha_B(\alpha_E+\alpha_F-\alpha_E\alpha_F)} & 0 & 0 \\ \frac{\overline \alpha_A\overline \alpha_B}{\alpha_{CD}(\alpha_E+\alpha_F-\alpha_E\alpha_F)} & \frac{\alpha_B\overline \alpha_B}{ \alpha_A\alpha_{CD}(\alpha_E+\alpha_F-\alpha_E\alpha_F)} & \frac{-\overline\alpha_A\overline \alpha_B}{ \alpha_A(\alpha_E+\alpha_F-\alpha_E\alpha_F)} & 0 & 0 \\ \frac{-\overline \alpha_A\overline \alpha_B}{\alpha_B(\alpha_E+\alpha_F-\alpha_E\alpha_F)}  &  \frac{-\overline\alpha_A\overline \alpha_B}{ \alpha_A(\alpha_E+\alpha_F-\alpha_E\alpha_F)} & \mathcal{I}^{-1}_{33}(\alpha) &  \frac{-\overline\alpha_E\overline \alpha_F}{ \alpha_F(\alpha_A+\alpha_B-\alpha_A\alpha_B)} & \frac{-\overline\alpha_E\overline \alpha_F}{ \alpha_E(\alpha_A+\alpha_B-\alpha_A\alpha_B)} \\ 0& 0 & \frac{-\overline\alpha_E\overline \alpha_F}{ \alpha_F(\alpha_A+\alpha_B-\alpha_A\alpha_B)} & \frac{\alpha_E\overline \alpha_E}{\alpha_{CD} \alpha_F(\alpha_A+\alpha_B-\alpha_A\alpha_B)} & \frac{\overline\alpha_E\overline \alpha_F}{ \alpha_{CD}(\alpha_A+\alpha_B-\alpha_A\alpha_B)} \\ 0 & 0 & \frac{-\overline\alpha_E\overline \alpha_F}{ \alpha_E(\alpha_A+\alpha_B-\alpha_A\alpha_B)} & \frac{\overline\alpha_E\overline \alpha_F}{ \alpha_{CD}(\alpha_A+\alpha_B-\alpha_A\alpha_B)} & \frac{\alpha_F\overline \alpha_F}{\alpha_{CD}\alpha_E(\alpha_A+\alpha_B-\alpha_A\alpha_B)}
 \end{array} \right)
\end{displaymath}
\begin{displaymath}
\begin{split}
\mathcal{I}^{-1}_{33}(\alpha) & = \frac{1}{\alpha_A\alpha_B\alpha_E\alpha_F(-\alpha_A\overline \alpha_B-\alpha_B)(-\alpha_E\overline \alpha_F-\alpha_F)} (-\alpha_{CD}(-\alpha_B\overline\alpha_B\alpha_E\alpha_F-\alpha_A^2\overline \alpha_B\alpha_E(-1+\alpha_B(2+\alpha_{CD}(-\alpha_E\overline \alpha_F-\alpha_F)))\alpha_F\\
& + \alpha_A(-\alpha_E\alpha_F+\alpha_B^2\alpha_E\alpha_F(-3+\alpha_{CD}(\alpha_E+\alpha_F-\alpha_E\alpha_F))+\alpha_B(-\alpha_F\overline \alpha_F+\alpha_E(-1+7\alpha_F-3\alpha_F^2)+\alpha_E^2(1-3\alpha_F+2\alpha_F^2)))))
\end{split}
\end{displaymath}
\hrulefill
}
\vspace{-0.25em}
\caption{The inverse of the Fisher information matrix governing the confidence intervals for models in Eq.(\ref{eq-confidenceInterval1}). Here, the order of the coordinates is $\alpha_A,\alpha_B,\alpha_{CD},\alpha_E,\alpha_F$.} \label{fig-equation2}
\vspace{-0.75em}
\end{figure*}

\begin{theorem}
\label{thr-MLE}
The estimates computed by Algorithm~\ref{alg-MLE} are the MLEs of the link loss rates in the original tree topology in Fig.~\ref{fig-generalTreeModel}.
\end{theorem}

The proof of Theorem~\ref{thr-MLE} relies on the following two lemmas, whose proofs are provided in Appendix A.3. (Theorem~\ref{thr-MLE} is then proved in Appendix A.4.)
\begin{lemma}
\label{thr-bottomLinks}
The solutions of the likelihood equations of the original tree and the reduced multicast tree are related via: (i) $\hat{\alpha}_k={\hat{\alpha}_k}^m$, $k<D$; and (ii) $\hat{\alpha}_{CD} = \hat{\alpha}^m_{agg}/(1-\prod_{i=1}^{P} \overline \beta^r_{f(C)_i})$.
\end{lemma}
\begin{lemma}
\label{thr-topLinks}
The solutions of the likelihood equations of the original tree and the reduced reverse multicast tree are related via: (i) $\hat{\alpha}_k={\hat{\alpha}_k}^r$, $k>C$; and (ii) $\hat{\alpha}_{CD} = \hat{\alpha}^r_{agg}/(1-\prod_{j=1}^{Q} \overline \beta^m_{d(D)_j})$.
\end{lemma}
We note that the likelihood functions of the original tree and the reduced multicast (or reverse multicast) tree are different. What the aforementioned lemmas establish is that these likelihood functions are maximized for the same values of their common variables.

{\bf Complexity.} Algorithm~\ref{alg-MLE} is very efficient. In the first two steps, it calls MINC and RMINC. MINC (and thus RMINC) is known to be efficient by exploiting the hierarchy of the tree topology to factorize the probability distribution and recursively compute the estimates. The computation at each node is at worst proportional to the depth of the tree \cite{minc}. The last step, $\hat{\alpha}_{CD}$, uses the estimates $\hat{A}_k,\hat{\gamma}_k$ already computed in the first two steps.

{\bf Rate of Convergence of the MLE.} We can provide the rate of convergence of $\hat{\alpha}$ to the true value $\alpha$. The Fisher information matrix at $\alpha$ based on $X_{(R)}$ is obtained from $\mathcal{I}_{jk}(\alpha)=-E\frac{\partial^2\mathcal{L}}{\partial \alpha_j \partial \alpha_k}(\alpha)$ \cite{minc}. We have that:

\begin{theorem}
\label{thr-convergenceRate}
$\mathcal{I}(\alpha)$ is non-singular, and as $n\rightarrow \infty$, $\sqrt{n}(\hat{\alpha}-\alpha)$ converges in distribution to $\mathcal{N}(0,\mathcal{I}^{-1}(\alpha))$.
\end{theorem}

The proof follows from the asymptotic properties of the MLEs \cite{minc, statistics}. Therefore, asymptotically for large $n$, with probability $1-\delta$ (for $1-\delta$ confidence interval), $\hat{\alpha}_k$ lies between the points:\footnote{$z_{\delta / 2}$ denotes the number that cuts off an area $\delta / 2$ in the right tail of the standard normal distribution.}
\begin{equation}
\label{eq-confidenceInterval1}
\alpha_k \pm z_{\delta / 2} \sqrt{\frac{\mathcal{I}_{kk}^{-1}(\alpha)}{n}}
\end{equation}

\begin{example}

We now illustrate our results by revisiting the example 5-link tree topology in Fig.~\ref{figure_1}(c). Note that here, following the notation described in Section~\ref{sec-model}, we use the notation $\alpha_A$, $\alpha_B$, $\alpha_E$, and $\alpha_F$, for the four edge links in Fig.~\ref{figure_1}(c), instead of $\alpha_{AC}$, $\alpha_{BC}$, $\alpha_{DE}$, and $\alpha_{DF}$, respectively, which were used in Example~\ref{ex:ex1}.

{\bf Maximum Likelihood Estimator.} The two source nodes $A$ and $B$ send probe packets $x_1=[1,0]$ and $x_2=[0,1]$, respectively. The space $\Omega$ consists of ten possible outcomes shown in Table~\ref{tabl_1}. Table~\ref{tabl_1} also shows the corresponding outcomes for the reduced multicast and reverse multicast trees. From Eq.(\ref{eq-gamma-r}) and Eq.(\ref{eq-gamma-m}), we have that:
\begin{displaymath}
\hat{\gamma}^r_A =  \hat{p}_1+ \hat{p}_3+ \hat{p}_4+ \hat{p}_6 + \hat{p}_7+ \hat{p}_9 
\end{displaymath}
\begin{displaymath}
\hat{\gamma}^r_B =  \hat{p}_2+ \hat{p}_3+ \hat{p}_5+ \hat{p}_6+ \hat{p}_8+ \hat{p}_9
\end{displaymath}
\begin{displaymath}
\hat{\gamma}^r_C =  \hat{\gamma}^m_D = \hat{p}_1+\hat{p}_2+\hat{p}_3+\hat{p}_4+\hat{p}_5+\hat{p}_6+\hat{p}_7+\hat{p}_8+\hat{p}_9= 1-\hat{p}_0
\end{displaymath}
\begin{displaymath}
\hat{\gamma}^m_E = \hat{p}_1+ \hat{p}_2+ \hat{p}_3+ \hat{p}_7+ \hat{p}_8+ \hat{p}_9 
\end{displaymath}
\begin{displaymath}
\hat{\gamma}^m_F =  \hat{p}_4+ \hat{p}_5+ \hat{p}_6+ \hat{p}_7+ \hat{p}_8+ \hat{p}_9
\end{displaymath}
We then solve Eq.(\ref{eq-MC1}) for $\hat{A}^m_k$ and Eq.(\ref{eq-RMC1}) for $\hat{A}^r_k$, and then we find $\hat{\alpha}_A$ and $\hat{\alpha}_B$ from Eq.(\ref{eq-RMC2}), $\hat{\alpha}_E$ and $\hat{\alpha}_F$ from Eq.(\ref{eq-MC2}), and $\hat{\alpha}_{CD}$ from Eq.(\ref{eq-alphaCDestimate}), as follows:
\begin{equation}
\label{eq-alphaAandB}
\hat{\alpha}_{A} = \frac{\hat{\gamma}^r_A+\hat{\gamma}^r_B-\hat{\gamma}^r_C}{\hat{\gamma}^r_B} \quad , \quad \hat{\alpha}_{B} = \frac{\hat{\gamma}^r_A+\hat{\gamma}^r_B-\hat{\gamma}^r_C}{\hat{\gamma}^r_A}
\end{equation}
\begin{equation}
\label{eq-alphaEandF}
\hat{\alpha}_{E} = \frac{\hat{\gamma}^m_E+\hat{\gamma}^m_F-\hat{\gamma}^m_D}{\hat{\gamma}^m_F} \quad , \quad \hat{\alpha}_{F} = \frac{\hat{\gamma}^m_E+\hat{\gamma}^m_F-\hat{\gamma}^m_D}{\hat{\gamma}^m_E}
\end{equation}
\begin{equation}
\label{eq-alphaCD}
\hat{\alpha}_{CD} = \frac{\hat{\gamma}^r_A \hat{\gamma}^r_B \hat{\gamma}^m_E \hat{\gamma}^m_F}{\hat{\gamma}^m_D (\hat{\gamma}^r_A+\hat{\gamma}^r_B-\hat{\gamma}^r_C)(\hat{\gamma}^m_E+\hat{\gamma}^m_F-\hat{\gamma}^m_D)}
\end{equation}

{\bf Confidence Intervals.} Fig.~\ref{fig-equation2} shows $\mathcal{I}^{-1}(\alpha)$ for the confidence intervals in Eq.(\ref{eq-confidenceInterval1}). We note that the confidence intervals for parameters $\hat{\alpha}$ can be obtained by inserting Eq.(\ref{eq-alphaAandB}), Eq.(\ref{eq-alphaEandF}), and Eq.(\ref{eq-alphaCD}) into Fig.~\ref{fig-equation2}.
\hfill{$\square$}
\end{example}

\subsection{MLE of a Single Link} \label{sec-estimation-onelink}

Section~\ref{sec-estimation1} provides a computationally efficient way to estimate all link loss rates at the same time, under the mild assumption that the tree is of the form depicted in Fig.~\ref{fig-generalTreeModel}. If one is allowed to pick the sources and the receivers in the tree, then one can ensure that this mild assumption holds. 

However, there are practical scenarios where one might not want to or might not be able to use this scheme. First, if we are not allowed  to choose the sources, \eg due to practical constraints, it is possible that the monitoring scheme does not have the desired property of Fig.~\ref{fig-generalTreeModel}, \ie all coding points may not be above all branching points. An example is Case 3 in the 5-link topology of Fig.~\ref{fig_basic}: all links are still identifiable, but the assumption does not hold and the MLE provided in the previous Section does not apply. Second, we are often not even interested in estimating the loss rates for all links; it is common that only one or a few bottleneck/congested links are of interest. In general topologies, focusing on a few, as opposed to all, links has the side benefit that we may not need to deal with cycles, if they do not appear in the paths that go through the links of interest.

In all these cases, we propose that one estimates the loss rate of one link at a time. Recall the discussion in Section \ref{sec-identify-one-link}. The conditions for identifiability of a link (say link $CD$ in Theorem \ref{theorem_1}) still apply, while the other four links  $AC$, $BC$, $DE$, and $DF$ in the 5-link topology can be interpreted as paths from/to the sources/receivers; \ie we do not care about the individual link loss rates on these paths. Depending on the constraints on the selection of sources, any of the 4 cases in the 5-link topology of Fig.~\ref{fig_basic} may be possible. We note that Table I and Algorithm~\ref{alg-MLE} correspond to Case 1 in the 5-link topology. Tables for the other three cases are provided in Appendix B.1.

In fact, similar MLE algorithms can be provided for all other 3 cases. For example, MINC and RMINC can be used for Cases 2 and 4 directly. Only Case 3 needs to be estimated similarly to Case 1 using reductions and Table VI. For Case 3, the reduced multicast tree will consist of $AC$, $CB$, $CD$ and $DF$. We use MINC on this tree to infer the loss rates $\overline \alpha_{AC}$ and $\overline \alpha_{CB}$. The reduced reverse multicast tree will consist of $AC$, $CD$, $ED$ and $DF$. We use RMINC on this tree to infer the loss rates $\overline \alpha_{ED}$ and $\overline \alpha_{DF}$. We can then replace these results in the likelihood function and find $\overline \alpha_{CD}$ by maximizing it. In general, an algorithm similar to Alg.~\ref{alg-MLE} can be developed to compute the MLE for the single link of interest: we first compute  MINC on the reduced multicast tree, then RMINC on the reduced reverse multicast tree, and then we estimate link $CD$ using a similar procedure as in Appendix A.2.

{\em Remarks:} Note that even when we focus on estimating a single link, the brute force approach appears to be computationally demanding even though it involves only 5 variables. Therefore, the efficient computation of the MLE for a single link is an important contribution on its own.  

\begin{algorithm}[t!]
\caption{\label{alg-subtree} Subtree Decomposition Algorithm:}
{\footnotesize
Consider a tree $G$, with sources $S$ and receivers $R$. Each source sends one probe packet. Each receiver receives at most one probe packet.
\begin{itemize}
\item Determine the coding points. These  partition $G$ into $|\mathcal{T}|\leq 2M-1$ subtrees.
\item  For each of the $|\mathcal{T}|$ subtrees:
\begin{itemize}
\item {\em If the multicast tree is rooted at a coding point:}
\begin{itemize}
\item if any of the descendant receivers receives a probe,
use this experiment as a measurement on the subtree.
\item otherwise, w.p. $p$ assume no node in $R$ received a probe packet, and w.p. $(1-p)$ ignore the experiment.
\end{itemize}
\item {\em If the multicast tree is rooted at a source $S_i$:}\\
Consider each coding point $C$ that acts as a receiver:
\begin{itemize}
\item if no descendant receivers $C(R)$ observed a
probe, assume,  w.p. $p$, that $C$ received a packet, and w.p. $(1-p)$, that it did not.
\item otherwise
\begin{itemize}
\item if at least one of $C(R)$ observed a linear combination of $x_i$,
deduce that  $C$ received $x_i$.
\end{itemize}
\end{itemize}
\end{itemize}
\end{itemize}
}
\end{algorithm}

\subsection{Heuristic Approaches for Loss Estimation} \label{sec-estimation2}

Beyond tree topologies, there is no known computationally efficient algorithm to compute the MLE of {\em all} link loss rates. In this Section,  we propose three heuristic estimation algorithms and evaluate their performance through simulation. The first two (subtree decomposition and MINC-like heuristic, in Sections \ref{sec:subtree} and \ref{sec:minc-like}, respectively) are specific to trees, while the third one (belief propagation, in Section \ref{sec:trees-bp}) applies also to general graphs.

\subsubsection{\label{sec:subtree}{\bf Subtree Decomposition}}

Algorithm~\ref{alg-subtree} partitions the tree into multicast subtrees separated by
coding points. Each coding point virtually acts as a receiver for
incoming flows and as a source for outgoing flows. As a result, each
subtree will either have a coding point as its source, or will have
at least one coding point as a receiver. In each subtree, we can
then use the ML estimator (MINC) proposed in \cite{minc}.

Note that we can only observe packets received at the edge of the
network, but not at the coding points. However, we can still infer
that information from the observations at the receivers downstream
from the coding point. The fact that we infer observations of the coding-points
from the observations of the leaves is what makes this
algorithm suboptimal, while MINC in each partition is optimal.

We introduce the probability $p$ in order to account for the fact that
if none of the receivers in $C(R)$ receives a packet,
this might be attributed to two distinct events: either
the coding point $C$ itself did not receive a packet,
or $C$ did receive a packet, which got subsequently lost
in the descendent edges. For example, in Fig.~$\ref{fig_ninelinks}$, consider the tree rooted
at $S_1$; if $R_2$ receives $x_1$ or $x_1+x_2$, we deduce that $x_1$
was received at node $4$. If $R_2$ receives $x_2$, we deduce that
$x_1$ was not received at node $4$. If $R_2$ does not receive a
probe packet, then, with probability $1-p$, we assume that $4$ did not
receive a probe packet. Ideally, $p$ should match the probability that $C$ correctly received a probe packet. This depends on the graph structure and on the loss probabilities downstream of $C$, and possibly prior information we may have about the link loss rates.

\begin{figure}[t!]
\centering
\vspace{-1.0em}
\psset{unit=0.038in}
\begin{pspicture}(5,5)(53,48)
\begin{small}
\rput(5,5){\circlenode[linecolor=red]{N10}{10}}\rput(-2,5){\large{$R_4$}}
\rput(5,25){\circlenode[linecolor=red]{N9}{9}}\rput(-1,25){\large{$R_3$}}
\rput(13,15){\circlenode[linecolor=red]{N6}{6}}
\rput(28,15){\circlenode[linecolor=red]{N5}{5}}
\rput(38,5){\circlenode[linecolor=red]{N2}{2}}
\rput(44,5){\large $S_2$}
\rput(28,30){\circlenode[linecolor=green]{N4}{4}}
\rput(18,40){\circlenode[linecolor=green]{N8}{8}}\rput(12,40){\large{$R_2$}}
\rput(43,30){\circlenode[linecolor=blue]{N3}{3}}
\rput(53,20){\circlenode[linecolor=blue]{N1}{1}}
\rput(59,20){\large $S_1$}
\rput(53,40){\circlenode[linecolor=blue]{N7}{7}}
\rput(59,40){\large{$R_1$}}
\ncline[linecolor=red]{<-}{N9}{N6}
\Bput{$\alpha_6$}
\ncline[linecolor=red]{<-}{N10}{N6}
\Bput{$\alpha_5$}
\ncline[linecolor=red]{<-}{N6}{N5}
\Aput{$\alpha_9$}
\ncline[linecolor=red]{->}{N5}{N4}
\Bput{$\alpha_8$}
\ncline[linecolor=green]{->}{N4}{N8}\Aput{$\alpha_1$}
\ncline[linecolor=blue]{<-}{N4}{N3}\Aput{$\alpha_7$}
\ncline[linecolor=blue]{->}{N3}{N7}\Aput{$\alpha_2$}
\ncline[linecolor=blue]{<-}{N3}{N1}\Aput{$\alpha_3$}
\ncline[linecolor=red]{<-}{N5}{N2}\Aput{$\alpha_4$}
\end{small}
\end{pspicture}
\caption{A network topology with $9$ links. The link orientation depicted corresponds to nodes $1$ and $2$ acting as sources of probes.} 
\label{fig_ninelinks}
\end{figure}
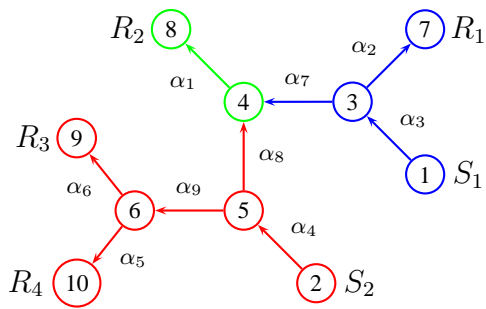

\subsubsection{\label{sec:minc-like}{\bf MINC-like Heuristic}}

For every multicast node, we can use the MINC algorithm described in \cite{minc}. For every
coding point, we can use RMINC described in Section~\ref{sec-MINCandRMINC}. 

Similarly to the subtree decomposition, we infer which probes have been received by an interior node $i$ from observations at the downstream receivers. In particular, if at least one receiver downstream of $i$ has received a probe with any content (the probe is from at least one source and potentially contains the \texttt{XOR} of probes from multiple sources), then we can infer that $i$ received the packet. This can be used to compute the probability $\gamma_i$, in the terminology of MINC \cite{minc}. If no downstream receiver got any probe, we decide w.p. $p$ whether the node $i$ received a probe or not, exactly the same as in the subtree decomposition. The reductions shown in Fig.~\ref{fig:reductions} use similar arguments and can serve as examples. 

Different from the subtree decomposition, which estimates the $\alpha$'s locally in each subtree, we use the mapping from $\gamma$'s to $\alpha$'s provided in MINC \cite{minc} to estimate the $\alpha$'s in the entire graph. This heuristic is optimal for multicast and reverse multicast configurations, and for configurations that are concatenations of the two, but suboptimal for any other configuration.

\subsubsection{\label{sec:trees-bp} {\bf Belief Propagation}}

We propose to use a Belief Propagation (BP) approach, similar to what was proposed in \cite{belief}.
Unlike the previous two heuristics, which are specific to tree topologies, the BP approach also applies to general graphs.
The first step in the BP approach is to create the factor graph corresponding to our estimation problem.
Fig.~\ref{fig_bipartite} shows the factor graph corresponding to the 9-link tree shown in Fig. \ref{fig_ninelinks}.
This is a bipartite graph: on one side there are the links (variable nodes),
whose loss rates we want to estimate; on the other side there are
the paths (function nodes) that are observed by each received probe.
An edge exists in the factor graph between a link and a path, if the
link belongs to this path in the original graph. Note that in tree
topologies, there exists exactly one path for every source-receiver pair;
while this is not the case in general graphs. Once the factor graph is created from the original graph, each
received probe triggers message passing and results in an estimate
of link success probabilities; these estimates from different probes
are then combined using standard methods \cite{belief}. The result
is an estimate ($\hat{\alpha}_e$) of the actual success probability
($\alpha_e$) of every link $e \in E$.

\subsection{\label{sec:trees-simulations}Simulation Results}

In this Section, we evaluate the heuristic estimators via simulation and we compare
them to each other as well as to multicast-based tomography. The main finding is that using more than one source helps: using multiple sources and network coding (even with suboptimal estimation) outperforms a single multicast tree (even with optimal estimation), thus demonstrating the usefulness of our approach.\footnote{Note that using more than one multicast sources, without network coding, would traditionally require to combine the observations from the two trees in a suboptimal way \cite{general}, thus further degrading the performance; that is why we skip the comparison and compare only against a single multicast tree and optimal estimation, which has the best performance among the baselines.}

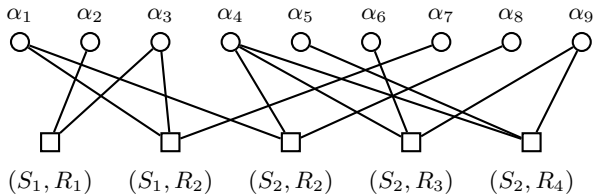
\begin{figure}[t!]
\begin{center}
\psset{unit=0.035in}
\begin{pspicture}(6,-5)(89,18)
\begin{small}
\rput(10,0){\rnode{W1}{\psframebox[]{}}}
\rput(10,-6){$(S_1,R_1)$}
\rput(28,0){\rnode{W2}{\psframebox[]{}}}
\rput(28,-6){$(S_1,R_2)$}
\rput(46,0){\rnode{W3}{\psframebox[]{}}}
\rput(46,-6){$(S_2,R_2)$}
\rput(64,0){\rnode{W4}{\psframebox[]{}}}
\rput(64,-6){$(S_2,R_3)$}
\rput(82,0){\rnode{W5}{\psframebox[]{}}}
\rput(82,-6){$(S_2,R_4)$}
\rput(5.5,15){\circlenode{N1}{}}
\rput(5.5,19){$\alpha_1$}
\rput(16,15){\circlenode{N2}{}}\rput(16,19){$\alpha_2$}
\rput(26.5,15){\circlenode{N3}{}}\rput(26.5,19){$\alpha_3$}
\rput(37,15){\circlenode{N4}{}}\rput(37,19){$\alpha_4$}
\rput(47.5,15){\circlenode{N5}{}}\rput(47.5,19){$\alpha_5$}
\rput(58,15){\circlenode{N6}{}}\rput(58.5,19){$\alpha_6$}
\rput(68.5,15){\circlenode{N7}{}}\rput(68.5,19){$\alpha_7$}
\rput(79,15){\circlenode{N8}{}}\rput(79,19){$\alpha_8$}
\rput(89.5,15){\circlenode{N9}{}}\rput(89.5,19){$\alpha_9$}
\ncline[]{-}{W1}{N2}
\ncline[]{-}{W1}{N3}
\ncline[]{-}{W2}{N3}
\ncline[]{-}{W2}{N7}
\ncline[]{-}{W2}{N1}
\ncline[]{-}{W3}{N4}
\ncline[]{-}{W3}{N8}
\ncline[]{-}{W3}{N1}
\ncline[]{-}{W4}{N4}
\ncline[]{-}{W4}{N9}
\ncline[]{-}{W4}{N6}
\ncline[]{-}{W5}{N4}
\ncline[]{-}{W5}{N9}
\ncline[]{-}{W5}{N5}
\end{small}
\end{pspicture}
\end{center} 
\vspace{-0.25em}
\caption{Bipartite graph corresponding to the 9-link example
tree in Fig.~\ref{fig_ninelinks}. It indicates which edges belong to which
observable paths.} \label{fig_bipartite}
\vspace{-0.25em}
\end{figure}

Consider the 45-link topology shown in Fig.~\ref{fig_45links}, where all links have the same success rate $\alpha$. We will estimate $\alpha$ and compare different methods in terms of their estimation accuracy. First, we did simulations for $\alpha=0.7$, a large number of probes, and repeated for many experiments. We looked at the mean square error ($MSE$)  at each link. The results are shown in Fig.~\ref{fig_all45links} for the following three algorithms:
\begin{itemize}
\item a single multicast source $S_1$ and maximum likelihood estimation (top
plot).
\item two sources $S_1,S_2$, network coding at the
middle node $C$, and the MINC-like heuristic (middle plot).
\item the same two sources and coding point, with the subtree
estimation  algorithm (bottom plot).
\end{itemize}
Notice that in the case of two sources, the 45-link topology is partitioned into 3 subtrees: one rooted at $A$ (where probe $x_1$ flows), another rooted at $D$ (where $x_2$ flows) and a third one rooted at $B$ (where $x_1+x_2$ flows).

One can make several observations from this graph. First, using two sources and network coding, even with suboptimal estimators, performs better than using a single multicast source and an ML estimator. Indeed, the residual entropy (which is the metric that summarizes the $MSE$ across all 45 links) is lower for two sources with the MINC-like ($ENT=-317.9$) and for the subtree-decomposition ($ENT=-314.9$) heuristics, than it is for the single source MLE ($ENT=-294.5$). This illustrates the benefit of using multiple sources.  Second, notice that the $MSE$ for individual links is smaller in the lower two graphs than in the top graph, for all links except for links $43$, $44$, $45$, for which it is significantly higher. This is no coincidence: links $43$, $44$, $45$ are the middle ones (CA, CB, CD in Fig.~\ref{fig_45links}). This is due to the fact that we cannot directly observe the packets received at the coding point C and we have to infer them from observations at the leaves of the subtree rooted at B. The performance of the heuristics could further improve by using the following tweak: we could estimate what probes are received at C, using observations from leaves not only in the subtree rooted at B, but also from the subtrees rooted at A and D.

The above simulations were for a single value of $\alpha=0.7$. We then exhaustively considered several values of $\alpha$ (same on all links) and $n$ (the number of probes). The results are shown in Fig.~\ref{fig_45links_vs_all}. We can see that, even with suboptimal estimation, using two sources consistently outperforms a single multicast source, even with MLE estimation. This is apparent in Fig.~\ref{fig_45links_vs_all}, where the $ENT$ metric for the single source (drawn in bold lines) is consistently above the other two algorithms.\footnote{Two observations on the $ENT$ metric: First, the differences in the value of $ENT$ are significant, although this is not visually obvious; recall that $ENT$ is defined by taking the sum of the $log$ of the $MSE$'s. Second, $ENT$ can be $<0$, it is the differential entropy that matters.}

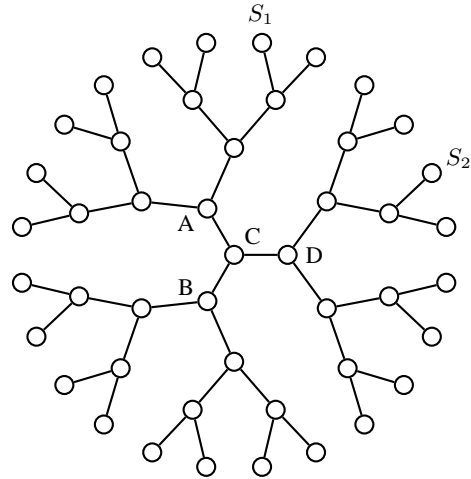
\begin{figure}
\vspace{1.0em}
\begin{center}
\psset{unit=0.028in}
\begin{pspicture}(-34,-40)(35,36)
\begin{small}
\rput(3.5,3.5){C}
\rput(0,0){\circlenode{C}{}}
\rput(-5,8.6603){\circlenode{A}{}}\rput(-9,6){A}
\rput(-5,-8.6603){\circlenode{B}{}}\rput(-9,-6){B}
\rput(10,0){\circlenode{D}{}}\rput(15,0){D}
\rput(0,-20){\circlenode{C1}{}}
\rput(-17.3205,-10.0){\circlenode{C2}{}}
\rput(-17.3205,10){\circlenode{C3}{}}
\rput(0,20){\circlenode{C4}{}}
\rput(17.3205,10){\circlenode{C5}{}}
\rput(17.3205,-10){\circlenode{C6}{}}
\rput(7.7646,-28.9778){\circlenode{D1}{}}
\rput(-7.7646,-28.9778){\circlenode{D2}{}}
\rput(-21.2132, -21.2132){\circlenode{D3}{}}
\rput(-28.9778, -7.7646){\circlenode{D4}{}}
\rput(-28.9778, 7.7646){\circlenode{D5}{}}
\rput(-21.2132,21.2132){\circlenode{D6}{}}
\rput(-7.7646,28.9778){\circlenode{D7}{}}
\rput(7.7646,28.9778){\circlenode{D8}{}}
\rput(21.2132,21.2132){\circlenode{D9}{}}
\rput(28.9778,7.7646){\circlenode{D10}{}}
\rput(28.9778, -7.7646){\circlenode{D11}{}}
\rput(21.2132,-21.2132){\circlenode{D12}{}}
\rput(15.3073, -36.9552){\circlenode{F1}{}}
\rput(5.2210, -39.6578){\circlenode{F2}{}}
\rput(-5.2210, -39.6578){\circlenode{F3}{}}
\rput(-15.3073, -36.9552){\circlenode{F4}{}}
\rput(-24.3505, -31.7341){\circlenode{F5}{}}
\rput(-31.7341, -24.3505){\circlenode{F6}{}}
\rput(-36.9552, -15.3073){\circlenode{F7}{}}
\rput(-39.6578, -5.2210){\circlenode{F8}{}}
\rput(-39.6578, +5.2210){\circlenode{F9}{}}
\rput(-36.9552, +15.3073){\circlenode{F10}{}}
\rput(-31.7341, +24.3505){\circlenode{F11}{}}
\rput(-24.3505, +31.7341){\circlenode{F12}{}}
\rput(-15.3073, +36.9552){\circlenode{F13}{}}
\rput(-5.2210,+39.6578){\circlenode{F14}{}}
\rput(5.2210, +39.6578){\circlenode{F15}{}}\rput(5,45){$S_1$}
\rput(15.3073, +36.9552){\circlenode{F16}{}}
\rput(24.3505, +31.7341){\circlenode{F17}{}}
\rput(31.7341, +24.3505){\circlenode{F18}{}}\rput(42,18){$S_2$}
\rput(36.9552, +15.3073){\circlenode{F19}{}}
\rput(39.6578,  5.22){\circlenode{F20}{}}
\rput(39.6578, -5.2210){\circlenode{F21}{}}
\rput(36.9552, -15.3073){\circlenode{F22}{}}
\rput(31.7341, -24.3505){\circlenode{F23}{}}
\rput(24.3505, -31.7341){\circlenode{F24}{}}
\ncline[]{-}{C}{A}
\ncline[]{-}{C}{B}
\ncline[]{-}{C}{D}
\ncline[]{-}{B}{C1}
\ncline[]{-}{B}{C2}
\ncline[]{-}{A}{C3}
\ncline[]{-}{A}{C4}
\ncline[]{-}{D}{C5}
\ncline[]{-}{D}{C6}
\ncline[]{-}{C1}{D1}
\ncline[]{-}{C1}{D2}
\ncline[]{-}{C2}{D3}
\ncline[]{-}{C2}{D4}
\ncline[]{-}{C3}{D5}
\ncline[]{-}{C3}{D6}
\ncline[]{-}{C4}{D7}
\ncline[]{-}{C4}{D8}
\ncline[]{-}{C5}{D9}
\ncline[]{-}{C5}{D10}
\ncline[]{-}{C6}{D11}
\ncline[]{-}{C6}{D12}
\ncline[]{-}{D1}{F1}
\ncline[]{-}{D1}{F2}
\ncline[]{-}{D2}{F3}
\ncline[]{-}{D2}{F4}
\ncline[]{-}{D3}{F5}
\ncline[]{-}{D3}{F6}
\ncline[]{-}{D4}{F7}
\ncline[]{-}{D4}{F8}
\ncline[]{-}{D5}{F9}
\ncline[]{-}{D5}{F10}
\ncline[]{-}{D6}{F11}
\ncline[]{-}{D6}{F12}
\ncline[]{-}{D7}{F13}
\ncline[]{-}{D7}{F14}
\ncline[]{-}{D8}{F15}
\ncline[]{-}{D8}{F16}
\ncline[]{-}{D9}{F17}
\ncline[]{-}{D9}{F18}
\ncline[]{-}{D10}{F19}
\ncline[]{-}{D10}{F20}
\ncline[]{-}{D11}{F21}
\ncline[]{-}{D11}{F22}
\ncline[]{-}{D12}{F23}
\ncline[]{-}{D12}{F24}
\end{small}
\end{pspicture}
\end{center} 
\vspace{-0.5em}
\caption{A tree with $45$ links used for simulating the suboptimal estimators.} \label{fig_45links}
\vspace{-0.5em}
\end{figure}

\begin{figure}[t!]
\centering
\includegraphics[scale=0.6]{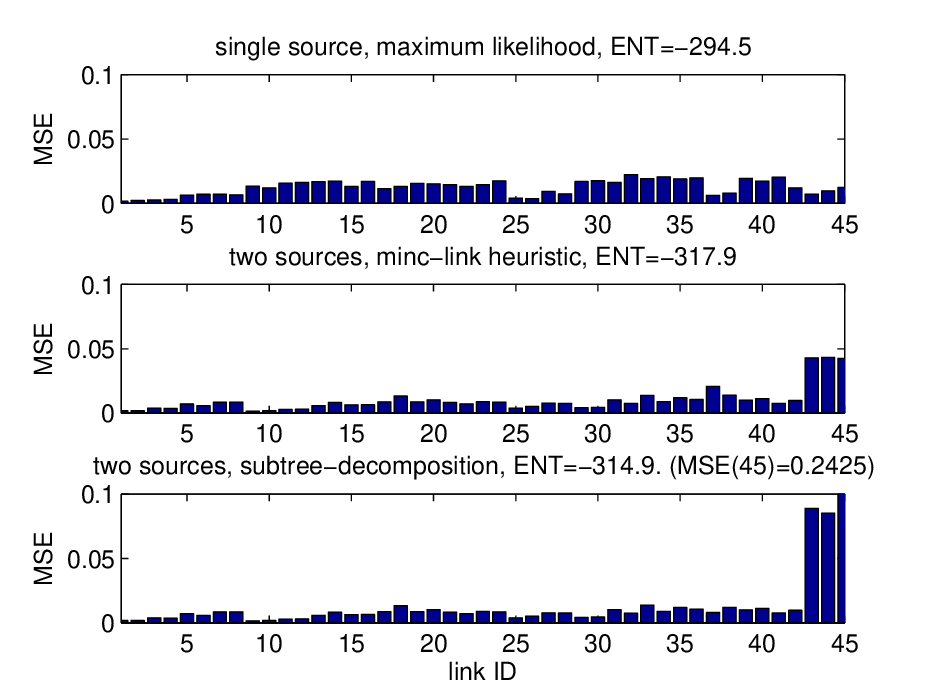}
\vspace{-2.0em}
\caption{Comparison of one multicast source + MLE vs. two sources + network coding + suboptimal estimation (subtree decomposition and MINC-like heuristic). We show the $MSE$ for each link in the 45-link topology.}
\label{fig_all45links}
\end{figure}

\begin{figure}
\centering
\subfigure[ENT vs. number of probes]{
  \includegraphics[width=2.5in]{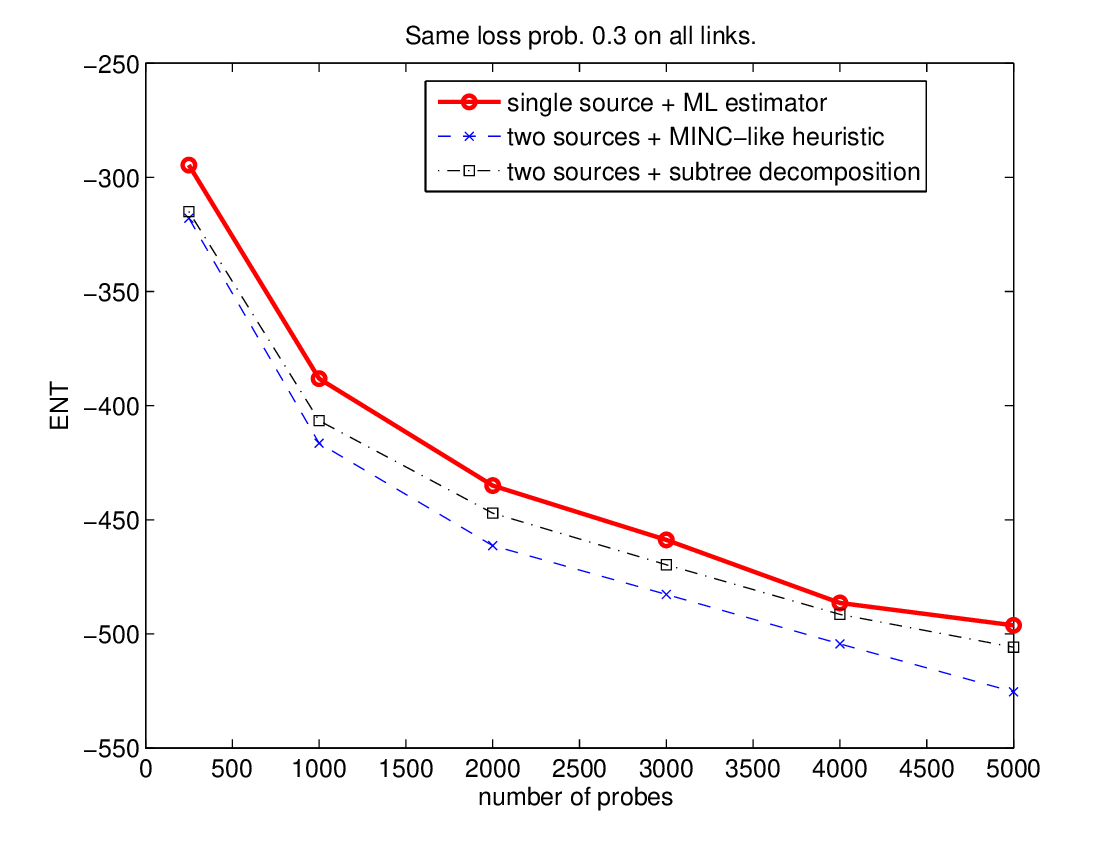}}
\subfigure[ENT vs. loss probability (same on all links)]{
\includegraphics[width=2.5in]{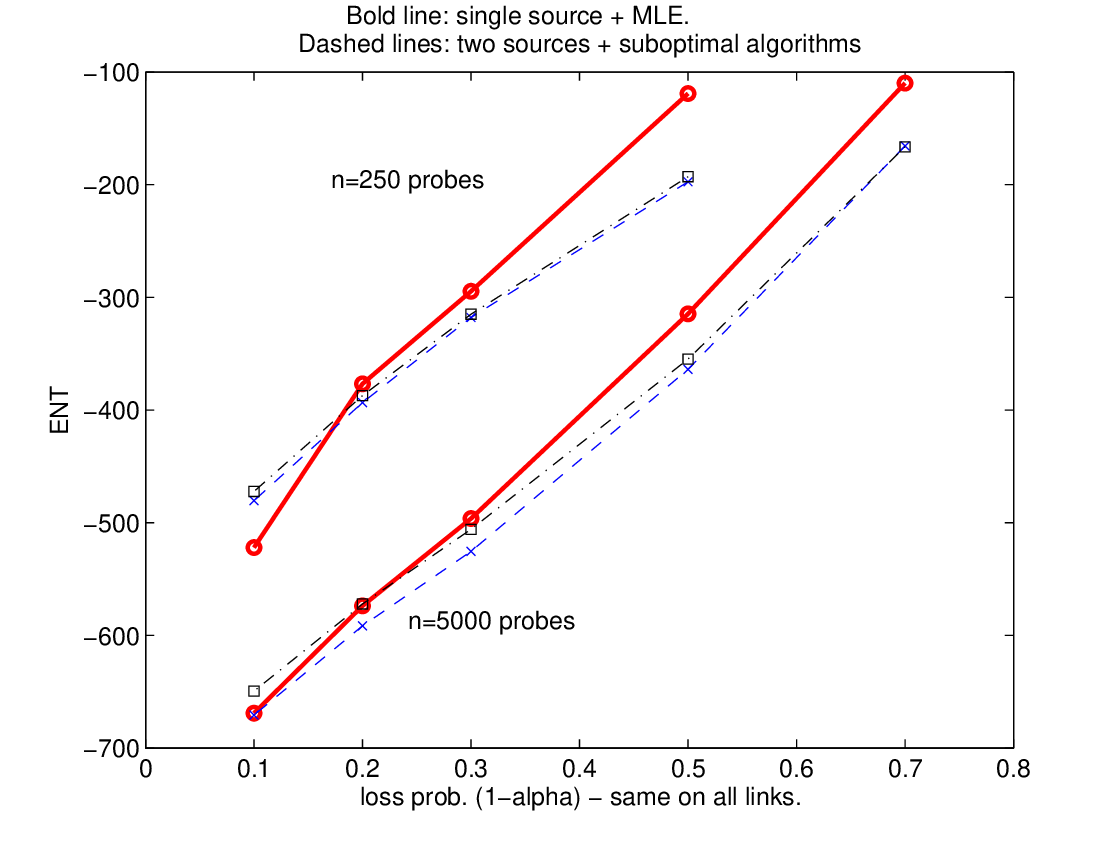}}
\caption{Comparison of one source with MLE, to two sources with  suboptimal estimation (MINC-like  and subtree estimation algorithms) for the 45-link tree. The comparison summarizes the error $ENT$ over all links.}
\label{fig_45links_vs_all} 
\end{figure}

In Fig.~\ref{figa}, we compare the MINC-like and the BP algorithms over the 45-link network,
in terms of the ENT measure, and as a function of the number of probes $n$. Both algorithms yield better performance (lower ENT values) as the number of sources increases from one to five.
The MINC-like algorithm performs better for  the multicast tree, in which case it coincides with the ML estimator, as well as for the two source tree. However, belief propagation offers significantly better performance for the case of three and five sources. This trend can be explained by looking at the number of cycles in the factor graph. A cycle is created in the factor graph of a network configuration when (1) two different paths have more than one link in common and (2) a set of $m$ paths, say $W_m$, covers a set $E_m$ of $m$ links, with each of the paths in $W_m$ containing at least two links in $E_m$. As the factor graph becomes more and more cyclic, the performance of the sum-product algorithm degrades.

Finally, in Fig.~\ref{figd}, we compare the performance of belief propagation to ML estimation using a single source. We considered two trees: the 45-link and another, randomly generated 200-link, tree. Because $ENT$ captures the error over all links, and the two topologies have different numbers of links, we use $ENT_{av}$ (defined as the $ENT$ value divided by the number of network links) for a fair comparison of the two topologies. $ENT_{av}$ for the 45-link tree is better (lower) than that of the 200-link tree for a given number of probes. We see that the BP algorithm closely follows the optimal ML estimator, for the range of number of probes and for both trees considered.

\section{\label{sec:general}General Topologies}

In this section, we extend our approach from trees to general topologies. The difference in the second case is the presence of cycles, which poses two challenges: (i) probes may meet more than once and (ii) probes may be trapped in loops. To deal with these challenges, in this section, we propose (i) an orientation algorithm for undirected graphs and (ii) probe coding schemes, whose design is more involved than in trees.

The approach followed by prior work on tomography over general networks was to cover the graph
with several multicast \cite{general} and/or unicast probes \cite{tomo-unicast-2, nowak}. This approach faces several challenges. (a) The selection of multicast/unicast probes so as to minimize the total bandwidth (cost) is an NP-complete problem. (b) Having several probes from different source-destination paths cross the same link leads to bandwidth waste (especially close to sources or receivers). (c) Finding an optimal and/or practical method to combine the observations from different multicast/unicast paths is a non-trivial problem, addressed in a suboptimal way \cite{general}.

\begin{figure}
\centering
\includegraphics[angle=-90,width=2.5in]{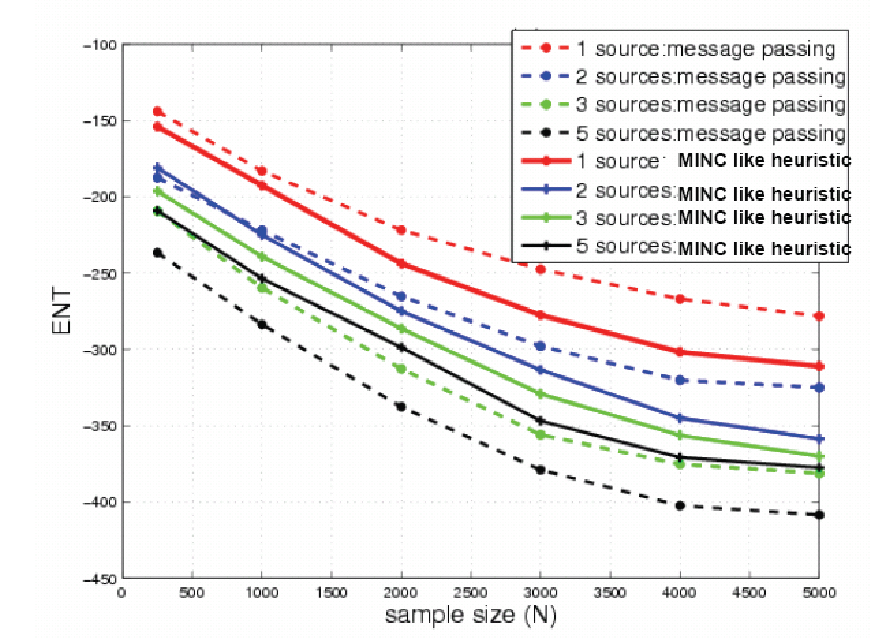}
\caption{Estimation error for two suboptimal algorithms (BP and MINC-like) for the 45-link tree. $ENT$ vs. number of probes.} 
\vspace{-0.25em}
\label{figa}
\end{figure}

\begin{figure}
\centering
\includegraphics[angle=0,width=2.5in]{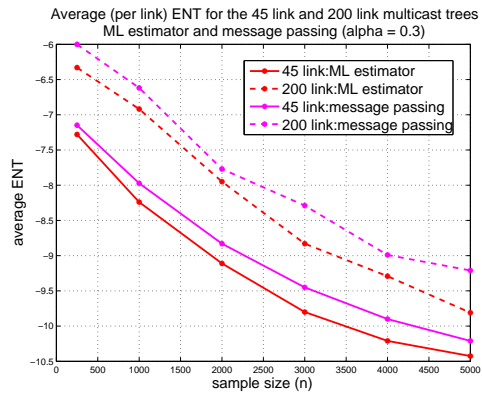}
\vspace{-0.25em}
\caption{Comparing BP to MLE for the 45-link and 200-link trees. $ENT_{av}$ is  $ENT$ divided by the number of links.}
\vspace{-0.5em}
\label{figd}
\end{figure}

In contrast, using network coding allows us to measure all links with a single probe per link
and brings the following benefits: (a) It makes the selection of routes so as to minimize the cost of linear complexity. (b) It eliminates the waste of bandwidth by having each link traversed by exactly one probe per experiment; furthermore, each network coded probe brings more information, as it observes several paths at the same time. (c) It does not need to combine observations from different experiments for estimation (as all links in the network are probed exactly once in one pass/experiment).

Because of the aforementioned features, the benefits of the network coding approach compared to traditional tomographic approaches are even more pronounced in general topologies than they were in tree topologies.

In this Section, we describe the framework for link loss tomography in general graphs.
 In particular, we address the four subproblems mentioned in Section \ref{sec:decomposition}: (1) identifiability of links (2) how to select the routing (3) how to perform the code design, and (4) what estimation algorithms to use. We evaluate our approach through extensive simulations on two realistic topologies: a small research network (Abilene), used to illustrate the ideas; and a large commercial ISP topology (Exodus), used to evaluate the performance in large graphs.

\subsection{\label{sec:general-identifiability}Identifiability}

The identifiability of an edge  given a fixed monitoring scheme follows  from Theorem \ref{theorem_1} in Section \ref{sec-identify-one-link}. CD is the edge we would like to identify, and we interpret the edges AC, BC,  DE and DF as paths that connect CD to sources and destinations. In particular, we are  able to identify the link loss rate of edge CD from the probes collected at the receivers, if we can reconstruct the table associated with one of the cases in Fig. \ref{fig_basic} (all tables are provided for completeness in Appendix B.1).

In a general topology,  it is desirable to be able to know the state of all paths, $\{\mathcal{P}\}$, that connect the sources to all receivers, at the end of each experiment. Let $\mathcal{P}(e)$ denote the set of paths that are routed from a source to a receiver, and employ an edge $e$. We refer to {\em path identifiability} as the ability to uniquely map each possible observation (received probes at all receivers) to the state of the paths $\{\mathcal{P}\}$, \ie which paths operated and which failed during the experiment. For a formal definition, please see Eq.(\ref{eq:constraints}) and the related discussion. From the state of the paths, we can tell which links worked (w.p. 1)  and which likely failed (with the associated probability). Moreover, knowing the state of the paths is particularly well suited for running the belief propagation algorithm that we use  for estimation of general graphs: indeed, message-passing in the BP algorithm is triggered by giving the state of the paths as input. Therefore, we will attempt to make the maximum number of path states distinguishable, by appropriate selection of coding. The following example indicates how the selection of a coding scheme can allow more or less path states to be distinguishable at a receiver.

\begin{figure}
\centering
\includegraphics[scale=0.6]{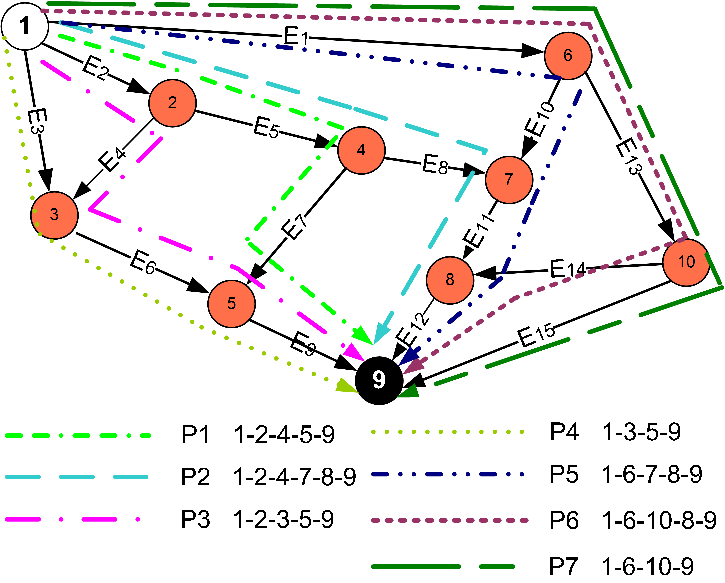}
\caption{Example of a general topology (Abilene). For one source
(node 1), we show the orientation of edges, the resulting receiver
(node 9) and the possible paths from the source to the receiver ($P_1, ...P_7$).}
\label{fig-abilene-photo} 
\end{figure}

\begin{example}\label{ex3}
Consider the network and edge orientation shown in Fig.~\ref{fig-abilene-photo}; this is based on a real backbone topology (Abilene \cite{abilene}), as will be discussed in detail in a later Section.
Node 1 acts as a source and node 9 as a receiver; assume that all intermediate nodes are only allowed to do \texttt{XOR} operations.

Note that paths $P_3$ and $P_1$ overlap twice: on edge $E_2$, and later on edge $E_9$. If all links in both paths  function, the \texttt{XOR} operations ``cancel'' each other out, resulting in  exactly the same observation with both paths being disrupted. More specifically, the following two events become indistinguishable: (i) all edges function:   node $5$ receives packet $x_2$ through
edge $E_7$ and packet $x_2+x_3$ through edge $E_6$, and  sends packet $x_3$ through edge $E_9$ to the receiver; (ii) edges  $E_4$ and $E_7$ fail, while all other edges function: node $5$ only receives packet $x_3$ from its incoming links, and again sends packet $x_3$ through edge $E_9$ to the receiver. On the other hand, if we allow coding operations over a larger alphabet, as in Example~\ref{ex1}, these two events result in observing the distinct packets (i) $3x_2+x_3$ and (ii) $x_3$ at the receiver.
\hfill{$\square$}
\end{example}

\subsection{Routing}
First, we discuss the case where we want to estimate the success rate associated with a specific subset of links, and we express the corresponding optimization problem as a LP that can be solved in polynomial time. Then, we examine the practical special case where we are interested in measuring all links, and which will be the main focus of the rest of the Section.

\subsubsection{\label{sec_cost} {\bf Minimum Cost Routing}}

Consider an arbitrary network topology, a given set $S$ of nodes  that can act as sources, a given set $R$ of nodes  that can act as receivers, and a set $I$ of edges whose link success rates we want to estimate.  Our goal is to estimate the success probability for all links in $I$ at the minimum bandwidth cost. That is, we assume that a cost $C(e)$ is associated with each edge $e$, that is proportional to the flow through the edge. We are interested in identifying the success rate $\alpha_e$ of edge $e\in I$. Let the $\rho$ be the rate of probes crossing that edge, in a manner consistent with the identifiability conditions for edge $e$.

{\em Remarks.} We note that the flow-based formulation of this problem does not rely on any major assumption. The accuracy of estimation depends only on the number of probes and not on the rate of the probe flows. The rates determine how quickly those n packets will be collected. {\em E.g.,} for smaller rates, it will take longer to collect the n packets. We also note that having flows coded together in an edge does not reduce the estimation accuracy. In fact, a coded packet observes more than one path, thus increasing the estimation accuracy vs. bandwidth tradeoff.

The minimum cost routing problem was shown to be NP-hard, when performing tomography with multicast trees \cite{layout}. Indeed, the problem of even finding a single minimum cost Steiner tree is NP-hard. In contrast, we show here that if we use  network coding, we can find the minimum cost routing in polynomial time. In the case of network coding, to ensure identifiability, we want to route flows so that the conditions in Theorem \ref{theorem_1} are satisfied. We will consider the flow-interpretation of paths in Theorem \ref{theorem_1}, {\em i.e.}, we will think of each path as a flow of fixed rate $\rho$. To ensure minimum cost, we want these flows to use the minimum resources possible.

Below we provide a Linear Programming (LP) formulation that allows to solve the minimum cost cover problem in polynomial time, provided that we allow intermediate nodes to combine probes.
We assume that there are no capacity constraints on the edges of the network, {\em i.e.,}
we can utilize each edge as much as we want. This is a realistic assumption, since the rate $\rho$ at which we send probe packets would be chosen to be a very small fraction of the network capacity, and nowhere close to consuming the whole capacity.

{\em Intuition.} Following an approach similar to \cite{Li_optimal}, we introduce conceptual flows that can share a link without contending for the link capacity. We associate with each edge
$e_i \in I$ one such conceptual flow $f^{i}$. We want each $f^{i}$ to bring probe packets to link $e_i=u_iv_i\in I$, in a manner consistent with the conditions of Theorem~\ref{theorem_1} for edge $e_i$. We  allow conceptual flows corresponding to different edges $e_i$ to share edges of the graph without contention, and will measure through a total flow $f$ the utilization of edges by probe packets. We use the condition $f^i\leq f$ to express the fact that each packet in $f$ might be the linear combination of several packets of
conceptual flows.

{\em Notation.} Let $C: E\rightarrow R^+$ be our cost function that associates a non-negative cost $C(e)$ with each edge $e$. We are interested in minimizing the total cost $\sum_{e} C(e) f(e)$, where $f(e)$ is the flow through edge $e$. We also denote by $f_{in}(v)/f_{out}(v)$ the total incoming/outgoing flow of vertex $v$ and with $f_{in}(e)/f_{out}(e)$  the total incoming/outgoing flow to edge $e$. The same notation but with the superscript $i$, {\em e.g.,} $f_{in}^i(u)$ has the same meaning but specifically for conceptual flow $f^i$. We connect all nodes in $S=\{S_i\}$ to a common source node $\cal{S}$ through a set of infinite-capacity and zero-cost edges $E_S=\{ {\cal{S}}S_{i} \}$. Similarly, we connect the nodes in $R=\{R_i\}$ to a common node $\cal{R}$ using  an infinite-capacity and zero-cost set of edges $E_R=\{R_i\cal{R}\}$.

We summarize the LP program for Minimum Cost Routing below:
{\footnotesize
\begin{align*}
  &  \min \sum_{e} C(e) f(e)  \\
 & f(e) \leq \rho  ~~~~ \forall  e \in E -E_S - E_R\\
 & f(e)=\rho   ~~~ \forall  e \in I\\
 &\mbox{Each conceptual flow $f^i$ }   \mbox{corresponding to  $e_i=u_iv_i$ satisfies the constraints:} \\
&   f^i(e)\leq f(e) ~~~ \forall  e \in E-e_i\\
&   f^i(e)\geq 0  ~~~ \forall e \in E\\
&   f_{in}^i(\mathcal{S})=0  \\
&   f_{out}^i(\mathcal{R})=0   \\
&   f_{in}^{i}(u)= f_{out}^{i}(u)  ~~~\forall u \in V-\{\mathcal{S,R},u_i,v_i\}\\
&  f_{in}^i(u_i)\geq \rho    \qquad\text{/*conceptual flow of rate at least } \rho \text{ gets into } (u_i, v_i) \text{*/}\\
&  f_{in}^i(u_i)+f_{out}^i(u_i)\geq 3\rho   \\
&  f_{out}^i(v_i)\geq \rho   \qquad  \text{/*conceptual flow of rate at least } \rho \text{ gets out of } (u_i, v_i) \text{*/} \\
&  f_{in}^i(v_i)+f_{out}^i(v_i)\geq 3\rho
\end{align*}
}
The idea is to lower-bound the probe rate $f(e)$, in edge $e$, given the conceptual flows and
the condition $f^i(e) \le f(e)$. Solving this LP will give us a set of flows and paths, for each
edge $e=(u_i,v_i)$. To ensure identifiability, we need to additionally select a coding scheme, so that the flows arriving  and leaving at $u_i$ and $v_i$ utilize distinct packets, {\em i.e.,} from the observable events at the sink, we can reconstruct for edge $e$ the probability of the events of one of the cases 1-4 in identifiability.

In summary, the minimum cost routing problem, so as to identify the loss rates of a predefined set of edges $I$, can be solved in linear complexity when network coding is used, while the same problem is NP-hard without network coding.

\subsubsection{\bf Routing (including Source Selection and Link Orientation) for Measuring {\em all} Links }

If we are interested in estimating the success rate of  {\em all} identifiable edges of the graph, as opposed to just a restricted set $I$ as in the previous Section, we do not need to solve the above LP. We can simply have each source send a probe and each intermediate node forward a combination of  its incoming packets to its outgoing edges. This  simple scheme utilizes each edge of the graph exactly once per time slot (set of probes sent by the sources) and thus, requires the minimum total bandwidth.  Moreover, if an edge is identifiable, there  exists a coding scheme that allows it to be so.
Example \ref{ex3} and Fig. \ref{fig-abilene-photo} demonstrate such a situation: the source (node 1) sends one probe per experiment, which gets routed and coded inside the network, crossing each link exactly once, and eventually arriving at the receiver (node 9).

{\bf Challenge I: Cycles.} One novel challenge we face in general topologies compared to trees is that probes may be trapped in cycles.
Indeed, if network nodes simply combine their incoming packets
and forward them towards their outgoing links, in a distributed
manner and without a global view of the network, then probes may get trapped in a positive feedback loop (cycle)
that consumes network resources without aiding the estimation process. The following example illustrates
such a situation.
\begin{example}\label{ex2}
Consider again the network shown in  Fig.~\ref{fig-abilene-photo},
but now assume that the orientation of edges $E_4$ and $E_6$ were
reversed. Thus, edges  $E_4$,  $E_5$,  $E_7$, and $E_6$ create a
cycle between nodes $2$, $4$, $5$, and $3$. The probe packets
injected by nodes $3$ and $2$ would not exit this loop.
\hfill{$\square$}
\end{example}

To address this problem, we could potentially equip intermediate nodes with additional functionalities, such as removal of packets that have already visited the same node. This is not practical because it requires keeping  state at intermediate nodes; furthermore, such operations would need to be repeated for every set of probes, leading to increased processing and complexity.

We take a different approach:  we remove cycles. Starting from an undirected graph $G=(V,E)$, where the degree of each node is either one (leaves) or at least three (intermediate nodes), we impose an orientation on the edges of the graph so as to produce a directed acyclic graph (DAG). Our approach is only possible if we are given some flexibility to choose nodes that can act as sources or receivers of probe packets, among all nodes, or among a set of candidate nodes.

There are many algorithms one can use to produce a DAG. Below we propose our own orientation algorithm, Alg. \ref{alg-senders}, that in addition to removing cycles, also achieves some goals related to our problem. In particular, starting from a set of nodes that act as senders $S\subset V$, Alg. \ref{alg-senders} selects an orientation of the graph and a set of  receivers so that (i) the resulting graph is acyclic, (ii) a small number of receiver nodes is selected\footnote{Given a set of sources, one can always produce an orientation and a set of receivers that comprise a DAG, which is what Alg. \ref{alg-senders} does. Conversely, given a set of receivers one can always produce an orientation and a set of sources that comprise a DAG. If both the sets of sources and receivers are fixed, a DAG may not always exist, depending on the topology.}, which is desired for the efficient data collection, and (iii) the resulting DAG leads to a factor graph that works well with belief propagation estimation algorithms. Alg. \ref{alg-senders} guarantees identifiability, but is heuristic with respect to criteria (ii) and (iii); it is important to note, however, that optimizing for criterion (iii) is an open research problem (as discussed in Section \ref{sec:general-bp}).

\begin{algorithm}[t]
\caption{\label{alg-senders} {\bf Orientation Algorithm:} Given
graph $G=(V,E)$ and senders $S \subset V$, find receivers $R \subset
V$ and orientation $\forall~e \in E$, s.t. there are no cycles and
all edges are identifiable.}
{\footnotesize
\begin{algorithmic}[1]

    \FORALL{undirected edges $e=(s,v_2),~s \in S$ }
    \STATE Set outgoing orientation $s \rightarrow v_2$
    \ENDFOR
 \STATE $R=\{s \in S$ that have incoming oriented edges$\}$
 \STATE $V_1=S$;
 \STATE $V_2=\{v_2\in V-V_1: s.t.~\exists~edge~(v_1,v_2)~from~v_1\in V_1\}$

\WHILE{ $V_2\ne \emptyset $ } \STATE Identify and exclude receivers:
find $r\in V_2$ without unset edges: $R := R \bigcup \{r\}$;
$V_2 := V_2-\{r\}$ \STATE Find nodes $U_1\subset V_2$ that
have the smallest number~of~edges~with~unset~orientation. \STATE
Find nodes $U_2\subset U_1$ that have the minimum distance from the
sources $S$. Choose one of them: $v^*\in U_2$.
 \STATE  Let $E^*=\{(v^*,w)\in E$ s.t. $w\in V-V_1\}$

    \FORALL{undirected edges $(v^*,w) \in E^*$}
    \STATE set direction to $v^* \rightarrow w$
    \ENDFOR
    \STATE Update $V_1 := V_1\bigcup \{v^*\}$
\STATE Update $V_2 := \{$(none-$V_1$) nodes one edge away
from current $V_1\}$
 \ENDWHILE

\end{algorithmic}
}
\end{algorithm}
We now describe Alg. \ref{alg-senders}. We sequentially visit the vertices of the graph,
starting from the source, and selecting an orientation for all edges of the visited vertex. This orientation can be thought of as imposing a partial order on the vertices of the graph: in a sense, no vertex is visited before all its parent vertices in the final directed graph.

Lines $1-3$ attempt to set all links attached to the sources as outgoing. If we allow an
arbitrary selection of sources, we may fall into cases where sources
contain links to other sources. In this case, one of the sources
will also need to act as a receiver, \ie we allow the set $S$ of
sources and the set $R$ of receivers to overlap.
In the main part of the algorithm, nodes are divided into three sets:
\begin{itemize}
\item A set of nodes $V_1$, which we have already visited and have
already assigned orientation to all their attached edges. Originally,
$V_1:=S$.
\item  A set of nodes $V_2$, which are one edge away from nodes in $V_1$ and
are the next candidates to be added to $V_1$.
\item The remaining nodes are either receivers $R$ or just nodes not visited yet $V_3 :=
V-V_1-V_2-R$. 
\end{itemize}
In each step of the algorithm, one node $v^* \in V_2$ is selected,
all its edges that do not have an orientation are set to outgoing,
and $v^*$ is added to $V_1 := V_1 \bigcup \{v^*\}$. Note
that the orientation of the edges going from $V_1$ to $V_2$ is already
set. However, a node $v \in V_2$ may have additional unset edges; if
it does not have unset edges, then it becomes a receiver $R := R \bigcup \{v\}$.

We include two heuristic criteria in the choice of $v^* \in V_2$:
(i) first we look at nodes with the smallest number of unset edges;
(ii) if there are many such nodes, then we look for the node with
the shortest distance from the sources $S$; if there are still many
such nodes, we pick one of them at random. The rationale behind
criterion (i) is to avoid creating too many receivers. The rationale
behind criterion (ii) is to create a set of paths from sources to
receivers with roughly the same path length. The criteria (i) and (ii)
are just optimizations that can affect the estimation performance\footnote{One could use
different criteria to rank the candidates $v^*$, so as to enforce
additional desirable properties. Here we used shortest path from the sources
to impose a breath-first progression of the algorithm and paths with roughly the same length.
One could also use other criteria to optimize for the alphabet size
and/or the complexity and performance of the estimation algorithms.}. The algorithm
continues until all nodes are assigned to either $R$ or $V_1$.

\lemma{Algorithm~\ref{alg-senders} produces an acyclic orientation.}\\
\begin{proof} At each step, a node is selected and all its
edges which do not have a direction are set as outgoing. This
sequence of selected nodes constitutes a topological ordering. At
any point of the algorithm, there are directed paths from nodes
considered earlier to nodes considered later. A cycle would exist if
and only if for some nodes $v_i$ and $v_j$: $v_j$ is selected at
step $j>i $ and the direction on the undirected edge $(v_i, v_j)$
is set to $v_i \leftarrow v_j$. This is not possible since if  there
were an edge $(v_i, v_j)$, it would have been set at the earlier step
$i$  at the opposite direction $v_i \rightarrow v_j$. Therefore,
the resulting directed graph has no cycle. It is possible, however,
that there are nodes with no outgoing edges, which become the
receivers.
\end{proof}

We note that the key point that enables us to create an acyclic orientation graph
for an undirected graph is that we allow the receivers to be one of
the outputs of the algorithm. Note that a similar algorithm can be
formulated for the symmetric problem, where the receivers $R$ are
given and the orientation algorithm produces a (reverse) orientation
and a set of sources $S$, s.t. that there are no cycles. However, if
 both $S$ and $R$ are fixed, there is no orientation algorithm
that guarantees the lack of cycles for all graphs.

\lemma{Algorithm~\ref{alg-senders} guarantees identifiability of every link in a
general undirected graph consisting of logical links ({\em i.e.,} with degree $\ge 3$), and for any choice of sources.}
\begin{proof}
The proof follows directly from the fact that the degree of each node is greater than or equal to three
(assuming logical links only), each edge bringing or removing the same amount of flow. Thus, either
the node is a source or a receiver, or the conditions of Theorem \ref{theorem_1} and Fig. \ref{fig_basic} are satisfied.
\end{proof}

\subsection{\label{sec:general-coding}Code Design}

{\bf Challenge II: Code Design affects Identifiability.}
Another novel challenge that we face in general topologies compared to trees is that simple \texttt{XOR} operations do not guarantee path identifiability, as we saw in Example \ref{ex3}. We deal with this challenge using linear operations over higher field sizes as the following example illustrates.

\begin{example}\label{ex1}
Let us revisit the general topology shown in Fig.~\ref{fig-abilene-photo} and briefly discussed in Example \ref{ex3}.
Node $1$ acts as a source: in each experiment, it sends probes $x_1$, $x_2$ and $x_3$
through its outgoing edges $E_1$, $E_2$ and $E_3$, respectively. Nodes $2$,
$4$, $6$, $10$ simply forward their incoming packets to all their
outgoing links. Node $3$ performs coding operations as follows: if
within a predetermined time-window it only receives probe packet
$x_2$, it simply forwards this packet. The same holds if it only receives
 probe packet $x_3$. If, however, it receives both packets $x_2$
and $x_3$, it linearly combines them  to create the packet $x_2+x_3$
that it then sends through its outgoing edge $E_6$. Nodes $5$, $7$
and $8$ follow a similar strategy.
If all links are functioning, node $5$ sends packet  $3x_2+x_3$,
node $7$ sends packet $x_1+x_2$ and finally, node $8$ sends packet
$3x_1+x_2$. The receiver node $9$ observes, in each experiment, three
incoming probe packets. {\em E.g.}, if it only observes the incoming
packet $x_3$, it knows that all paths from the source $S$ have
failed, apart from path $P_4$. Therefore, it infers that no packets
were lost on edges $E_3$, $E_6$, $E_9$. \hfill{$\square$}
\end{example}

More generally, we are interested in practical code design schemes that allow for identifiability of all edges in general topologies. We will achieve this goal by designing for path identifiability, which is a different condition.
In particular, we are interested in coding schemes that allow us to identify the maximum number of 
path states. This can be achieved by mapping the failure of each subset of paths to a distinct probe observed at the receivers.
For this to be possible, (i) the alphabet
size must be sufficiently large and (ii) the coding coefficients must be carefully assigned to edges.

Recall that receiver nodes only have incoming edges. Let $e_{R_j}$ be an edge adjacent to a receiver $R_j$ and $\mathcal{P}(e_{R_j})$ be the set of paths that connect all source nodes  to receiver $R_j$, and have $e_{R_j}$ as their last edge. We say that a probe coding scheme allows maximum  path identifiability if it allows the receiver $R_j$, by observing the received probes from edge $e_{R_j}$ at a given experiment, to determine which of the $\mathcal{P}(e_{R_j})$ paths have been functioning during this experiment and which have not.

\subsubsection{Alphabet Size}
There is a tradeoff between the field size and path identifiability.
On one hand, we want a small field size mainly for low computation (to do linear operations at intermediate nodes) and secondarily for
bandwidth efficiency (to use a few bits that can fit in a single probe packet). In practice, the latter is not a major
problem, because for each probe, we can allocate as many bits as the maximum IP packet size, which is quite large in the Internet.\footnote{The MTU (maximum transmission unit) on the Internet is at least 575 Bytes (4800 bits), and up to 1500 bytes (12000 bits), including headers. However, in simulation of realistic topologies, we did not need to use more than 18 bits.}
However, for computation purposes, it is still important that we keep the field size
as small as possible.  On the other hand, a larger field size makes it easier to achieve path identifiability.

For maximum path identifiability, there is the following loose lower bound on the required alphabet size.
\begin{lemma} Let $G=(V,E)$ be acyclic and let $\mathcal{P}_m$ denote the maximum number of paths sharing an incoming edge of any receiver $R_j$, {\em i.e.,} \mbox{$\mathcal{P}_m=\max_{e_{R_j}}$ $\mathcal{P}(e_{R_j})$}. The alphabet size must be greater than or equal to $\log{\mathcal{P}_m}$.
\end{lemma}
\begin{proof}
Assume that one of the $\mathcal{P}_m$ paths is functioning while
all the others are not. Since two paths cannot overlap in all edges,
there exists a set of edge failures such that this event occurs. For
the receiver to determine which of the $\mathcal{P}_m$ paths function
and which ones fail, it needs to receive at least $\mathcal{P}_m$ distinct
values. Essentially, the field size should be large enough to allow
for distinguishing among all possible paths arriving at each receiver.
Therefore, we need a field size $q \ge \mathcal{P}_m$.
\end{proof}

What the above lemma essentially counts is the number of distinct values that we need to
be to able to distinguish. This can be achieved using either scalar network coding over
a finite field $F_q$ of size $q$, or vector linear coding with vectors of appropriate length.
{\em E.g.,} see \cite{vnc} for an application to the multicast scenario, where scalar network coding 
over a finite field of size $q$ was treated as equivalent to vector network coding over the space of 
binary vectors  of length $\log q$.

The reader will immediately notice that there is an exponential number of paths and failure patterns.
We would like to note that this is not unique to our work, but inherent to tomography problems
that try to distinguish between exponentially large number of configurations, {\em e.g.,} transfer matrices
and their failure patterns in the passive tomography \cite{jaggi, jaggi-arxiv}.
Even in that case, simulations of large topologies, such as Exodus, showed that a moderate field size is sufficient in practice. However, in our case of active tomography, a potentially large alphabet size is needed only if one insists to infer the loss rates on {\em all links simultaneously}. In practice, one can infer the loss rates on links one-by-one, by carefully selecting the probes and measuring only the corresponding paths, thus creating the ``5-link'' motivating example, where \texttt{XOR} operations are sufficient.

\subsubsection{Code Design}

Having a large alphabet size is necessary but not sufficient to guarantee path identifiability.
We also need to assign coefficients $\{c_h\}$ so that the failure of every subset of paths leads to a distinct observable outcome (received probe content).
Here we discuss how to select these coefficients.

Consider a particular incoming edge $e_{R_j}$ to a receiver $R_j$ and let $m$ be the number of paths arriving at this edge from source $S_i$. Consider one specific path $h$ that connects source $S_i$ to $R_j$  via edges $e_{h_1},e_{h_2},...e_{R_j}$. The contribution $P_h$ from path $h$ to the observed probe is what we call a {\em path monomial}, {\em i.e.,} the product of coefficients on all edges across the path and of probe $\mathcal{X}_{S_i}$ sent by source $S_i$:
$$P_h=c_{h_1}\cdot c_{h_2}... \cdot c_{R_j} \cdot \mathcal{X}_{S_i}$$
For simplicity, we use $P_h$ to denote both a path and the corresponding path monomial. Note that each path consists of a distinct subset of edges; as a result, no path monomial is a factor of any other path monomial. We can  collect all the monomials $P_h$ in a column vector $\vec{P}_{e_{R_j}}=(P_1,\; P_2, \ldots P_m)$.

If all paths arriving at edge $e_{R_j}$ are working (no link fails), the received probe at that edge is the summation of the contributions $\vec{P}=(P_1,P_2, ...P_m)$ from all $m$ paths:
$$\textit{Probe received through } e_{R_j} \textit{ (when no loss)}= P_1+P_2+...P_m$$
In practice, however, any subset of these $m$ paths may fail due to loss on some links and the received probe becomes the summation of the subset of paths that did not fail. Let $\vec{X}=(x_1, x_2, ...x_m)$ be the vector indicating which paths failed: $x_k=0$ if path $k$  failed and $1$ otherwise.
Therefore, the probe received through $e_{R_j}$, in the case of loss, is:
$$\textit{Probe received through } e_{R_j} \textit{ (when loss)}=\vec{X}\cdot \vec{P}= \sum_{k=1}^m x_k\cdot P_k,$$
where $\vec{X}$ is the indicator vector corresponding to the loss pattern, \ie has entry zero if a path fails, and one otherwise. The vector $\vec{X}$ can take $2^m$ possible values; let $\vec{X_k}$ denote the $k^{th}$ possible value, $k=0,...2^m-1$. To guarantee identifiability, no two subsets $k,l$ of failed paths should lead to the same observed probe:
 $\vec{X_k} \cdot \vec{P} \neq \vec{X_l} \cdot \vec{P}$.

Therefore, a successful code design should lead to $2^m$ distinct probes, one corresponding to a different subset of paths failing. In other words, to guarantee identifiability, the coefficients $\{c_e\}_{e\in E}$ assigned to edges $E$ should be such that:
$\vec{X_k} \cdot \vec{P} - \vec{X_l} \cdot \vec{P} \ne 0,~\forall~k,l=0,...2^{m}-1$. We can write all these constraints together as follows, which is essentially the definition of {\em path identifiability},  mentioned in the beginning of Section~\ref{sec:general-identifiability}:
\begin{equation}
\label{eq:constraints}
\prod_{k,l =0,...2^{m}-1} (\vec{X_k} \cdot \vec{P}_{e_{R_j}} - \vec{X_l} \cdot \vec{P}_{e_{R_j}}) \ne 0
\end{equation}
Since each $P_h=c_{h_1}\cdot c_{h_2} ... \cdot c_{R_j}\cdot \mathcal{X}_{S_i}$ is a monomial, with variables the coding coefficients $\{c_e\}_{e\in E}$, the left hand side in Eq.(\ref{eq:constraints}) is a multivariate
polynomial $f(c_1,c_2,... c_{|E|})$ with degree in each variable at most $d\le 2^m$.

\begin{lemma}
The multivariate polynomial $f(c_1,c_2,... c_{|E|})$ at the left side of Eq.(\ref{eq:constraints}) is not identically zero.
\end{lemma}
\begin{proof}
The ``grand'' polynomial is not identically zero because each factor in the product $(\vec{X_k} \cdot \vec{P}_{e_{R_j}} - \vec{X_l} \cdot \vec{P}_{e_{R_j}})$ is a nonzero polynomial in $\{c_h\}.$ Indeed, $\vec{X_k}$ and $\vec{X_l}$ differ in at least one position, say $g$, corresponding to a monomial $P_g$. Consider the following assignment for the variables $\{c_h\}$. Assign to all the variables in this monomial a value equal to one. Assign to all other variables $\{ c_h\}$ a value of zero. Since no monomial is a factor of any other monomial, this implies that the vector $\vec{P}_{e_{R_j}}$ takes value one at position $g$, and zero everywhere else. Thus, this assignment results in a non-zero evaluation for the polynomial $(\vec{X_k} \cdot \vec{P}_{e_{R_j}} - \vec{X_l} \cdot \vec{P}_{e_{R_j}})$, and as a result, this cannot be identically zero.
\end{proof}

Up to now, we have considered paths that employ the same incoming edge. We can repeat exactly the same procedure for all
incoming edges, and generate, for each such edge, a polynomial in the variables $\{c_h\}$. Alternatively, we could also find these polynomials by calculating the transfer matrix
between the sources and the specific receiver node using the state-space representation of the network and the algebraic tools developed in
\cite{algebraic}. Either way, the code design consists of finding values for the variables $\{c_h\}$ so that the product of all polynomials, $f$, evaluates to a nonzero value.
There are several different ways to find such assignments, extensively
studied in the network coding literature, {\em e.g.,} \cite{schwartz, harvey, ho-random}.
One way to select the coefficients is randomly, and this is the approach we follow in the simulations.
In that case, it is well-known that we can make the probability that $f(c_1,c_2,... c_{|E|})=0$ arbitrarily
small, by  selecting the coefficients randomly over a large enough field\footnote{From the Schwartz-Zippel Lemma
\cite{schwartz}, which has been instrumental for network coding \cite{ho-random},
we know the following. If $f(c_1,c_2,... c_{|E|})$ is a non-trivially zero polynomial with degree at most $d$ in each variable, and we choose $\{c_e\}_{e\in E}$ uniformly
at random in $F_q$ with $q>d$, then the probability that $f(c_1,c_2,... c_{|E|})=0$ is at most $1-{(1-\frac{d}{q})}^{|E|}$.}.

{\bf Deterministic Operation.} We emphasize that although the coefficients may be selected randomly (at setup time), the operation of intermediate nodes (at run time) is deterministic.
 At setup time, we select the coefficients and we verify the identifiability conditions, and select new coefficients if needed for the conditions to be met. After the selection is finalized, we learn the coefficients and use the same ones at each time slot. Learning the coefficients is important in order to be able to infer the state of the paths and links.

{\bf State Table and Complexity Issues.}
Once the coefficients are randomly selected, we need to check whether the constraints
summarized in Eq.(\ref{eq:constraints}) are indeed satisfied. If they are satisfied, the code design guarantees
identifiability; if they are not satisfied, then we can make another random selection and check again. One could
also start from a small field size and increase it after a number of failed trials.

\begin{algorithm}[t]
\caption{\label{alg-observations} Deduce State of the Paths from the Observations}
{\footnotesize
\begin{algorithmic}
\FORALL{ $S_i \in Senders $ }
        \FORALL{ $R_j \in Receivers $ }
                \FORALL{ incoming links $e_{R_j}$ }
                \STATE Map the observed probe to the state of all paths from $S_i$ to $R_j$ coming through link $e_{R_j}$.
                \ENDFOR
        \ENDFOR
\ENDFOR
\end{algorithmic}
}
\end{algorithm}

\begin{figure}[t]
\centering{\includegraphics[scale=0.22]{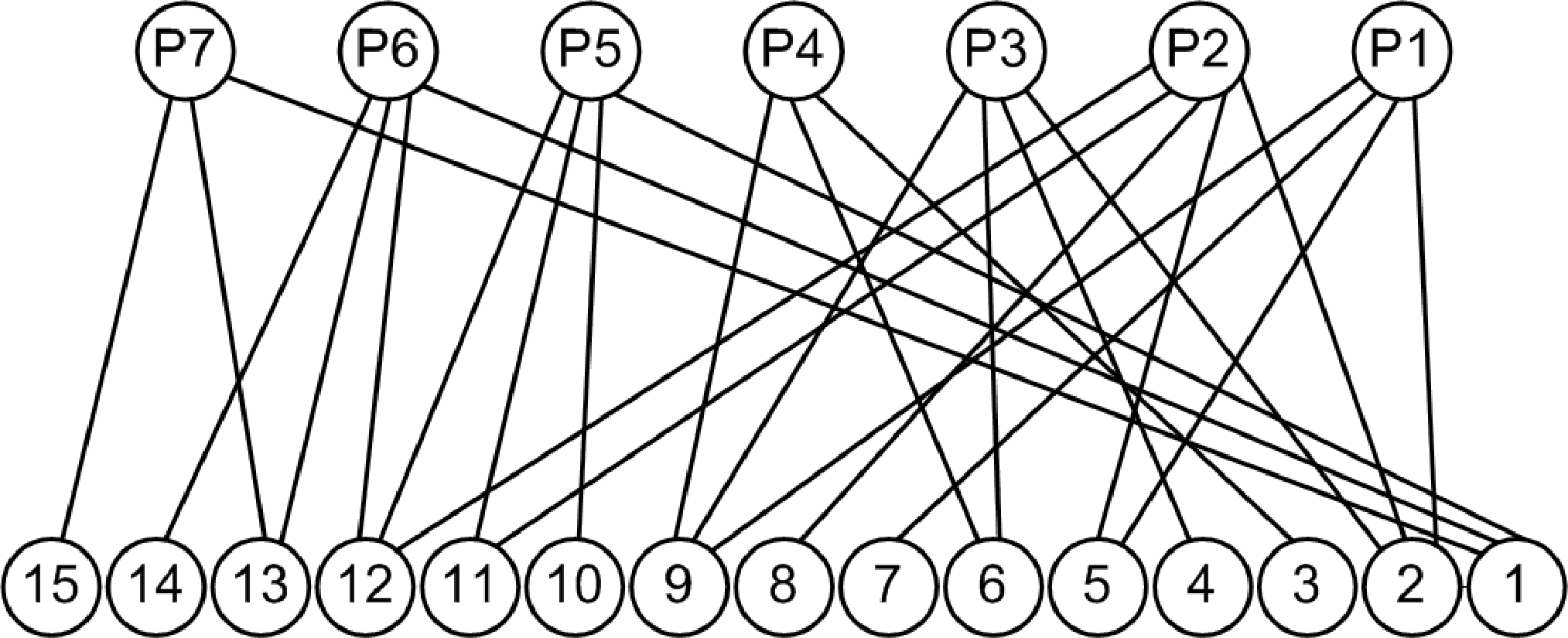}}
\vspace{-0.75em}
\caption{Factor graph corresponding to the Abilene graph (shown in Fig.~\ref{fig-abilene-photo}).
It maps the 15 links to the 7 observable paths at the single receiver (9).
It is used for the belief propagation estimation algorithm.\label{figs-fg}}
\vspace{-0.25em}
\end{figure}

The evaluation of Eq.(\ref{eq:constraints}) above requires to check an exponential number of
constraints, up to $2^m$, where $m$ is the number of paths for a triplet (source, receiver, edge at receiver). Because the current orientation algorithm does not exclude any edges in the process of building the DAG, we might end up with a large number of paths depending on the connectivity of the topology and the selection of the sources\footnote{{\em E.g.,} for the Abilene topology shown in Fig. \ref{fig-abilene-photo}, with 1 source, there were at most three paths per $(S_i,R_j,e_{R_j})$ triplet, but for the larger Exodus topology (described in Section \ref{sec:general-simul}) with 5 sources, the average and maximum number of paths per triplet were $9$ and $25$, respectively (for a specific selection of sources in both topologies).}. This motivated us to look into ways for reducing the number of paths per triplet\footnote{For example, if we are willing to accept less than 100\% path identifiability, we can randomly assign coefficients without checking for identifiability conditions. From the observed probes at the receivers, we then infer the subset of paths that failed by looking up a table which is pre-computed by solving a subset sum problem. 
If we identify one or more subsets of paths that when failing lead to the same observed probe, we can use a heuristic, {\em i.e.,} pick one of the candidate subsets, their union or intersection. We then feed the state of the paths to the BP estimation algorithm. This is the approach we follow in the simulation Section.}. Even putting aside the exponential number of paths for a moment, the problem is essentially a subset sum: we receive a symbol at a receiver and we would like to know which combinations of non-failed paths add up to this number. This is a well-known NP-hard problem.

\begin{figure*}[t!]
\centering
\vspace{-2.0em}
{\includegraphics[scale=0.315,angle=90]{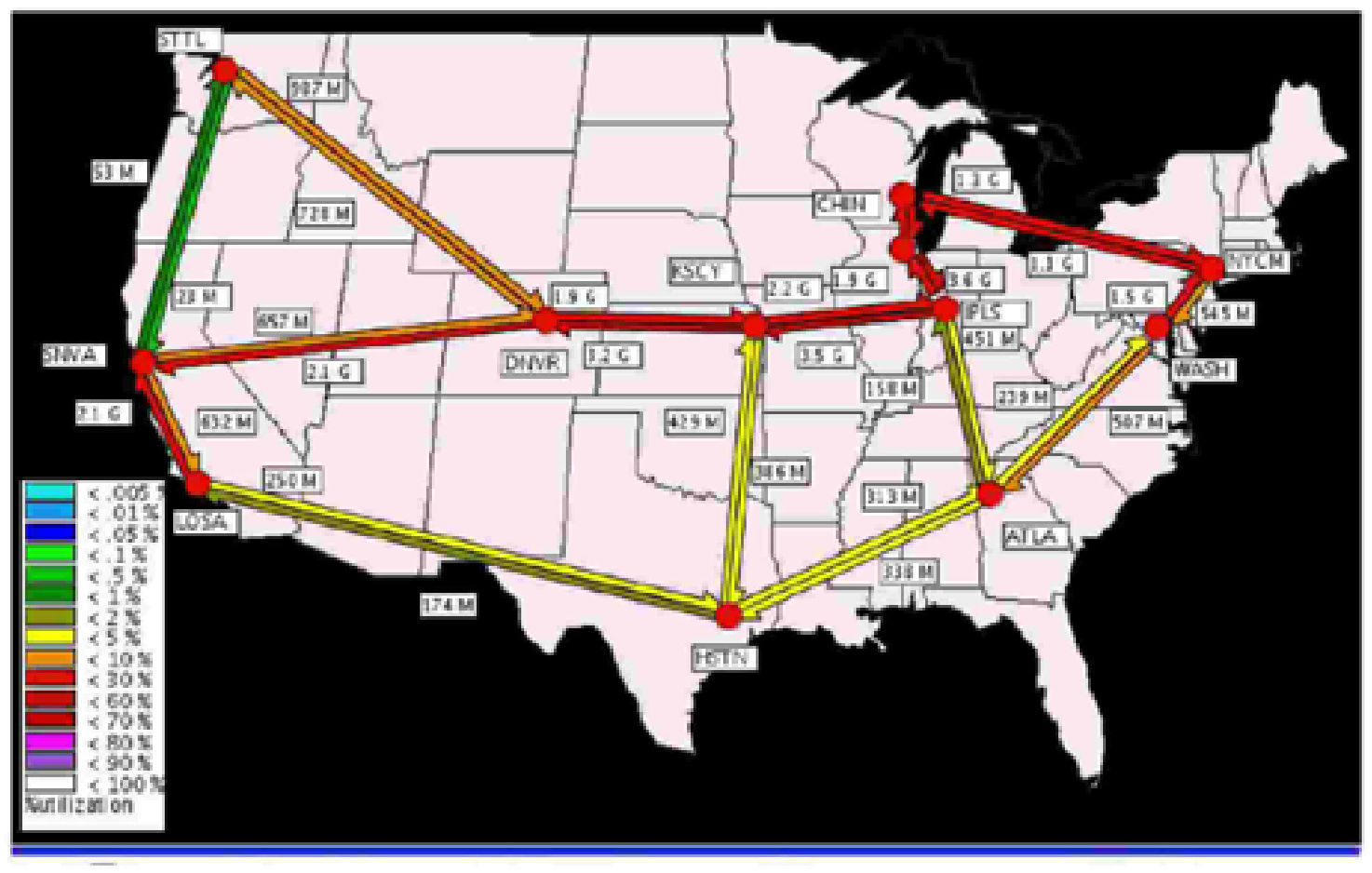}}
{\includegraphics[scale=0.3, angle=90]{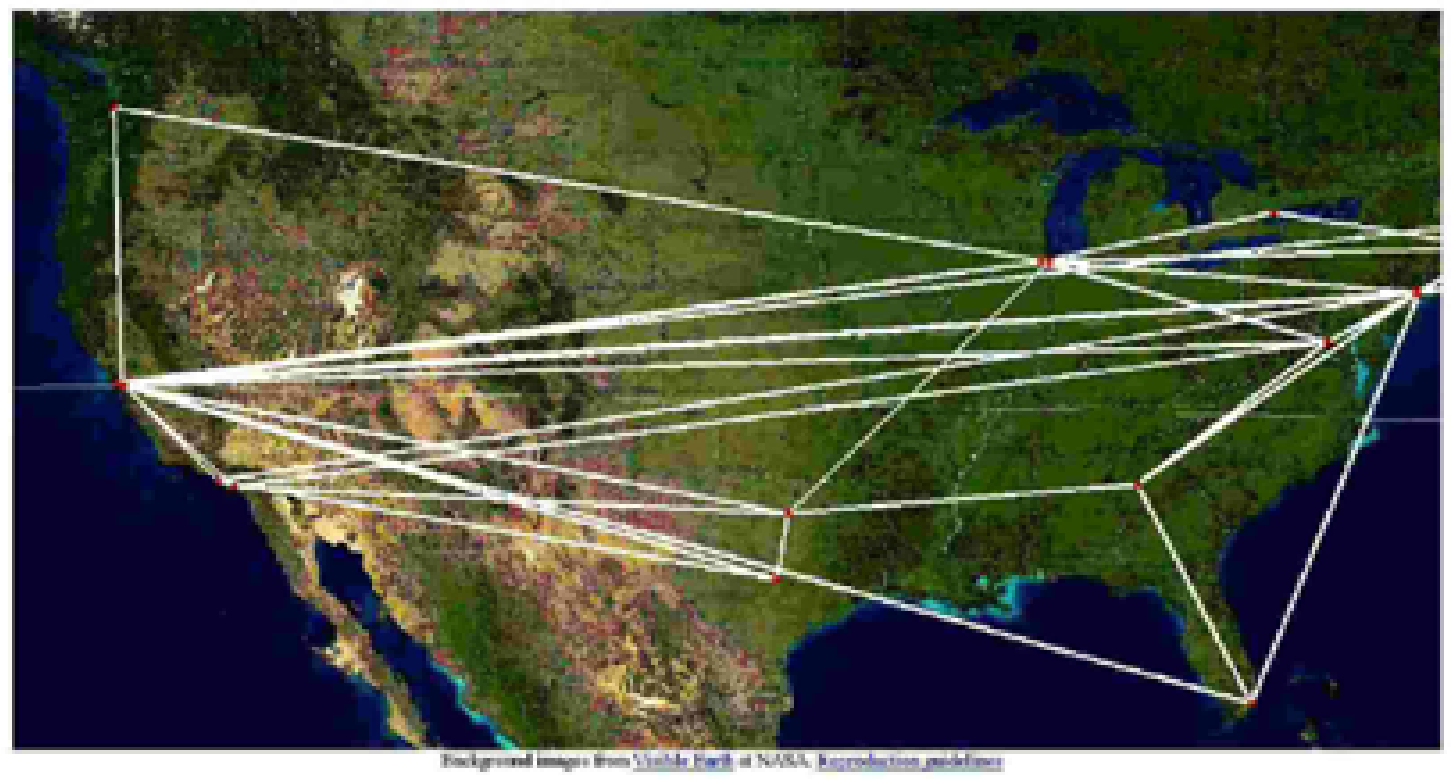}}
\vspace{-3.0em}
\caption{Topologies used in simulations. (a) Left: Abilene Backbone Topology (small research network). (b) Right: Exodus POP Topology (large ISP).} \label{figs-isp-maps}
\end{figure*}

This being said, we do not expect this to be a source of high complexity in practice for several reasons. First, the algorithm that maps the received symbol to a state of paths can be run offline and the table can be computed and stored. This is a static scenario, since coding coefficients remain the same across scenarios. Therefore, we incur setup complexity once in the beginning, but not during run time. All we need to do every time we receive a symbol is
just a table lookup, which is inexpensive ($O(1)$), when implemented using hash tables. 
Second, this design is only necessary if one wants to infer {\em all}
links at the same time, which may be an overkill in practice. The most typical use of our framework in
practice will be for inferring the loss rates of a few congested specific links of interest, in which case we do not need to keep track of the state of all paths, and the size of the table reduces.

\subsection{\label{sec:general-bp}Loss Estimation using Belief Propagation}
For our approach to be useful in practice, we need to employ a low
complexity algorithm that allows to quickly estimate the loss rate
on every link from all the observations at the receiver. Because MLE is quite involved for general graphs,
especially large ones, we use a suboptimal algorithm instead; in
particular, we use the Belief Propagation (BP) approach that we also used
for trees, see Section \ref{sec:trees-bp}.

There are two steps involved in the algorithm for each round of received probes. First, from the observations, we need to deduce the state of the paths traversed by these probes, as described in Algorithm \ref{alg-observations}. The second step is to use the Belief-Propagation (BP) algorithm, to approximate Maximum Likelihood (ML) estimation.
Once we know which paths worked and which failed in this round, we feed this information into the factor graph, which triggers iterations, and
leads to the estimate of the success rate. Similarly to trees, the factor graph is again a bipartite graph, between links
and paths containing these links. For example, Fig. \ref{figs-fg} shows the bipartite graph corresponding to the Abilene topology of Fig.\ref{fig-abilene-photo}, which we have been discussing in all the examples in this Section.

The main difference in the general graphs compared to the trees is that there are multiple
(instead of exactly one) paths between a source and a receiver; this has two implications. The first implication is that the design of the coding scheme must allow us to deduce the state of these multiple paths between a source, a receiver and
an incoming edge at the receiver $(S_i,R_j, e_{R_j})$; this has been extensively discussed in the previous Section on code design. The second
implication is that there are more cycles in the factor graph of a general graph, which affects the estimation accuracy of the BP algorithm.

In general, the performance of the BP algorithm depends on the properties of the factor graph. Several problems have been identified in the BP
literature depending on the existence of cycles, the ratio of factors vs. variables (\eg links per path) and other structural properties (stopping sets,  trapping sets, diameter). Fixing such BP-specific problems
are outside the scope of this paper and is a research topic on its own. However, we did address two of the aforementioned problems, using existing
proposals from the BP literature.
First, for performance enhancement in the presence of cycles in the factor
graph, we used a modification of the standard BP, similar to what was proposed
in the context of error correcting codes \cite{bp1}. The idea  is to combat the overestimation
of beliefs by introducing a multiplicative correction factor $a<1$ for messages passing between
variables (links) and factors (paths)\footnote{In the same way, we could also use an additive correction factor instead. Making those factors adaptive could give even better results.
In the same paper \cite{bp1}, additional modifications of the factor graph (junction tree algorithm, and generalized belief propagation) to deal
with cycles have been proposed, which we did not implement in this paper. Other possible modifications of the BP include: \cite{bp2}, a multistage iterative decoding algorithm that combines belief propagation with ordered
statistic decoding, and reaches close to the performance of MLE although with a higher complexity than BP; and \cite{bp3}, which
uses a probabilistic schedule for message passing between variable nodes and check nodes in the factor
graph instead of simple message flooding at every iteration.}. Second, we designed the orientation algorithm to traverse the actual topology in a breadth-first manner in order to produce short paths and thus small ratio
of links per path in the factor graph, which has a good effect on the BP performance. More generally, we note that the properties of the factor graph depend on the orientation algorithm. One could optimize
the orientation algorithm to achieve desired properties of the factor graph. In this paper, we have not done modifications other than the two
mentioned above because (i) the overall estimation worked well in all the practical cases we tried, and (ii) the design of a factor graph for better BP performance is a research topic on its own and outside the scope of this work.

\subsection{\label{sec:general-simul}Simulation Results}

We now present extensive simulation results over two realistic topologies.

\subsubsection{Network Topologies}

We used two realistic topologies for our simulation, namely the
backbones of Abilene and Exodus shown in Fig.~\ref{figs-isp-maps}.
Abilene is a high-speed research network operating in the US and
information about its backbone is available online \cite{abilene}.
Exodus is a large commercial ISP, whose backbone map was inferred by
the Rocketfuel project \cite{rocketfuel}. Both topologies were
pre-processed to create logical topologies that have degree at least~3.
For Exodus, nodes with degree~$2$
were merged to create a logical link between the neighbors of such
nodes, while nodes with degree 1 were filtered; the resulting logical
topology contains $48$ nodes and $105$ links. For the Abilene topology,
due to its small size, in addition to merging some links in tandem,
more links were added; the modified topology comprises of $10$ nodes
and $15$ links, and is the one shown in Fig.~\ref{fig-abilene-photo} and used as an example of a general topology throughout Section \ref{sec:general}.

\begin{figure}[t!]
\centering \subfigure[All possible placements of one
source]{\includegraphics[scale=0.25,angle=-90]{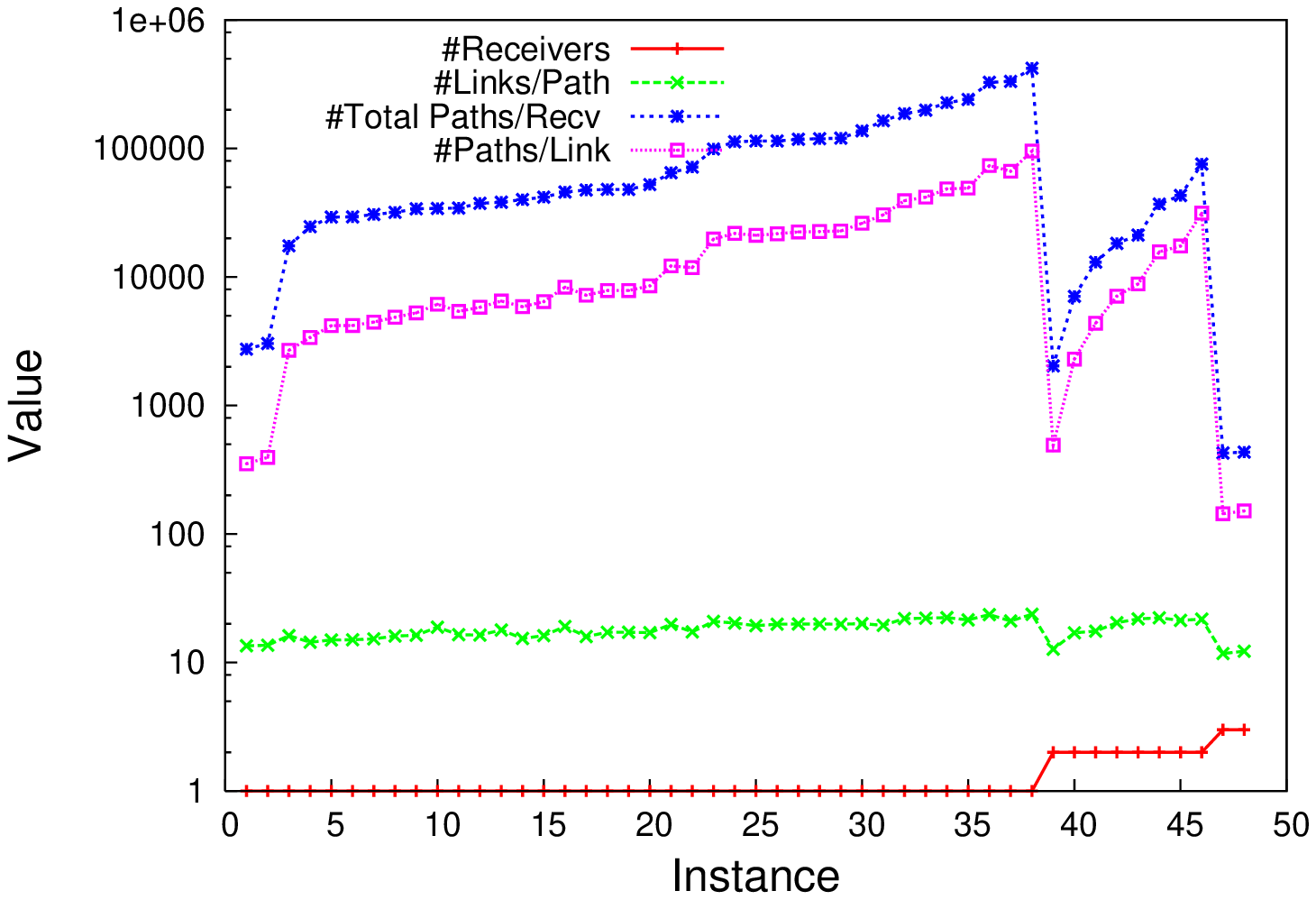}}
\subfigure[All possible placements of two
sources]{\includegraphics[scale=0.25,angle=-90]{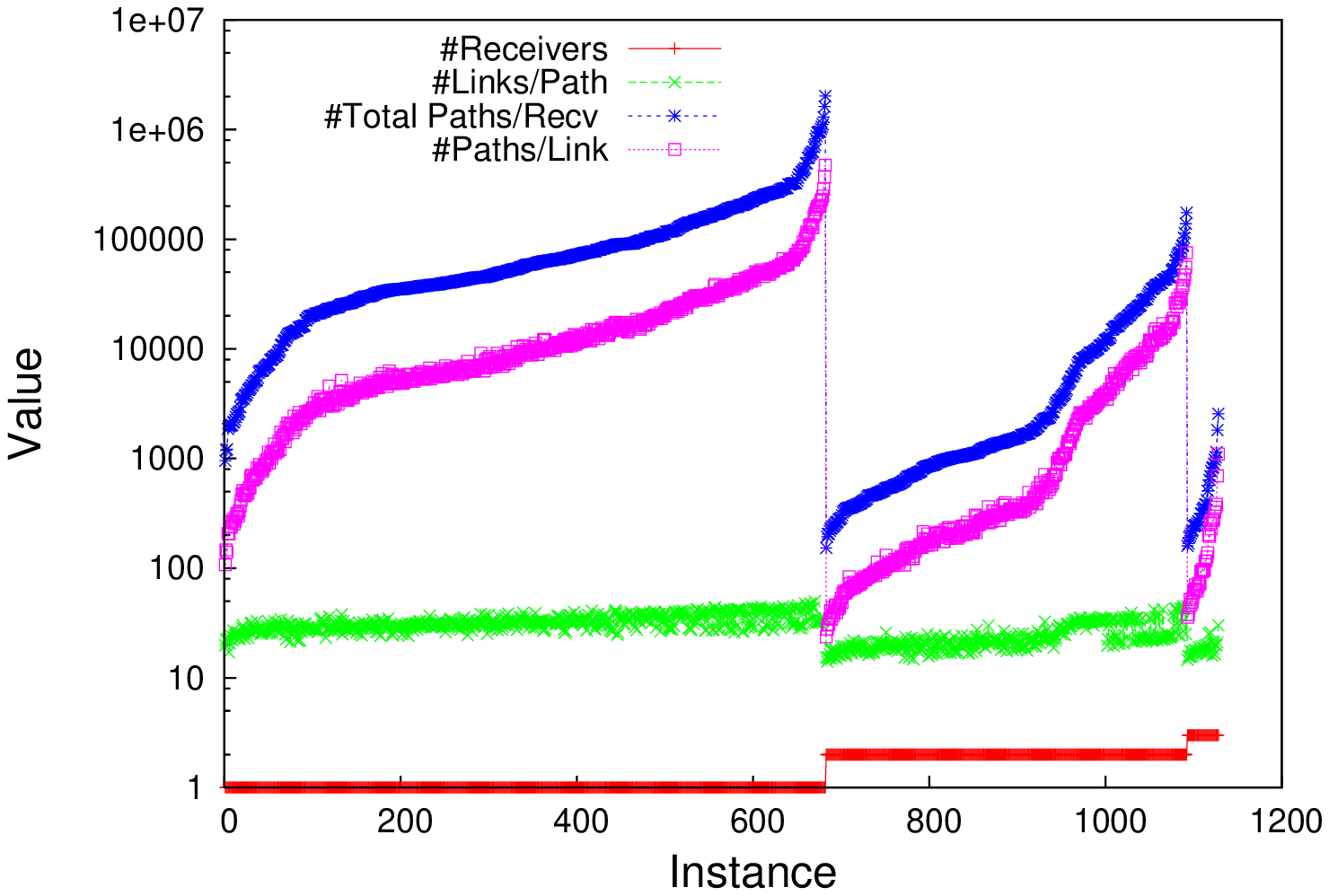}}
\caption{Running the Orientation Algorithm on the Exodus topology.} \label{fig-orientation-exodus} 
\end{figure}

\begin{table}[t!]
\scriptsize
\centering
\begin{tabular}[]{|l|l|l|l|l|l|l|l|}
\hline
Topology & {Srcs}-{Recvs}                   &  Coding      & Links /   &  Paths /      &  Edge Disj. \\
~        & ~     &  Points     &  Path     &    Link       &  Paths     \\
\hline
Abilene    & \{1\}-\{9\}        &     4         &  3.85    &    1.8        &    3               \\
\hline
                & \{5\}-\{6\}        &     4        &  3.71    &   1.73        &    3              \\
\hline
                & \{9\}-\{2\}        &     4        &   4.28   &     2.0       &    2            \\
\hline
                & \{1,9\}-\{7\}      &     5        &   3.25     &   1.73        &    4                \\
\hline
                & \{3,6\}-\{9\}      &     5        &    4      &    2.13      &    4             \\
\hline
                & \{9,6\}-\{4\}      &     5        &   3.25     &   1.73       &    4              \\
\hline
                & \{1,5,9\}-\{7\}    &     5        &   3.2       &    2.13        &    5               \\
\hline
                & \{1,4,10\}-\{9\}   &     6        &     3      &      2.33       &   6                 \\
\hline
Exodus         & \{39,45\}-\{30,40\}       &    25     &    9.47     &    56.47      &     4             \\
\hline
\end{tabular}
\vspace{0.5em}
\caption{Properties of the orientation graphs produced by Alg. \ref{alg-senders} for
different topologies and choices of sources.}
\label{table2-results}
\vspace{-0.75em}
\end{table}

For all simulations, the link losses on different links are assumed independent,
and may take large values as they reflect  losses on logical links, comprising
of cascades of physical links, as well as events related to congestion control within the network.

\begin{figure}[t]
\centering
{\includegraphics[scale=0.3, angle=-90]
{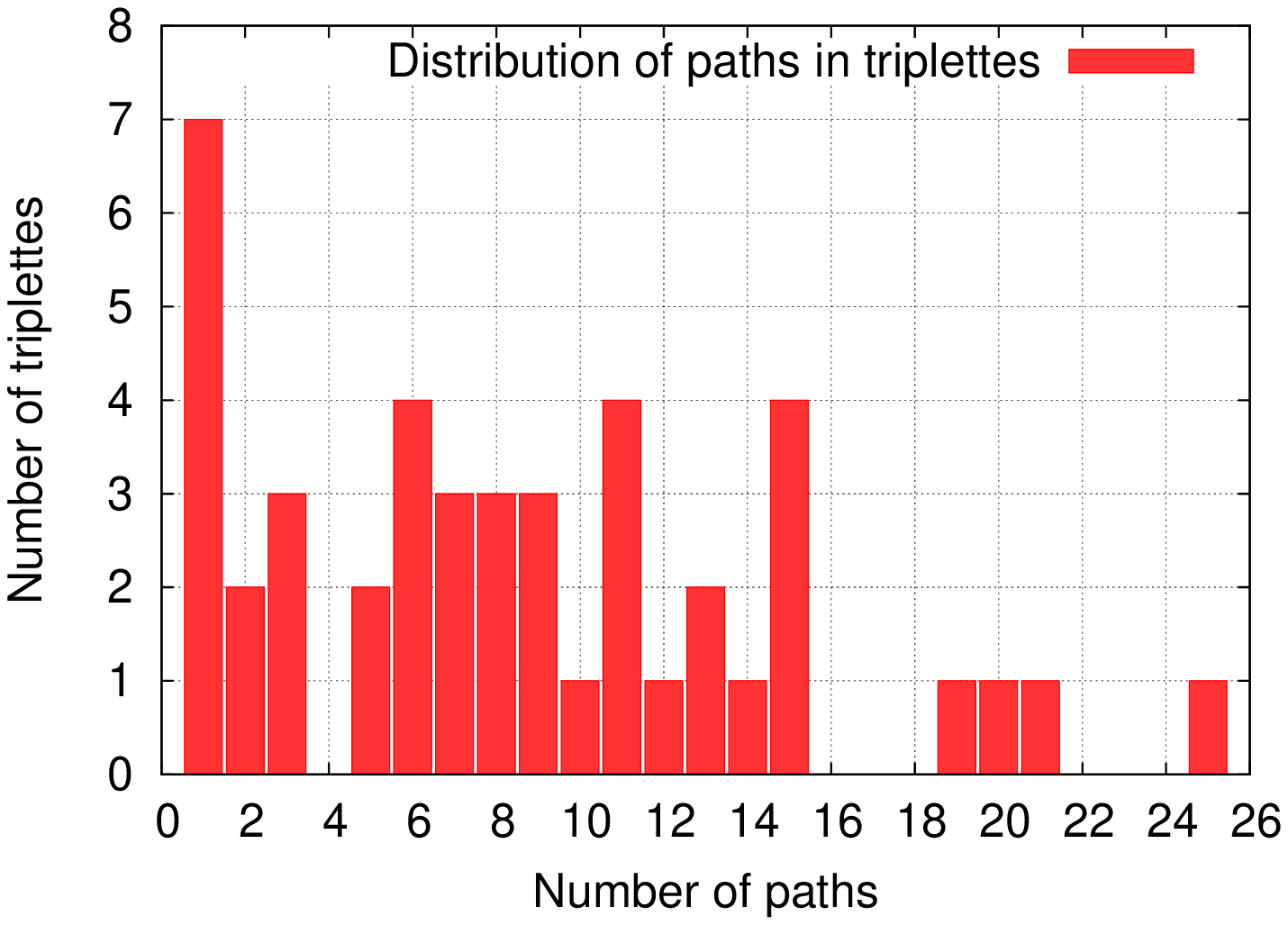}}
\caption{Distribution of the number of paths for all triplets $(S_i,R_j,e_{R_j})$ for the Exodus topology. \label{fig:triplettes}}
\vspace{-0.25em}
\end{figure}

\begin{figure}[t!]
\centering
{\includegraphics[scale=0.3, angle=-90]
{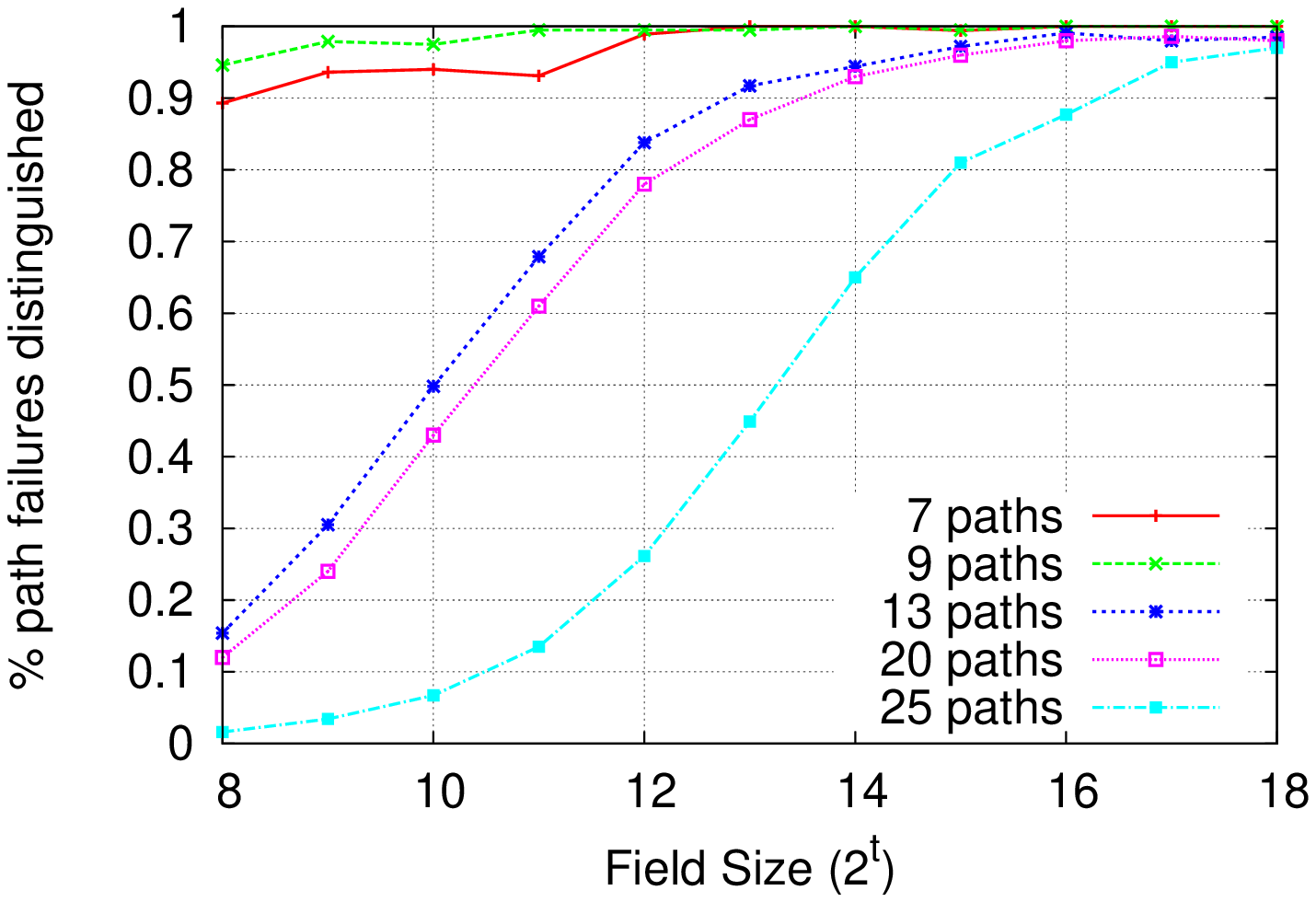}}
\vspace{-0.25em}
\caption{Random code design for the Exodus topology. The X-axis shows the field size over which we choose the coding coefficients randomly: finite fields with different sizes ($F_{2^8}-F_{2^{18}}$). The Y-axis shows the effect on path identifiability (probability of success, defined as the \% of the paths in a triplet $(S_i,R_j,e_{R_j})$ that we can uniquely distinguish from the observed outcome). 
\label{fig:coeffs-identifiability}}
\vspace{-0.75em}
\end{figure}

\subsubsection{Results on the Orientation Algorithm}

In Fig.~\ref{fig-orientation-exodus}, we consider the Exodus topology
and we run the orientation algorithm for all possible placements of
one and two sources; we call each placement an ``instance''. We are
interested in the following properties of the orientation produced
by Alg. \ref{alg-senders}:
\begin{itemize}
\item the number of receivers: a small number allows
for local collection of probes and easier coordination.
\item the number of distinct paths per receiver: this relates to the
alphabet size and it is also desired to be small.
\item the number of paths per link and links per path: these affect the
performance of the belief propagation algorithm.
\end{itemize}
Fig.~\ref{fig-orientation-exodus} shows the above four metrics, sorting
the instances first in increasing number of receivers and then in
increasing paths/receiver. The following observations can be made.
First, the number of receivers produced by our orientation algorithm
is indeed very small, as desired. Second, the number of links per
path is almost constant, because by construction, the orientation
algorithm tries to balance the path lengths. Third, the
paths/receiver and paths/link metrics, which affect the alphabet
size and the quality of the estimation, can be quite large; however, they
decrease by orders of magnitude for configurations with a few
receivers; therefore, such configurations should be chosen in
practice. Finally, Table~\ref{table2-results} considers different
choices of sources in the (modified) Abilene and Exodus topologies,
and shows some properties of the produced orientation.

\begin{figure}[t]
\centering \subfigure[Estimated vs. real success rate (for 3000 probes)]
{\includegraphics[scale=0.4]{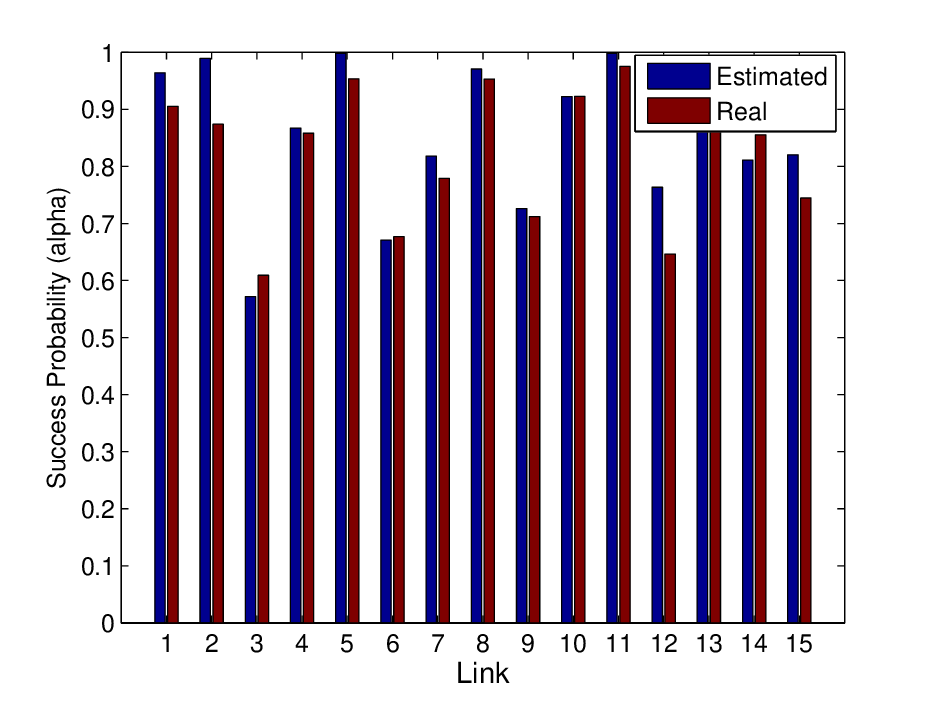}}
\subfigure[$ENT$ metric vs. number of probes]
{\includegraphics[scale=0.4]{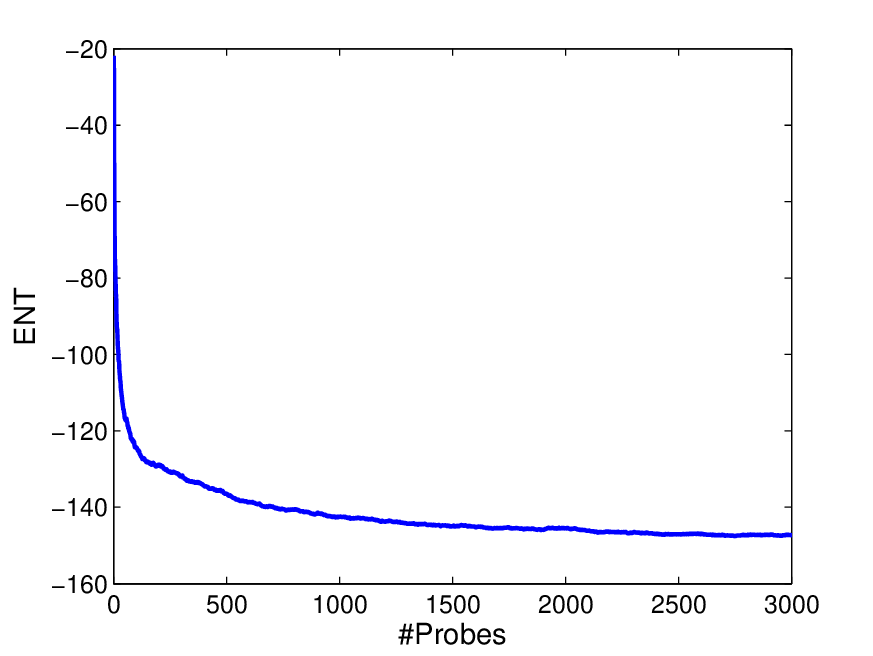}}
\caption{(Modified) Abilene topology. Loss rates ($\overline\alpha$'s) are
different across links: they are assigned inversely proportional to the
bandwidth of the actual links as reported in
\cite{abilene}. 
The resulting average loss rate is $17\%$.}
\label{fig:sc2-mse-alllinks} 
\vspace{-0.5em}
\end{figure}

\subsubsection{\label{sec:evaluate-code}Evaluation of Random Code Design for Real Topologies}

In this Section, we simulate {\em random code design schemes} for the example topologies of Abilene and Exodus.

Consider a particular incoming edge $e_{R_j}$ to a receiver $R_j$ and let $m$ be the number of paths arriving at this edge from the same source $S_i$. If two subsets of paths lead to the same probe, then they are indistinguishable, which leads to lack of
identifiability. In practice, since many of the paths for a triplet $(S_i,R_j,e_{R_j})$ share links between them, we have much less than $2^m$
possible distinct probes. The exact number depends on the connectivity of the topology. In the simulations, the content of the probe from each subset
of paths is used as a key to a hash table.
If two subsets lead to the same probe, then they will end up into the same bucket. The number of unique buckets in the hash
table gives us the number of different combinations of failed/non-failed paths that are distinguishable from each other.
We normalize this number by the total number of possible distinct subsets, and we call this number the probability of success (path identifiability)
of the code design for this particular triplet $(S_i,R_j,e_{R_j})$.

For the {\em Abilene topology} (10 nodes, 15 links), using one source and the orientation algorithm, we obtained a DAG with 1 receiver (Fig.~\ref{fig-abilene-photo}). The maximum number of paths observed for an incoming edge at the receiver was 3. A random choice of coding coefficients over a finite field of size $2^6$ was sufficient to achieve 100\% identifiability of all paths on all edges.

\begin{figure}[t!]
\centering \subfigure[one source: node 1]
{\includegraphics[scale=0.45]{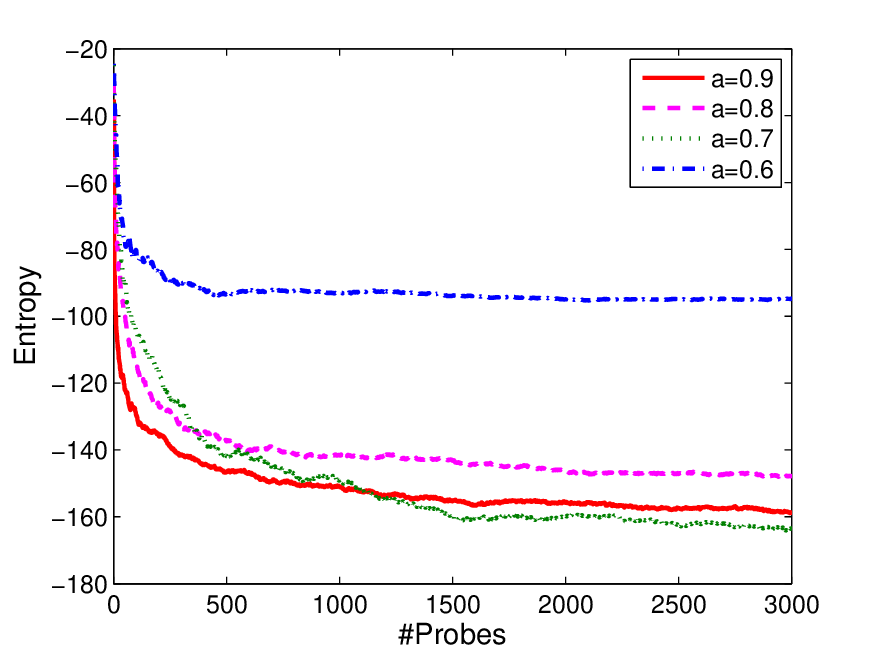}}
\subfigure[two sources: nodes 1 and 9]
{\includegraphics[scale=0.45]{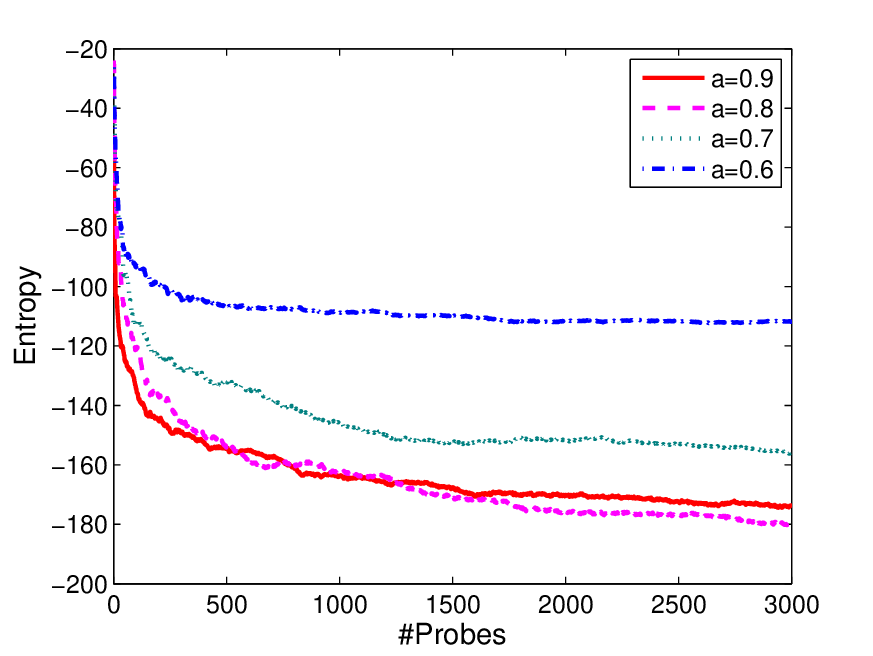}}
\vspace{-0.25em}
\caption{Abilene topology with the same $\alpha$ on all links. 
} \label{fig:abilene-ent-samealpha} 
\vspace{-0.75em}
\end{figure}

For the {\em Exodus topology} (48 nodes, 105 links), we select 5 sources, apply the orientation algorithm, and get three receivers. Fig.~\ref{fig:triplettes} shows the distribution of the number of paths for all triplets $(S_i,R_j,e_{R_j})$. There are 16 incoming edges to all three receivers, 44 triplets $(S_i,R_j,e_{R_j})$ and 377 paths from the sources to the receivers in total; this leads to an average of 9 paths and a maximum of 25 paths per triplet $(S_i,R_j,e_{R_j})$. We visit all nodes in a random order and we assign coefficients from a finite field with increasing size ($2^{10}-2^{18}$).

In Fig.~\ref{fig:coeffs-identifiability}, we show the probability of success in terms of path identifiability for five such triplets $(S_i,R_j,e_{R_j})$, with 7, 9, 13, 20 and 25 number of paths, respectively. The values are averaged over 5 different runs for each field size value. When we use random code selection over a field of size $2^{16}$ or larger, we get  good results: for a field of size $2^{18}$  or larger, we get almost 100\% success for all triplets. These are good results for a large realistic topology such as Exodus, since almost 100\% success is achieved with much less bits than the 1500 bytes of an IP packet. Random assignment of coefficients over a set of prime numbers leads to success probability above 98\% when we use up to prime 907 and field size $2^{18}$ for the linear operations. 

\subsubsection{\label{sec:general-bp-results}Results on Belief-Propagation (BP) Inference}
This Section presents results on the quality of the BP
estimation for different  assignments of loss rates to the links of the two considered topologies.

In Fig.~\ref{fig:sc2-mse-alllinks}, we consider the
Abilene topology with loss rates inversely proportional to the
bandwidth of the actual links; the intuition for this assignement is that links with high bandwidth are less likely to be congested.  We see that the estimation error for
each link ($MSE$) and for all links ($ENT$) decreases quickly. In
Fig.~\ref{fig:abilene-ent-samealpha}, the same topology is considered,
but with the same $\alpha$ on all links: again $ENT$ decreases with
the number of probes; as expected, the larger the $\overline \alpha$, the
slower the convergence; there is not a big difference between having
one or two sources in this case. Fig.~\ref{fig:exodus-ent} shows the
estimation error $ENT$ for the Exodus topology with uniform loss
rates. Finally, Table~\ref{table1-results} shows the results for
different numbers and placements of sources in the (modified) Abilene
topology. Unlike Fig.~\ref{fig:abilene-ent-samealpha}, Table~\ref{table1-results} shows that the choice of sources matters and that increasing the number of sources helps in decreasing the $ENT$.

\subsubsection{NC-Tomography vs. Multicast Tomography}

We finally compare the network-coding approach to traditional multicast tomography for general topologies \cite{general}. In the traditional approach, multiple multicast trees are used to cover the general topology, and the estimates from different trees are combined into one, using approaches in \cite{general}.

\begin{figure}[t!]
\centering
{\includegraphics[scale=0.45]{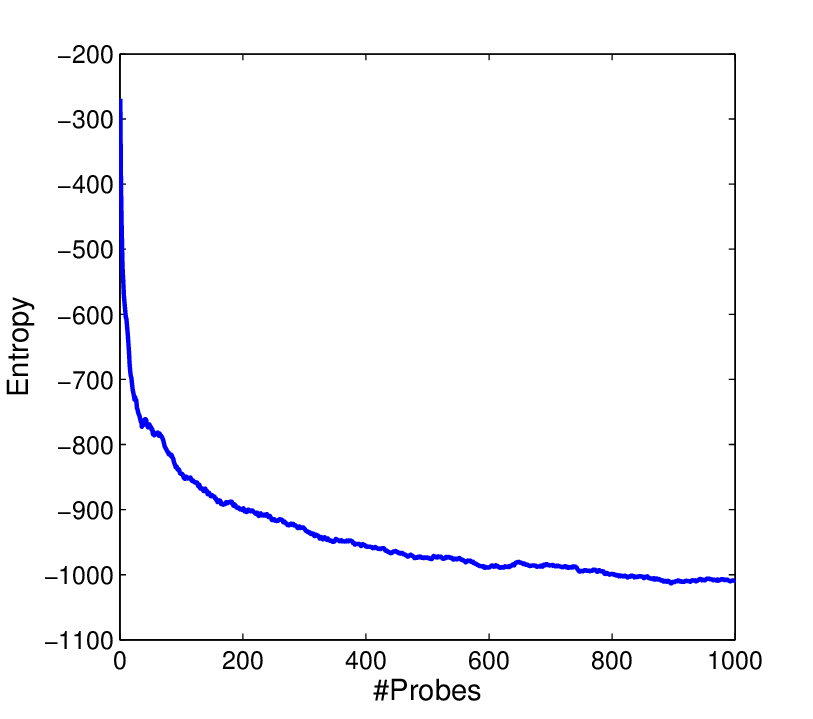}}
\vspace{-1.0em}
\caption{Exodus topology, considering different loss rates across
links: uniformly in $[1\%, 35\%]$. \label{fig:exodus-ent}}
\vspace{-0.5em}
\end{figure}

\begin{table}[t]
\scriptsize
\centering
\begin{tabular}[t]{|l||l|l|l|l|l|l|}
\hline
                       &      \multicolumn{6}{c|}{Entropy for loss rate same over all links } \\
{Srcs}-{Rcvs} &         $\overline \alpha$=0.05 &  $\overline \alpha$=0.1 &  $\overline \alpha$=0.15  &  $\overline \alpha$=0.2   & $\overline \alpha$=0.25   &  $\overline \alpha$=0.3  \\
\hline \hline
 \{1\}-\{9\}           &         -178.6  & -158.8 & -147.9 & -147.7  & -161.6  & -163.5  \\
\hline
\{5\}-\{6\}           &         -178.1  & -158.3 & -149.6 & -154.5  & -160.4  & -156.5  \\
\hline
 \{9\}-\{2\}            &         -176.1  & -163.3 & -155.8 & -161.2  & -166.6  & -151.7 \\
\hline
 \{1,9\}-\{7\}          &         -189.3  & -173.9 & -166.5 & -180.3  & -171.7  & -156.2 \\
\hline
 \{3,6\}-\{9\}          &         -186.2  & -176.2 & -171.3 & -177.8  & -166.7  & -151.4 \\
\hline
 \{9,6\}-\{4\}          &         -186.9  & -174.1 & -169.5 & -178.7  & -173.2  & -165.4 \\
\hline
\{1,5,9\}-\{7\}       &          -199.8  & -190.6 & -180.9 & -184.4  & -172.3  & -166.9  \\
\hline
 \{1,4,10\}-\{9\}       &         -186.4  & -183.9 & -178.3 & -182.3  & -177.3  & -173.2  \\
\hline
\end{tabular}
\vspace{0.25em}
\caption{Quality of Estimation for the (modified) Abilene topology and for different choices of source(s).} \label{table1-results}
\vspace{-1.25em}
\end{table} 

Fig. \ref{fig:comparison-to-traditional}(a) shows the topology we used in the comparison, which is taken from \cite{general}: Nodes $\{0,1,2,5\}$ are sources, nodes $\{12,...19\}$ are receivers, and all remaining nodes (shown as boxes) are intermediate nodes. When the traditional approach is used, probes are sent from each of the four sources to all receivers using a multicast tree, an estimate is computed from every tree, and then, the four estimates are combined into one using the minimum variance weighted average \cite{general}. When the network coding approach is used, the same four sources and the same receivers are used, but probes are combined at intermediate nodes $\{6,7\}$.
For a fair comparison, the same belief-propagation algorithm has been used for estimation over multicast trees and using the network coding approach.
 Fig.~\ref{fig:comparison-to-traditional}(b) shows the performance of both schemes. We see that the network coding approach achieves a better error vs. number of probes tradeoff. The main benefit in this case comes from the fact that the network coding approach eliminates the overlap of the multicast trees below nodes 6 and 7.
 
 \begin{figure}[t!]
\centering
\subfigure[A simulation topology from \cite{general}. Nodes $\{0,1,2,5\}$ are sources, nodes
$\{12,...19\}$ are receivers, and all remaining nodes (shown as boxes) are intermediate nodes.]
{\includegraphics[scale=0.3, angle=-90]{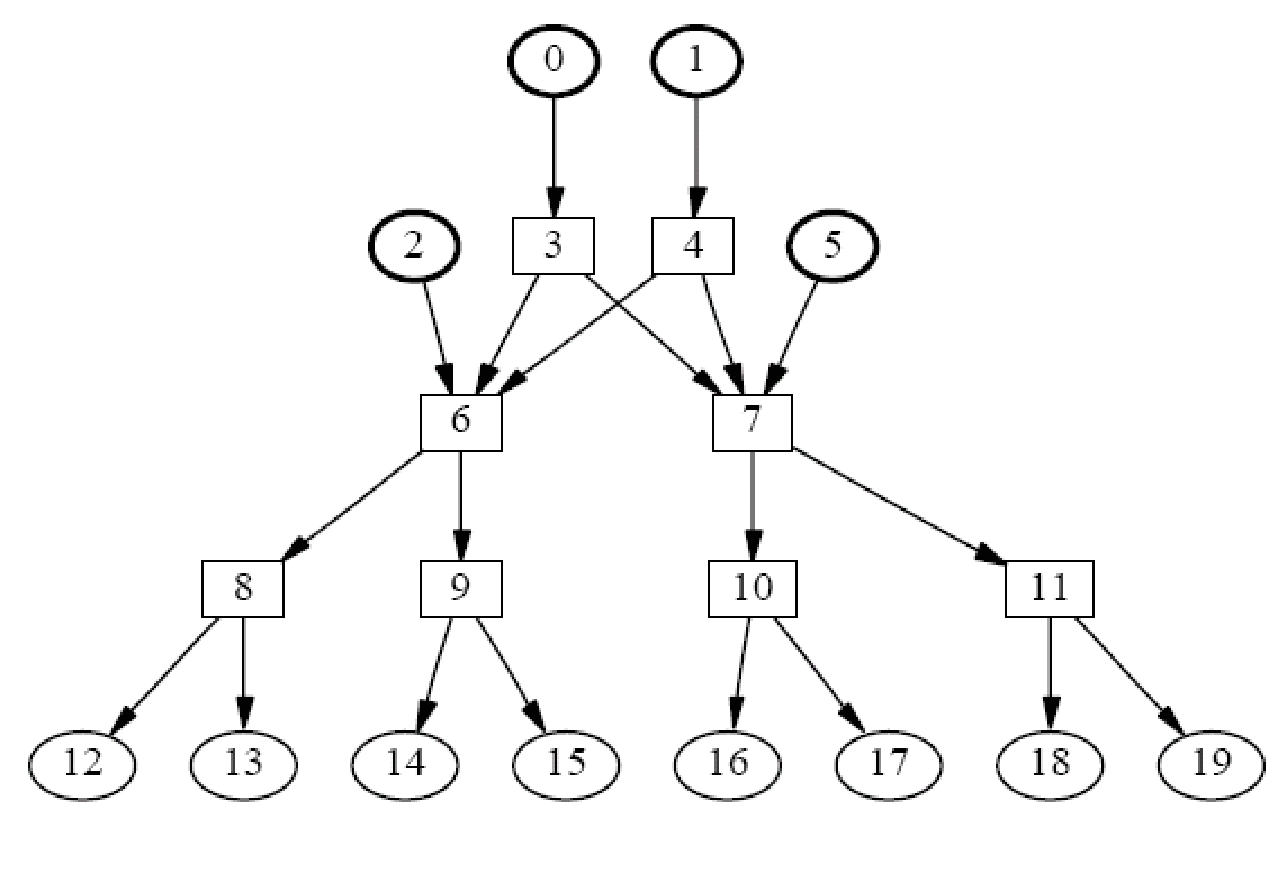}}
\subfigure[Performance of tomography: error ($ENT$) vs. number of probes. Solid and dashed lines correspond to the network coding approach and the traditional approach, respectively. All links have loss rate $\overline \alpha=0.04$.]
{\includegraphics[scale=0.27]{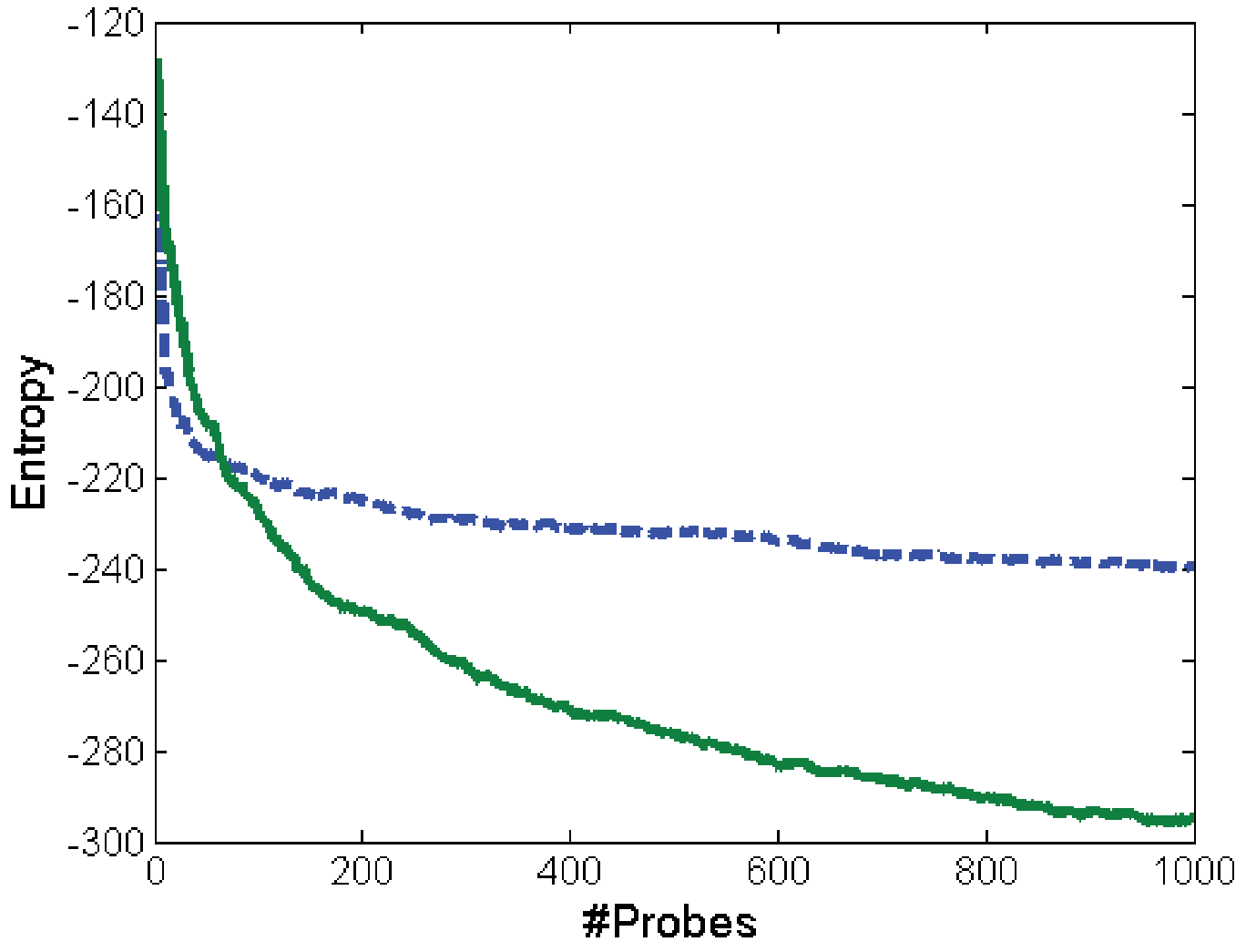}}
\vspace{-0.3em}
\caption{Comparison of network coding approach to traditional tomography. In both cases, the same sources and receivers are used. In the traditional case, four multicast trees are used and the estimates are combined using methods from \cite{general}. In the network coding case, probes are combined wherever they meet in the network (nodes 6 and 7). 
\label{fig:comparison-to-traditional}}
\vspace{-0.7em}
\end{figure}
 
There is of course a wealth of other tomographic techniques that are not simulated here. (For example, we could cover a general graph with unicast probes, but this would perform worse than using multicast probes.) The reason is that \cite{general} is directly comparable to our approach and thus highlights the intuitive benefits of network coding, everything else being equal. Network coding ideas could also be developed for and combined with other tomographic approaches. 

\begin{figure*}[t!]\centering
\begin{center}
\psset{unit=0.045in}
\begin{pspicture}(5,40)(105,100)
\psset{linewidth=0.5mm}
\begin{small}

\rput(15,103){ \large{\em (a) Real Topology}}
\rput(5,95){\circlenode{A}{A}}
\rput(25,95){\circlenode{B}{B}}
\rput(15,77){\circlenode{C}{C}}
\rput(15,63){\circlenode{D}{D}}
\rput(5,50){\circlenode{E}{E}} 
\rput(25,50){\circlenode{F}{F}}
\ncline[linewidth=0.5mm,linestyle=dashed, linecolor=blue]{->}{A}{C}\Bput{$x_1$}
\ncline[linewidth=0.5mm,linestyle=dashed, linecolor=red]{->}{B}{C}\Aput{$x_2$}
\ncline[linewidth=0.5mm,linecolor=purple]{->}{C}{D}\Aput{$x_1+x_2$}
\ncline[linewidth=0.5mm,linestyle=dashed, linecolor=purple]{->}{D}{E}\Bput{$x_1+x_2$}
\ncline[linewidth=0.5mm,linestyle=dashed, linecolor=purple]{->}{D}{F}\Aput{$x_1+x_2$}

\rput(55,103){\large{\em (b) Multicast Tree}} 
\rput(55,95){\circlenode{B2}{}}
\rput(55,63){\circlenode{D2}{D}}
\rput(45,50){\circlenode{E2}{E}}
\rput(65,50){\circlenode{F2}{F}}
\ncline[linewidth=1mm,linecolor=black]{->}{B2}{D2}\Bput{$x_1$ and/or $x_2$}
\ncline[linewidth=0.5mm,linecolor=purple]{->}{D2}{E2}
\ncline[linewidth=0.5mm,linecolor=purple]{->}{D2}{F2}

\rput[Br]{90}(63,95){$\underbrace{\text{~``aggregate''~link~ABCD~~~~}}$}

\rput(95,103){\large{\em (c) Reverse Multicast Tree}}
\rput(85,95){\circlenode{A3}{A}} 
\rput(105,95){\circlenode{B3}{B}}
\rput(95,77){\circlenode{C3}{C}}
\rput(95,50){\circlenode{F3}{}}
\ncline[linewidth=0.5mm, linestyle=dashed, linecolor=blue]{->}{A3}{C3}\Bput{$x_1$}
\ncline[linewidth=0.5mm, linestyle=dashed, linecolor=red]{->}{B3}{C3}\Aput{$x_2$}
\ncline[linewidth=1mm,linecolor=purple]{->}{C3}{F3}\Bput{$x_1$, $x_2$, or $x_1+x_2$}

\rput[Br]{90}(103,79){$\underbrace{\text{~``aggregate''~link~CDEF~~~~}}$}

\end{small}
\end{pspicture}
\end{center}
\vspace{-2.0em}
\caption{\label{fig:reductions} {\bf Reductions.} {\bf (a)} depicts the real topology based on conditions 1(b) and 2(b). The goal is to identify the loss rate of link $CD$. $A, B$ are sources and $E,F$ are receivers. $AC, BC, DE, DF$ can be either links or paths from/to the sources/receivers. In {\bf (b)}, we reduce the real topology to a multicast tree with three links: ``aggregate'' link $ABCD$ (which transmits {\em some} symbol, $x_1, x_2$ or $x_1\oplus x_2$, below $D$), and links $DE, DF$ (which broadcast that symbol). In {\bf (c)}, we reduce the real topology to a reverse multicast tree with three links: $AC, BC$ and ``aggregate'' link $CDEF$ (which transmits the symbol coming in $CD$ to at least one receiver). As shown in detail in Table \ref{tabl_1}, the observations in the reduced topologies are simply unions of disjoint observations in the original topology, and their probabilities are the sum of the probabilities of the corresponding observations in the original topology.}
\end{figure*}
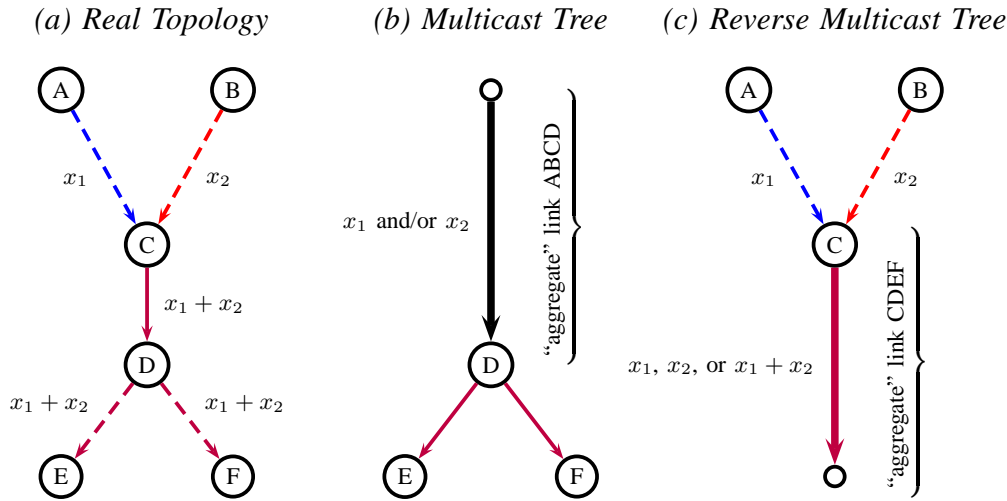

\section{\label{sec:conclusion}Conclusion}

In this paper, we revisited the well-studied and hard problem of link loss tomography using new techniques in networks equipped with network coding capabilities. We developed a novel framework for estimating the loss rates of some or all links in this setting. We considered trees and general topologies. We showed that network coding capabilities can improve virtually all aspects of loss tomography, including identifiability, routing complexity, and the tradeoff between estimation accuracy and bandwidth overhead.

\appendices
\section{Proofs of Theorems}

\subsection*{Appendix A.1: Proof of Theorem~\ref{theorem_1}}

\begin{proof}
To prove that conditions 1 and 2 are necessary, consider that 
condition 1 is not satisfied. Then $C$ can only receive one stream of
probe packets, since it is connected to only one source. There
exists an edge $e$ through which this stream of probe packets
arrives at node $C$. The link success rate associated with link $CD$
cannot be distinguished from the link success rate associated with
link $e$. More formally, if $\alpha_e$ is the success probability associated with link $e$
and $\alpha_{CD}$ is the success probability associated with link $CD$, then the variables $\alpha_e$ and $\alpha_{CD}$ appear always together ({\em e.g.,} in the expression $1-\alpha_e\alpha_{CD}$ in the probability function $P_{\alpha}$). Therefore there are many pairs of values $(\alpha_e, \alpha_{CD})$ that lead to the same $P_{\alpha}$. According to definition \ref{def_ident}, this means that link $CD$ is not identifiable.
 Similar arguments hold for the other conditions and this completes the forward argument.

\smallskip
Next, we prove that conditions 1 and 2 are sufficient for identifying link $CD$.

 First, let us consider Case 1, where Conditions 1(b) and 2(b) are satisfied. The remaining cases are similar and are discussed at the end of this proof. These conditions mean that the paths involving link $CD$ should be as depicted in Fig. \ref{fig:reductions}(a): $AC, BC, DE, DF$ can be either links or paths from/to the sources/receivers, respectively. In the latter case (when $AC,BC$ and $DE, DF$ depict paths), the path success probability can be computed from the success rates of the corresponding links. Essentially, Case 1 (also shown in Fig. \ref{fig_basic} -- 5-links, Case 1) generalizes the motivating example of Section \ref{sec:singlelink}, where the links $AC, BC, DE, DF$ are replaced by paths $AC, BC, DE, DF$ with the same success probability.

In Definition \ref{def_ident}, and consistently with \cite{minc}, we defined the links as identifiable iff the probability distribution $P_\mathbf{\alpha}$ uniquely determines the parameters $\mathbf{\alpha}$\footnote{Recall that $\alpha$ refers to the vector of all success probabilities, and $\alpha_e$ refers to the success probability of one particular edge $e$.}, {\em i.e.,} iff  for $\alpha, \alpha' \in (0,1]^{|E|}$, $P_\alpha=P_{\alpha'}$ implies $\alpha=\alpha'$. To establish the identifiability of link $CD$, we repeatedly apply the identifiability result for a 3-link multicast tree (from \cite{minc})  and for a reverse multicast tree (leveraging the reversibility property in Theorem \ref{th_reverse}, Section \ref{sec-MINCandRMINC}). Consider the two reductions of the actual 5-link topology (as described in Section~\ref{sec-reductions}), to a multicast tree (MT) shown in Fig.~\ref{fig:reductions}(b), and to a reverse multicast tree (RMT) shown in Fig.~\ref{fig:reductions}(c), respectively.

In case of the 3-link multicast tree consisting of $ABCD$ and $DE, DF$, Theorems 2 and 3 in \cite{minc} guarantee that  $\alpha_{DE}$, $\alpha_{DF}$, and $\alpha_{ABCD}$ are identifiable. Namely, $P_{\alpha'}^m= P_{\alpha}^m$ implies ${\alpha'}^m={\alpha}^m$.

On the other hand, since the MLE for the reverse multicast tree has the same functional form as the  multicast tree (as described in Section \ref{sec-MINCandRMINC}), using again the main result of \cite{minc}, we have that $P_{\alpha'}^{r}= P_{\alpha}^{r}$ implies ${\alpha'}^{r}={\alpha}^{r}$.

{\em Proving identifiability in the original topology, via contradiction}.
Consider the 5-link tree in Fig. \ref{fig:reductions}(a), and assume that there exist $\alpha, \alpha' \in (0,1]^{|E|}$ for which $P_{\alpha}=P_{\alpha'}$ and $\alpha \neq \alpha'$.

Use the multicast tree reduction to map the success rates $\alpha$  to  $\alpha^m$ and associated probabilities  $P_\alpha$ to  $P_\alpha^m$.
Similarly,  reduce the success rates $\alpha'$ to  $\alpha'^m$, and associated probabilities $P_{\alpha'}$ to  $P_{\alpha'}^m$. Since $P_{\alpha}=P_{\alpha'}$, we conclude that  $P_{\alpha}^m=P_{\alpha'}^m$. Because the topology in  Fig. \ref{fig:reductions}(b) is identifiable \cite{minc}, we conclude that $\alpha^m=\alpha'^m$. This implies that:
\begin{equation}
\label{eq11}
{\alpha'}_{DE}={\alpha'}_{DE}^m=\alpha_{DE}^m=\alpha_{DE}
\end{equation}
\begin{equation}
\label{eq12}
{\alpha'}_{DF}={\alpha'}_{DF}^m=\alpha_{DF}^m=\alpha_{DF}
\end{equation}
\begin{equation}
\label{eq13}
\begin{split}
(1-{\overline{\alpha'}}_{AC} {\overline{\alpha'}}_{BC}) {\alpha'}_{CD} & = {\alpha}_{ABCD}'^m \\
& =\alpha_{ABCD}^m=(1-\overline\alpha_{AC} \overline\alpha_{BC}) \alpha_{CD} \\
\end{split}
\end{equation}
Applying similar arguments for the reduction to a reverse multicast tree, we get that $\mathbf{\alpha}^{r}=\mathbf{\alpha'}^r$, and as a result:
\begin{equation}
\label{eq21}
{\alpha'}_{AC}={\alpha'}_{AC}^r=\alpha_{AC}^r=\alpha_{AC}
\end{equation}
\begin{equation}
\label{eq22}
{\alpha'_{BC}}={\alpha'}_{BC}^r=\alpha_{BC}^r=\alpha_{BC}
\end{equation}
\begin{equation}
\label{eq23}
\begin{split}
(1-{\overline{\alpha'}}_{DE} {\overline{\alpha'}}_{DF}) {\alpha'}_{CD} & ={\alpha}_{CDEF}'^r \\
& =\alpha_{CDEF}^r=(1-\overline\alpha_{DE} \overline\alpha_{DF}) \alpha_{CD} \\
\end{split}
\end{equation}
From Equations (\ref{eq11})-(\ref{eq23}), we conclude that $\alpha=\alpha'$, which is a contradiction. Therefore, $P_{\alpha}=P_{\alpha'}$ implies that $\alpha=\alpha'$, {\em i.e.,} identifiability. 

The remaining cases (combinations of clauses (a), (b), (c) in Conditions 1 and 2, other than 1(b) and 2(b)) are shown in Fig. \ref{fig_basic}. For example, Condition 1(a) or 2(a) corresponds to the 3-link multicast or reverse multicast tree, and the MINC MLE can then be used directly on these trees. Conditions 1(c) or 2(c) lead to the Cases 2-4 in Fig. \ref{fig_basic}, and similar reductions as in Case 1 can be used to prove identifiability. This completes the proof. 
\end{proof}

\subsection*{Appendix A.2: Estimating $\alpha_{CD}$}

\begin{proof}
Let us denote the outcomes in which link $CD$ has worked by $x_{CD}$; the outcomes in which at least one of the upstream paths to $C$ has worked by $x_{up}$; and the outcomes in which at least one of the downstream paths after $D$ has worked by $x_{dn}$. For the intersection of any two of these outcomes, \eg $x_{up}$ and $x_{dn}$, we use the notation $x_{up,dn}$. The independence of link loss rates indicates that $x_{up}$, $x_{dn}$, and $x_{CD}$ are independent. Therefore: 
\begin{equation}
\label{eq-alphaCDproof}
\hat{\alpha}_{CD} = \hat{p}(x_{CD})=  \hat{p}(x_{CD}|x_{up,dn}) = \frac{\hat{p}(x_{CD} \, \& \, x_{up,dn})}{\hat{p}(x_{up})\hat{p}(x_{dn})}
\end{equation}
The numerator equals $1-\hat{p}([0,0,\cdots,0])=\hat{\gamma}^r_C = \hat{\gamma}^m_D$. Also we have that:
\begin{equation}
\begin{split}
\hat{p}(x_{dn}) & =\hat{p}(x_{dn}|x_{up,CD})=\hat{p}(x_{dn}|X_D\neq [0,...,0])\\
& =1-\hat{p}(x^c_{dn}|X_D\neq [0,...,0])=1-\prod_{j=1}^{Q} \overline \beta^m_{d(D)_j}
\end{split}
\end{equation}
We can derive a similar expression for $\hat{p}(x_{up})$. Therefore:
\begin{equation}
\label{eq-alphaCDproof-2}
\hat{\alpha}_{CD}=\frac{1-\hat{p}([0,0,\cdots,0])}{(1-\prod_{i=1}^{P} \overline \beta^r_{f(C)_i})(1-\prod_{j=1}^{Q} \overline \beta^m_{d(D)_j})} 
 \end{equation}
By writing Eq.(\ref{eq-MCbeta1}) for $\overline \beta^m_D$ in Fig.~\ref{fig-generalMCreduction}, and by writing Eq.(\ref{eq-RMCbeta1}) for $\overline \beta^r_C$ in Fig.~\ref{fig-generalRMCreduction}, we conclude that:
\begin{equation} 
\label{eq-denominator}
1-\prod_{j=1}^{Q} \overline \beta^m_{d(D)_j} = \frac{\beta^m_D}{\alpha^m_{agg}} = \frac{\gamma^m_D}{A^m_D} \; , \; 1-\prod_{i=1}^{P} \overline \beta^r_{f(C)_i} = \frac{\beta^r_C}{\alpha^r_{agg}} = \frac{\gamma^r_C}{A^r_C}
\end{equation}
Eq.(\ref{eq-alphaCDestimate}) then follows from replacing these results into Eq.(\ref{eq-alphaCDproof-2}).
\end{proof}

\subsection*{Appendix A.3: Proof of Lemma~\ref{thr-bottomLinks}} 

\begin{proof}
In \cite{minc}, it has been shown that the likelihood function of the reduced multicast tree in Fig.~\ref{fig-generalMCreduction}, $\mathcal{L}^m(\alpha^m)$, can be written as the sum of three distinct parts in which the derivative $\partial \log p^m(x^m)/\partial \alpha^m_k$ is constant. These parts are $\Omega^m(k)$, the $\Omega^m(f^i(k)) \backslash \Omega^m(f^{i-1}(k))$, which we represent by $\Omega^m_2$ for simplicity, for $i=1,2,\cdots,l^m(k)$, and ${(\Omega^m(0))}^c$. The derivative in these parts is equal to $\frac{1}{\alpha^m_k}$, $\frac{1}{\overline \beta^m_{f^{i-1}(k)}}\frac{\partial \overline \beta^m_{f^{i-1}(k)}}{\partial \alpha^m_k}$, and $\frac{1}{\overline \beta^m_{0}}\frac{\partial \overline \beta^m_{0}}{\partial \alpha^m_k}$, respectively. Thus, the likelihood equation can be written as: 
\begin{equation}
\label{eq-MClikelihoodDerivative}
\begin{split}
\frac{\partial \mathcal{L}^m}{\partial \alpha^m_k} & =  \frac{1}{\alpha^m_k} \sum_{x^m\in \Omega^m(k)} n^m(x^m) \\
& + \sum_{i=1}^{l^m(k)} \{ \frac{1}{\overline \beta^m_{f^{i-1}(k)}}\frac{\partial \overline \beta^m_{f^{i-1}(k)}}{\partial \alpha^m_k} \sum_{x^m\in \Omega^m_2} n^m(x^m) \} \\
& + \frac{1}{\overline \beta^m_{0}}\frac{\partial \overline \beta^m_{0}}{\partial \alpha^m_k} \sum_{x^m\in {(\Omega^m(0))}^c} n^m(x^m)
\end{split}
\end{equation}
Similarly, we can split the likelihood function of the original tree, $\mathcal{L}(\alpha)$, into three parts in which $\partial \log p(x)/\partial \alpha_k$ is constant. These parts will be similar to those of a multicast tree, only with $\Omega^m(k)$ as defined for the original tree in Section~\ref{sec-MINCandRMINC}, and with $l^m(k)$ representing the number of ancestors of node $k$ up to node $C$ (instead of the root 0 in the multicast tree). The derivative $\partial \log p(x)/\partial \alpha_k$ over these parts is also similar to the multicast tree, \ie $\frac{1}{\alpha_k}$, $\frac{1}{\overline \beta^m_{f^{i-1}(k)}}\frac{\partial \overline \beta^m_{f^{i-1}(k)}}{\partial \alpha_k}$, and $\frac{1}{\overline \beta^m_{C}}\frac{\partial \overline \beta^m_{C}}{\partial \alpha_k}$, respectively. Therefore, we have that: 
\begin{equation}
\label{eq-generalLikelihoodDerivative}
\begin{split}
\frac{\partial \mathcal{L}}{\partial \alpha_k} & =  \frac{1}{\alpha_k} \sum_{x\in \Omega^m(k)} n(x) \\
& + \sum_{i=1}^{l^m(k)} \{ \frac{1}{\overline \beta^m_{f^{i-1}(k)}}\frac{\partial \overline \beta^m_{f^{i-1}(k)}}{\partial \alpha_k} \sum_{x\in \Omega^m_2} n(x) \} \\
& + \frac{1}{\overline \beta^m_{C}}\frac{\partial \overline \beta^m_{C}}{\partial \alpha_k} \sum_{x\in {(\Omega^m(C))}^c} n(x)
\end{split}
\end{equation}
{\em (i) $\hat{\alpha}^m_k$ vs. $\hat{\alpha}_k$, $k<D$.} We first compare the solutions $\hat{\alpha}^m_k$ of Eq.(\ref{eq-MClikelihoodDerivative}) and $\hat{\alpha}_k$ of Eq.(\ref{eq-generalLikelihoodDerivative}) for $k<D$. From Eq.(\ref{eq-nRelation}), we have:
\begin{equation}
\label{eq-equality1}
\sum_{x\in \Omega^m(k)} n(x) = \sum_{x^m\in \Omega^m(k)} n^m(x^m)
\end{equation}
\begin{equation}
\sum_{x\in \Omega^m_2} n(x) = \sum_{x^m\in \Omega^m_2} n^m(x^m)
\end{equation}
\begin{equation}
\label{eq-equality3}
\sum_{x\in {(\Omega^m(C))}^c} n(x) = \sum_{x^m\in {(\Omega^m(0))}^c} n^m(x^m)
\end{equation}
Therefore, for any link $k$ located below node $D$, we have that:
\begin{equation}
\label{eq-MCreduction1}
\frac{\partial \mathcal{L}^m}{\partial \alpha^m_k} = \frac{\partial \mathcal{L}}{\partial \alpha_k} \Longrightarrow \hat{\alpha}^m_k=\hat{\alpha}_k, \quad k<D
\end{equation}
{\em (ii) $\hat{\alpha}_{agg}^m$ vs. $\hat{\alpha}_{CD}$.} For $\alpha^m_{agg}$ and $\alpha_{CD}$, Eq.(\ref{eq-MClikelihoodDerivative}) and Eq.(\ref{eq-generalLikelihoodDerivative}) consist of only the first and the last terms. We have that: 
\begin{equation}
\label{eq-likelihoodDerivative-agg}
\frac{\partial \mathcal{L}^m}{\partial \alpha^m_{agg}} =  \frac{1}{\alpha^m_{agg}} \sum_{\Omega^m(D)} n^m(x^m) + \frac{1}{\overline \beta^m_{0}}\frac{\partial \overline \beta^m_{0}}{\partial \alpha^m_{agg}} \sum_{{(\Omega^m(0))}^c} n^m(x^m)
\end{equation}
\begin{equation}
\label{eq-likelihoodDerivative-CD}
\frac{\partial \mathcal{L}}{\partial \alpha_{CD}} =  \frac{1}{\alpha_{CD}} \sum_{x\in \Omega^m(D)} n(x) + \frac{1}{\overline \beta^m_{C}}\frac{\partial \overline \beta^m_{C}}{\partial \alpha_{CD}} \sum_{x\in {(\Omega^m(C))}^c} n(x)
\end{equation}
Thus, $\frac{\partial \mathcal{L}^m}{\partial \alpha^m_{agg}} \neq \frac{\partial \mathcal{L}}{\partial \alpha_{CD}}$, but the definition of $\overline \beta^m_k$ indicates that:
\begin{equation}
\label{eq-betaBar0}
\overline \beta^m_0 = 1-\alpha^m_{agg}(1-\prod_{j=1}^{Q} \overline \beta^m_{d(D)_j})
\end{equation}
\begin{equation}
\label{eq-betaBarC}
\overline \beta^m_C = 1-(1-\prod_{i=1}^{P} \overline \beta^r_{f(C)_i})\alpha_{CD}(1-\prod_{j=1}^{Q} \overline \beta^m_{d(D)_j})
\end{equation}
From Eq.(\ref{eq-equality1}), Eq.(\ref{eq-equality3}), Eq.(\ref{eq-betaBar0}), and Eq.(\ref{eq-betaBarC}), we find out that the solutions $\hat{\alpha}^m_{agg}$ of Eq.(\ref{eq-likelihoodDerivative-agg}) and $\hat{\alpha}_{CD}$ of Eq.(\ref{eq-likelihoodDerivative-CD}) are related via:
\begin{equation}
\label{eq-MCreduction2}
\hat{\alpha}_{CD} = \frac{\hat{\alpha}^m_{agg}}{1-\prod_{i=1}^{P} \overline \beta^r_{f(C)_i}}
\end{equation}
\end{proof}

{\em Note:} The proof of Lemma~\ref{thr-topLinks} is similar to the proof of Lemma~\ref{thr-bottomLinks} above. 

\begin{table}[t!]
\scriptsize
\centering
\begin{tabular}{|c|c|c|c|c|c|c|c|}
\hline
\multicolumn{3}{|c|}{Received at} & \multicolumn{5}{c|}{Is link ok?}  \\
\hline
B & E & F &    AC & BC & CD & DE & DF   \\
\hline
- & -  & -   & \multicolumn{5}{c|}{Multiple possible events}  \\
\hline
- & - & $x$              & 1 & 0 & 1 & 0 & 1  \\
\hline
- & $x$ & -              & 1 & 0 & 1 & 1 & 0  \\
\hline
-  & $x$ & $x$           & 1 & 0 & 1 & 1 & 1  \\
\hline
$x$ & - & -              & 1 & 1 & 0 & * & * \\
\hline
$x$ & - & -              & 1 & 1 & 1 & 0 & 0  \\
\hline
$x$ & - & $x$            & 1 & 1 & 1 & 0 & 1  \\
\hline
$x$ & $x$ & -            & 1 & 1 & 1 & 1 & 0  \\
\hline
$x$ & $x$ & $x$          & 1 & 1 & 1 & 1 & 1  \\
\hline
\end{tabular}
\vspace{0.5em}
\caption{Case 2\label{conf_2}}
\vspace{-1.5em}
\end{table}

\begin{table}[t]
\scriptsize
\centering
\begin{tabular}{|c|c|c|c|c|c|c|}
\hline
\multicolumn{2}{|c|}{Received at} & \multicolumn{5}{c|}{Is link ok?}  \\
\hline
B & F & AC & BC & CD & DE & DF   \\
\hline
- & -     & \multicolumn{5}{c|}{Multiple possible events}  \\
\hline
- & $x_1$                & 1 & 0 & 1 & 0 & 1  \\
\hline
- & $x_2$                & 1 & 0 & 0 & 1 & 1  \\
\hline
- & $x_2$                & 0 & * & * & 1 & 1  \\
\hline
- & $x_1 \oplus x_2$     & 1 & 0 & 1 & 1 & 1  \\
\hline
$x_1$  & -               & 1 & 1 & 0 & 0 & 1  \\
\hline
$x_1$  & -               & 1 & 1 & * & * & 0  \\
\hline
$x_1$ & $x_1$            & 1 & 1 & 1 & 0 & 1  \\
\hline
$x_2$ & $x_2$            & 1 & 1 & 0 & 1 & 1  \\
\hline
$x_1$ & $x_1 \oplus x_2$ & 1 & 1 & 1 & 1 & 1  \\
\hline
\end{tabular}
\vspace{0.5em}
\caption{Case 3\label{conf_3}}
\vspace{-1.5em}
\end{table}

\begin{table}[t!]
\scriptsize
\centering
\begin{tabular}{|c|c|c|c|c|c|}
\hline
Received at & \multicolumn{5}{c|}{Is link ok?}  \\
\hline
F & AC & BC & CD & DE & DF   \\
\hline
-      & \multicolumn{5}{c|}{Multiple possible events}  \\
\hline
$x_1$                & 1 & 0 & 1 & 0 & 1  \\
\hline
$x_2$                & 0 & 1 & 1 & 0 & 1  \\
\hline
$x_3$                & 0 & 0 & 1 & 1 & 1  \\
\hline
$x_3$                & * & * & 0 & 1 & 1  \\
\hline
$x_1 \oplus x_2$     & 1 & 1 & 1 & 0 & 1  \\
\hline
$x_1 \oplus x_3$     & 1 & 0 & 1 & 1 & 1  \\
\hline
$x_2 \oplus x_3$     & 0 & 1 & 1 & 1 & 1  \\
\hline
$x_1 \oplus x_2 \oplus x_3$  & 1 & 1 & 1 & 1 & 1  \\
\hline
\end{tabular}
\vspace{0.5em}
\caption{Case 4\label{conf_4}}
\vspace{-1.5em}
\end{table}

\begin{figure}[t!]
\centering
\subfigure[Estimator vs. number of probes]{\includegraphics[width=2.6in]{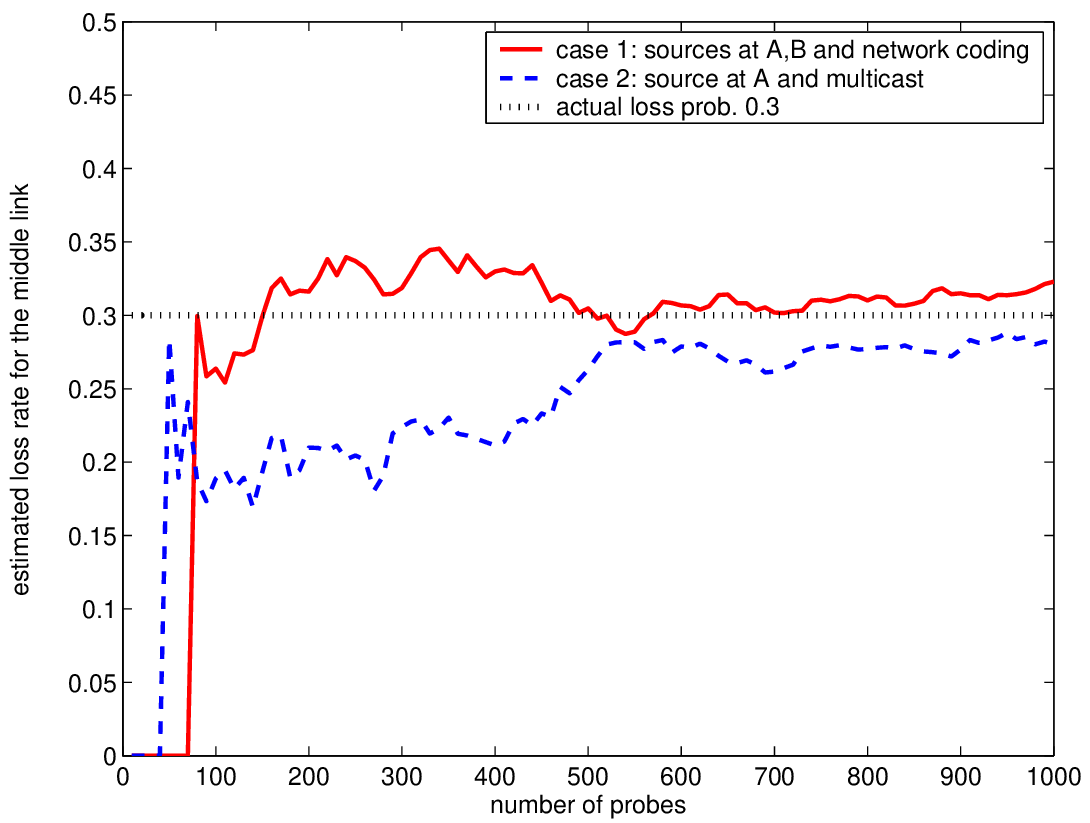}}
\subfigure[Estimation variance vs. number of probes]{\includegraphics[width=2.6in]{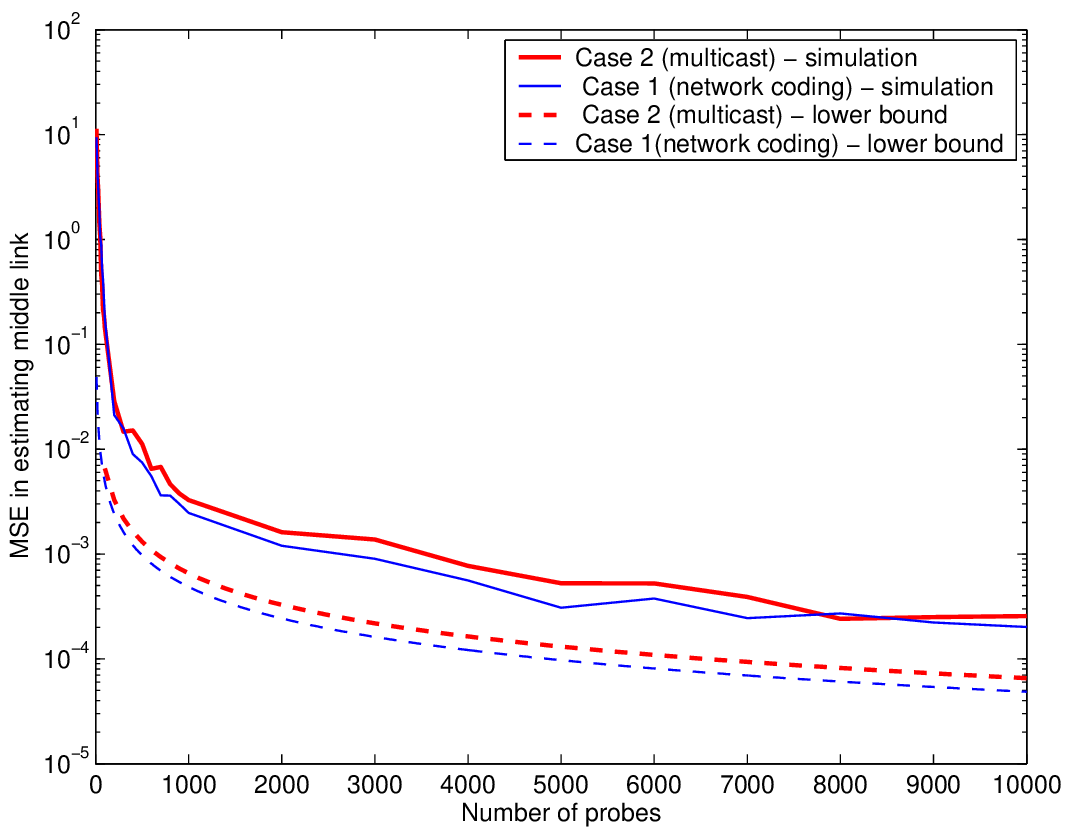}}
\vspace{-0.25em} 
\caption{Convergence of the ML estimator for cases 1, 2} 
\label{fig_exp_conv} 
\vspace{-0.75em}
\end{figure}

\subsection*{Appendix A.4: Proof of Theorem~\ref{thr-MLE}}

 \begin{figure*}[t]
\centering
\subfigure[All links have the same $\overline \alpha$]{
  \includegraphics[width=2.6in]{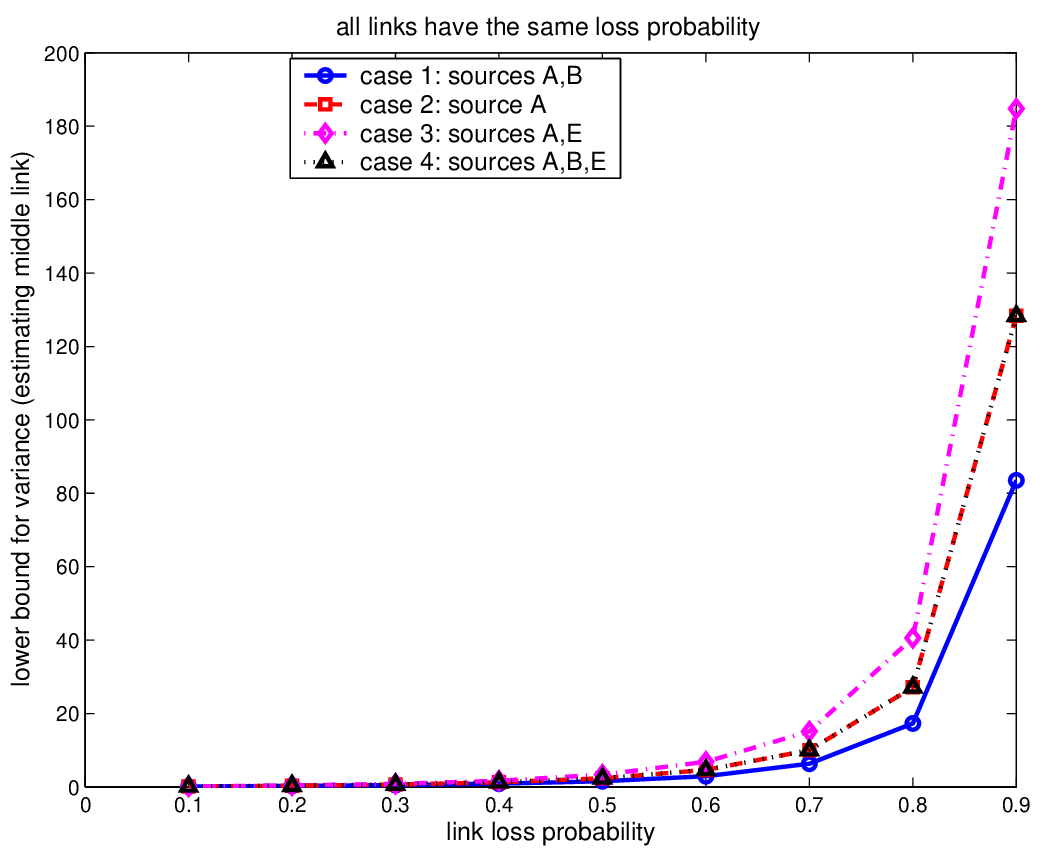}}
  \hspace{0.1\textwidth}
\subfigure[All edge links have the same $\overline \alpha_{edge}=0.5$.]{
  \includegraphics[width=2.6in]{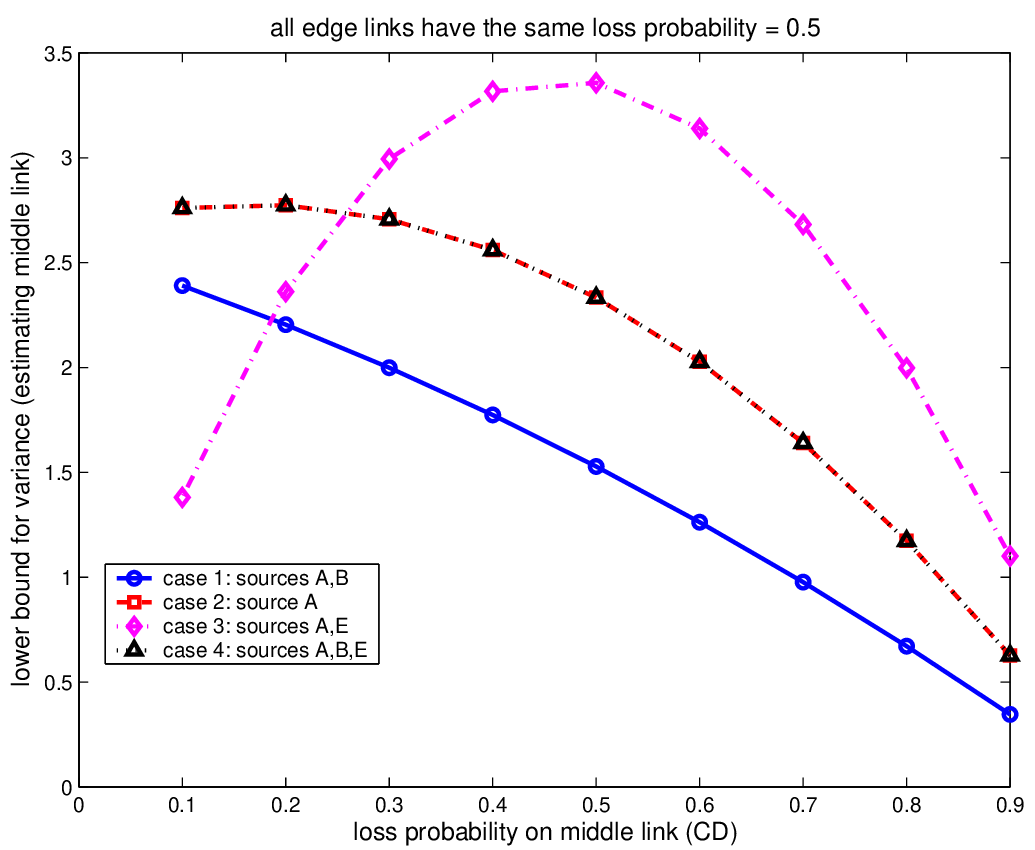}}
\vspace{-0.5em} 
\caption{Comparing the 4 cases in Fig.~\ref{fig_basic} in terms of the lower-bound of variance.} \label{fig_fischer_5}
\end{figure*}

\begin{figure*}[t]
\centering
\subfigure[All edge links have the same $\overline\alpha_{edge}$. Consider
all possible combinations of ($\overline\alpha_{edge},\overline\alpha_{middle}$).]{
  \includegraphics[width=2.6in]{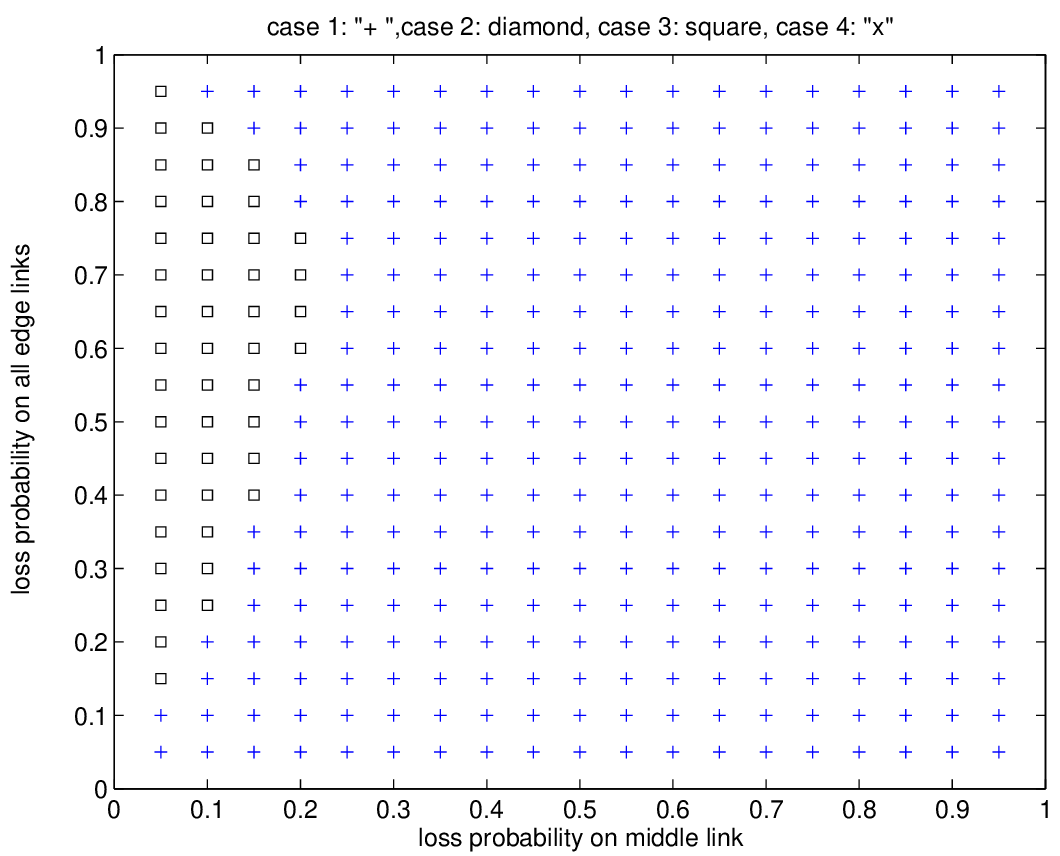}}
  \hspace{0.1\textwidth}
\subfigure[$\overline\alpha_{AC}=\overline\alpha_{BC}=\overline\alpha_{s}$, $\overline\alpha_{DE}=\overline\alpha_{DF}=\overline\alpha_{r}$, $\overline\alpha_{CD}=0.8$. Consider all combinations of ($\overline\alpha_{s},\overline\alpha_{r}$).]{
  \includegraphics[width=2.6in]{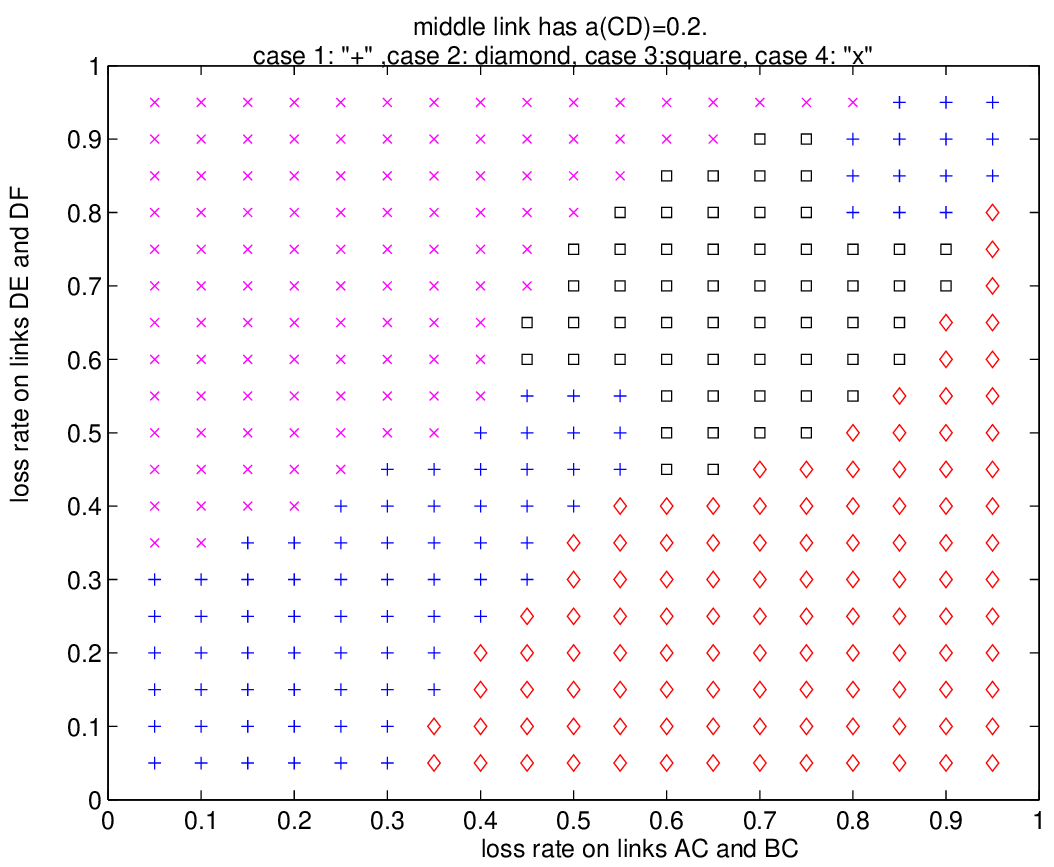}}
\vspace{-0.25em} 
\caption{We indicate which Case (among the four) performs better (has the lowest Cramer-Rao bound), for a given combination of loss rates on all 5 links.} \label{fig_sweeps}
\vspace{-0.25em}
\end{figure*}

\begin{proof}
In \cite{minc}, it has been shown that ${\hat{\alpha}}_{k}^m$ in Eq.(\ref{eq-MC2}) are the MLE of the multicast tree. Therefore, ${\hat{\alpha}}_{k}={\hat{\alpha}}_{k}^m$, $k<D$, are also the MLE  of the corresponding links in the original tree. In addition, by following the same approach as in \cite{minc} and due to the reversibility property, one can show that ${\hat{\alpha}}_{k}={\hat{\alpha}}_{k}^r$, $k>C$, are also the MLE  of the corresponding links in the original tree. For $\hat{\alpha}_{CD}$,  since $\hat{\alpha}^m_{agg}=\hat{A}^m_D$ and using Eq.(\ref{eq-denominator}), one can obtain Eq.(\ref{eq-alphaCDestimate}) from Eq.(\ref{eq-MCreduction2}). Therefore, Eq.(\ref{eq-alphaCDestimate}) is a solution of $\frac{\partial \mathcal{L}}{\partial \alpha_{CD}} = 0$. Furthermore, from Eq.(\ref{eq-likelihoodDerivative-CD}) and Eq.(\ref{eq-betaBarC}), we have that: 
\begin{equation}
\frac{\partial^2 \mathcal{L}}{\partial \alpha_{CD}^2} =  \frac{-1}{\alpha_{CD}^2} \sum_{x\in \Omega^m(D)} n(x) - \frac{1}{{\overline \beta^m_{C}}^2}(\frac{\partial \overline \beta^m_{C}}{\partial \alpha_{CD}})^2 \sum_{x\in {(\Omega^m(C))}^c} n(x)
\end{equation}
This is always negative. Therefore, $\mathcal{L}$ is concave in $\alpha_{CD}$ and Eq.(\ref{eq-alphaCDestimate}) is the unique solution of the likelihood equation. This solution is also in the desired range $(0,1]$, because from Eq.(\ref{eq-alphaCDproof-2}), we have that: 
\begin{displaymath}
\hat{\alpha}_{CD}>0 \Longleftrightarrow \hat{p}([0,0,\cdots,0])<1
\end{displaymath}
\ie not all packets are lost, which is the default assumption in tomography: no inference can be made without data. Also: 
\begin{displaymath}
\hat{\alpha}_{CD}<1 \Longleftrightarrow 1-\hat{p}([0,...0]) < 
(1-\prod_{i=1}^{P} \overline \beta^r_{f(C)_i})(1-\prod_{j=1}^{Q} \overline \beta^m_{d(D)_j})
\end{displaymath}
This is asymptotically true for $\alpha_{CD}>0$, because as $n\rightarrow \infty$, the percentage of packets that are {\em not} lost approaches the probability $(1-\prod_{i=1}^{P} \overline \beta^r_{f(C)_i})\alpha_{CD}(1-\prod_{j=1}^{Q} \overline \beta^m_{d(D)_j})$, which is $<(1-\prod_{i=1}^{P} \overline \beta^r_{f(C)_i})(1-\prod_{j=1}^{Q} \overline \beta^m_{d(D)_j})$. Therefore, Eq.(\ref{eq-alphaCDestimate}) is the MLE of $\alpha_{CD}$ in the original tree. 
\end{proof}

We now provide additional details and simulation results on the effect of the number and location of sources.

\section{The effect of the number and location of sources}

\subsection*{Appendix B.1: Various Configurations for the 5-link Topology}
Let us consider again the four cases shown in Fig.~\ref{fig_basic} for the basic 5-link topology.
The first case, also shown in Fig. \ref{figure_1}, has been discussed in length in Table \ref{tabl_1} and in Section \ref{sec:singlelink}. The corresponding tables used for estimation in Cases 2, 3 and 4 of Fig. \ref{fig_basic} are shown for completeness in Tables \ref{conf_2}, \ref{conf_3} and \ref{conf_4}.

\subsection*{Appendix B.2: Simulation Results for the 5-link Topology}

 Consider again the basic 5-link topology of Fig. \ref{fig_basic} and focus on estimating the middle link CD. Here we show that, even though with network coding links are identifiable for all four cases, the estimation accuracy differs.

In Fig.~\ref{fig_exp_conv}, we assume that all $5$ links have
$\overline\alpha=0.3$ and we look at the convergence of the MLE vs. number of
probes for {\em Case 1} (using network coding) and for {\em Case 2} (multicast
probes with source $A$). Fig.~\ref{fig_exp_conv}(a) shows the estimated
value (for one loss realization). Both estimators converge to the
true value, with the network coding being only slightly faster in
this scenario.

In Fig. \ref{fig_exp_conv}(b), we plot the mean-squared error of the
MLE  for  {\em Case 1} (using network coding) and for {\em Case 2} (multicast)
across number of probes. For comparison, we have also
plotted the Cramer-Rao bound  for link $CD$, which is consistent
with the simulation results. For this scenario, {\em Case 1} does slightly
better than {\em Case 2}, but not by a significant amount. This motivated
us to exhaustively compare all four cases  in Fig.~\ref{fig_basic},
for all combinations of loss rates on the $5$ links.

Fig.~\ref{fig_fischer_5} plots the Cramer-Rao bound for the four
cases as a function of the link-loss probability on the middle link.
The left plot assumes that $\overline\alpha$ is the same for all five links,
while the right plot looks at the case where the edge links have a fixed loss rate equal to $0.5$. We observe that {\em Case 1} shows to
achieve a lower $MSE$ bound. Interestingly, the curves for {\em Case
2} (multicast) and {\em Case 4} (reverse multicast) coincide. The difference
between the performance of different cases is more evident in the
right plot (Fig.~\ref{fig_fischer_5}(b)).

In Fig.~\ref{fig_sweeps}, we systematically consider possible
combinations of loss rates on the 5 links, and we show which case
estimates better the middle link. In the left figure, we assume that
all edge links have the same loss rate and we observe that for most
combinations of $(\overline\alpha_{middle},\overline\alpha_{edge})$, {\em Case 1} (shown in ``+'') performs better. In the right plot, we assume that the middle
link is fixed at $\overline\alpha_{CD}=0.8$ and that
$\overline\alpha_{AC}=\overline\alpha_{BC}=\overline\alpha_{s}$, $\overline\alpha_{DE}=\overline\alpha_{DF}=\overline\alpha_{r}$.
Considering all combinations ($\overline\alpha_{s},\overline\alpha_{r}$), each one
of the four cases dominates for some scenarios. An interesting
observation is, again, the symmetry between {\em Case 2} (multicast) and
{\em Case 4} (reverse multicast).

\section*{Acknowledgements}
The authors would like to thank Suhas Diggavi and Ramya Srinivasan for interactions on the problem of source selection.


\vspace*{-1.0\baselineskip} 
\begin{biography}[{\includegraphics[width=1in,clip,keepaspectratio]{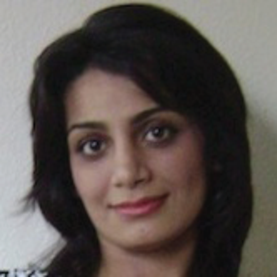}}]{Pegah Sattari}
(S'08) received the B.S. degree in Electrical Engineering from Sharif University of Technology, Tehran, Iran, in 2006, and the M.S. and Ph.D. degrees in Electrical and Computer Engineering from the University of California, Irvine, in 2007 and 2012, respectively. She is currently a senior software engineer at Jeda Networks Inc.. Her research interests include network measurement and analysis, network coding, and network tomography/inference problems. 
\end{biography}

\vspace*{-2.0\baselineskip} 
\begin{biography}[{\includegraphics[width=1in,clip,keepaspectratio]{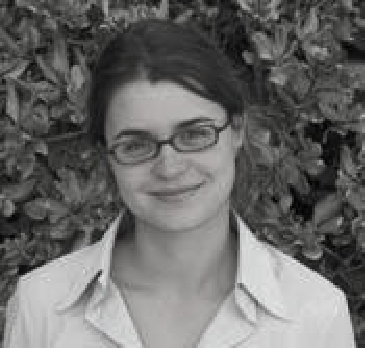}}]{Athina Markopoulou}
(S'98, M'02, SM'12) is an Associate Professor in the EECS Department at the University of California, Irvine. She received the Diploma degree in Electrical and Computer Engineering from the National Technical University of Athens, Greece, in 1996, and the M.S. and Ph.D. degrees in Electrical Engineering from Stanford University in 1998 and 2003, respectively. She has been a postdoctoral fellow at Sprint Labs and at Stanford University, and a member of the technical staff at Arastra Inc.. Her research interests include network coding, network measurement and security, media streaming and online social networks. She received the NSF CAREER award in 2008.
\end{biography}

\vspace*{-2.0\baselineskip} 
\begin{biography}[{\includegraphics[width=1in,clip,keepaspectratio]{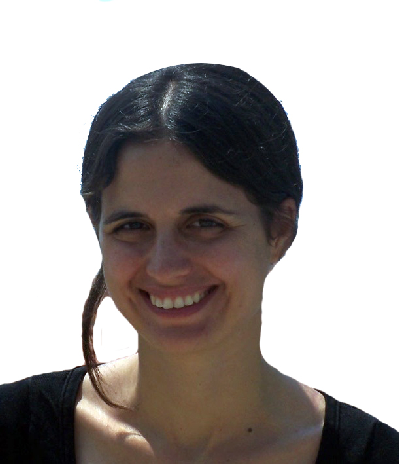}}]{Christina Fragouli}
(M'00) is an Associate Professor in the School of Computer and Communication Sciences, EPFL, Switzerland. She received the B.S. degree in Electrical Engineering from the National Technical University of Athens, Greece, in 1996, and the M.Sc. and Ph.D. degrees in Electrical Engineering from the University of California, Los Angeles, in 1998 and 2000, respectively. She has worked at the Information Sciences Center, AT\&T Labs, and the National University of Athens. She has also visited Bell Labs and DIMACS, Rutgers University. Her research interests include network coding, network information flow theory and algorithms, and connections between communications and computer science. She received the ERC Starting Grant from the European Research Council in 2009.
\end{biography}

\vspace*{-2.0\baselineskip}
\begin{biography}[{\includegraphics[width=1in,clip,keepaspectratio]{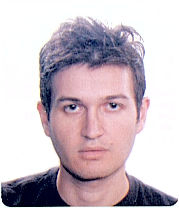}}]{Minas Gjoka} received his B.S. (2005) degree in Computer Science at the Athens University of Economics and Business, Greece, and his M.S. (2008) and Ph.D. (2010) degrees in Networked Systems at the University of California, Irvine. He is currently a postdoc at the University of California, Irvine. His research interests are in the general areas of networking and distributed systems, with emphasis
on online social networks, peer-to-peer systems, and network measurement.
\end{biography}

\end{document}